\documentclass[11pt,a4paper,oneside]{scrartcl}

\usepackage[T1]{fontenc}
\usepackage[utf8]{inputenc}
\usepackage{libertinus}
\usepackage{amsmath,amsthm}
\usepackage[libertinus,slantedGreek,varbb]{newtxmath}
\usepackage[scr=rsfs,frak=euler]{mathalpha}
\DeclareOldFontCommand{\rm}{\normalfont\rmfamily}{\mathrm}
\usepackage{microtype}
\usepackage[a4paper,margin=2.65cm]{geometry}
\usepackage{enumitem}
\usepackage{xcolor}
\usepackage{tikz}
\usetikzlibrary{arrows.meta,calc}
\usepackage{tikz-cd}
\definecolor{phys}{HTML}{0072B2}   
\definecolor{aux}{HTML}{D55E00}    
\definecolor{gsblk}{HTML}{E69F00}  
\definecolor{odblk}{HTML}{CC79A7}  
\definecolor{exblk}{HTML}{56B4E9}  
\colorlet{notecol}{black!62}
\colorlet{framecol}{black!4}
\tikzset{
 stability diagram/.style={x=1cm,y=1cm,font=\small,
   >={Latex[length=4pt,width=3pt]},line width=.5pt,line cap=round,line join=round},
 stability arrow/.style={->,line width=.6pt},
 fig note/.style={font=\footnotesize,text=notecol},
 fig title/.style={font=\small},
 blk/.style={fill=#1,draw=white,line width=.8pt},
}
\tikzcdset{arrow style=tikz,
  diagrams={>={Latex[length=4pt,width=3pt]},line width=.5pt,font=\small},
  panel label/.style={execute at end picture={
   \node[fig note,anchor=base] at ([yshift=-2.3em]
     $(\tikzcdmatrixname-2-1.base)!.5!(\tikzcdmatrixname-2-2.base)$) {#1};}}}
\newcommand\bmSide{1.75}\newcommand\bmGs{.4}
\colorlet{blkzero}{black!3}
\newcommand\blkz{\textcolor{black!45}{$0$}}
\newcommand\blockmat[9]{%
 \begin{scope}[shift={(#1-.5*\bmSide,.5*\bmSide)}]
  \path[blk=#2] (0,0) rectangle (\bmGs,-\bmGs);
  \path[blk=#3] (\bmGs,0) rectangle (\bmSide,-\bmGs);
  \path[blk=#4] (0,-\bmGs) rectangle (\bmGs,-\bmSide);
  \path[blk=#5] (\bmGs,-\bmGs) rectangle (\bmSide,-\bmSide);
  \node[font=\footnotesize] at (.5*\bmGs,-.5*\bmGs) {#6};
  \node[font=\footnotesize] at ({.5*(\bmSide+\bmGs)},-.5*\bmGs) {#7};
  \node[font=\footnotesize] at (.5*\bmGs,{-.5*(\bmSide+\bmGs)}) {#8};
  \node[font=\footnotesize] at ({.5*(\bmSide+\bmGs)},{-.5*(\bmSide+\bmGs)}) {#9};
  \draw[line width=.6pt] (.1,.04) -- (-.04,.04) -- (-.04,-\bmSide-.04) -- (.1,-\bmSide-.04);
  \draw[line width=.6pt] (\bmSide-.1,.04) -- (\bmSide+.04,.04)
                         -- (\bmSide+.04,-\bmSide-.04) -- (\bmSide-.1,-\bmSide-.04);
 \end{scope}}
\usepackage{hyperref}

\hypersetup{
 hidelinks,
 bookmarksnumbered=true,
 pdftitle={Bulk gap stability and spectral flow for fermions with weak long-range interactions},
 pdfauthor={Stefan Teufel}
}

\setkomafont{author}{\sffamily}
\setkomafont{date}{\sffamily}
\setkomafont{sectionentry}{\normalfont\normalsize\sffamily}
\DeclareTOCStyleEntry[beforeskip=0pt,linefill=\TOCLineLeaderFill]{tocline}{section}
\DeclareTOCStyleEntry[entryformat=\normalfont\sffamily,pagenumberformat=\normalfont\sffamily]{tocline}{subsection}
\setlist[itemize]{topsep=4pt,itemsep=2pt,parsep=1pt}
\setlist[enumerate]{topsep=4pt,itemsep=2pt,parsep=1pt}
\allowdisplaybreaks

\newtheorem{theorem}{Theorem}[section]
\newtheorem{proposition}[theorem]{Proposition}
\newtheorem{lemma}[theorem]{Lemma}
\newtheorem{corollary}[theorem]{Corollary}
\theoremstyle{definition}
\newtheorem{definition}[theorem]{Definition}
\newtheorem{hypothesis}[theorem]{Hypothesis}
\theoremstyle{remark}
\newtheorem{remark}[theorem]{Remark}

\let\proofqedsymbol\qedsymbol
\newcommand{\pushStatementQED}{\renewcommand{\qedsymbol}{\ensuremath{\diamond}}\pushQED{\qed}}
\AddToHook{env/hypothesis/begin}{\pushStatementQED}
\AddToHook{env/hypothesis/end}{\popQED}
\AddToHook{env/definition/begin}{\pushStatementQED}
\AddToHook{env/definition/end}{\popQED}
\AddToHook{env/remark/begin}{\pushStatementQED}
\AddToHook{env/remark/end}{\popQED}
\AddToHook{env/proof/begin}{\let\qedsymbol\proofqedsymbol}

\makeatletter
\newcommand{\keepwithproofpointer}{\@endparpenalty=10000\relax}
\makeatother

\newcommand{\A}{\mathcal A}
\newcommand{\Hsp}{\mathfrak h}
\newcommand{\C}{\mathbb C}
\newcommand{\R}{\mathbb R}
\newcommand{\Z}{\mathbb Z}
\newcommand{\1}{\mathbf 1}
\newcommand{\dd}{\mathrm d}
\newcommand{\I}{\mathrm{i}}
\newcommand{\E}{\mathrm{e}}
\newcommand{\D}{\mathrm D}
\newcommand{\dist}{\operatorname{dist}}
\newcommand{\supp}{\operatorname{supp}}
\newcommand{\diam}{\operatorname{diam}}
\newcommand{\spec}{\operatorname{spec}}
\DeclareMathOperator{\ran}{ran}
\newcommand{\ball}{\mathrm B}
\newcommand{\vac}{\mathrm{vac}}
\newcommand{\nf}{\mathrm{nf}}
\newcommand{\at}{\mathrm{pv}}
\newcommand{\st}{\mathrm{st}}
\newcommand{\sgnh}{\mathsf S}
\newcommand{\ph}{\mathrm{ph}}
\newcommand{\loc}{\mathrm{loc}}
\newcommand{\pol}{\mathrm{pol}}
\newcommand{\N}{\mathbb N}
\newcommand{\F}{\mathbb F}
\newcommand{\EE}{\mathbb E}
\newcommand{\cB}{\mathcal B}
\newcommand{\cC}{\mathcal C}
\newcommand{\cG}{\mathcal G}
\newcommand{\cH}{\mathcal H}
\newcommand{\cI}{\mathcal I}
\newcommand{\cL}{\mathcal L}
\newcommand{\cP}{\mathcal P}
\newcommand{\cS}{\mathcal S}
\newcommand{\cT}{\mathcal T}
\newcommand{\cV}{\mathcal V}
\newcommand{\frakb}{\mathfrak b}
\newcommand{\id}{\mathrm{id}}
\newcommand{\CAR}{\operatorname{CAR}}
\newcommand{\dvol}{d_{\rm vol}}
\newcommand{\Ccell}{C_{\rm cell}}
\newcommand{\Dcell}{D_{\rm cell}}
\title{Bulk gap stability and spectral flow\\
for fermions with weak long-range interactions}
\author{Stefan Teufel\thanks{Fachbereich Mathematik, Universit\"at
T\"ubingen, Auf der Morgenstelle 10, 72076 T\"ubingen, Germany.}}
\date{{October 29, 2026}}

\begin{document}

\maketitle

\begin{abstract}
We prove that weak long-range many-body interactions preserve the bulk gap of a free fermion system and that a quasi-local spectral flow transports the free ground state to ground states of the interacting system.
Let \(\Phi_h\) be a free fermionic interaction, possibly including pairing terms, whose infinite-volume Majorana coefficient operator \(h\) has no spectrum in \((-g,g)\), and let \(\cV\) be an even self-adjoint many-body interaction.
For sufficiently small \(|\lambda|\), we construct pure ground states \(\omega_{\lambda}\) of the interacting bulk system \(\Phi_{\lambda}=\Phi_h+\lambda\cV\) and quasi-local automorphisms \(\beta_{0,\lambda}\) such that \(\omega_{\lambda}=\omega_0\circ\beta_{0,\lambda}\), where \(\omega_0\) is the free ground state of \(\Phi_h\).
Each \(\omega_\lambda\) has a GNS Hamiltonian with gap at least \(g-C_*|\lambda|\|\cV\|_*\), where \(\|\cdot\|_*\) is the relevant interaction norm.

Our result allows polynomial decay of both the free coefficients and the many-body interaction.
For two-site interactions on a uniformly discrete point set \(\Gamma\subset\R^d\), a sufficient condition is decay as \((1+|x-y|)^{-\varpi}\) with \(\varpi>2d\), for both the hopping and pairing coefficients and the interaction terms.
The spectral-flow generator preserves the spatial exponent of the polynomial interaction norm, and its automorphisms satisfy locality bounds at that exponent.
If \(h\) and \(\cV\) have stretched-exponential decay, so does the spectral-flow generator.

The proof constructs the spectral flow and establishes the gap through an excitation-number comparison, using polynomial propagation bounds for evolved interactions proved in a companion paper.
\end{abstract}

\clearpage
\tableofcontents
\clearpage

\section{Introduction}
\label{sec:introduction}

For a quadratic fermionic system, a one-particle gap yields a gap above its quasifree many-body ground state.
This includes band insulators, whose ground states are Fermi seas, as well as systems with pairing, whose ground states are Bogoliubov vacua.
We prove stability of this bulk gap under the addition of weak, even, long-range many-body interactions and construct a smooth quasi-local automorphic path connecting the free ground state to interacting ground states.
The difficulty is that the added interactions are extensive: smallness in an interaction norm does not imply a small bounded perturbation of the many-body Hamiltonian.
We obtain both conclusions directly in infinite volume, assuming only a bulk gap for the free fermion system.

Gap stability for weakly interacting fermions has been established by Hastings \cite{Hastings19}, De Roeck and Salmhofer \cite{DRS19}, and Koma \cite{Koma20}, with estimates uniform in the size of the finite system.
The result of De Roeck and Salmhofer also allows polynomial spatial decay.
We build on Hastings' free-fermion reduction and the spectral-flow strategy of Bravyi and Hastings \cite{BH11,Hastings19}, using the infinite-volume approach of Nachtergaele, Sims, and Young \cite{NSY19,NSY22,NSY24}.
The new contribution here is to prove stability of the gap together with quasi-local ground-state transport for weakly interacting fermions directly in infinite volume, assuming only a bulk gap for the free system.
For polynomially decaying free and many-body interactions, the transport is smooth in the coupling, and both its generator and its automorphisms satisfy quantitative locality bounds with the same spatial decay exponent.

Our infinite-volume approach to gap stability and spectral flow has three motivations.
First, the bulk-gap assumption also covers models for which uniformly gapped finite-volume realizations are unavailable.
Second, quantitative spectral flow provides the regular ground-state paths needed in applications such as adiabatic and response theory.
Third, and most importantly for us, the gap and transport results provide the foundation for the companion theorems on ground-state uniqueness and a strong form of the LPPL principle (local perturbations perturb locally) for low energy states proved in~\cite{LPPLV3}.

While band insulators are an important application and provide useful intuition for the proof, we formulate the results more generally for gapped free fermion systems, including Hamiltonians with pairing.
Let \(\Phi_h\) be a free interaction on a countable metric space of uniform polynomial volume growth, with spatially decaying Majorana coefficient operator \(h\) and \(\spec(h)\cap(-g,g)=\varnothing\) for some \(g>0\).
Section~\ref{sec:setting} explains how \(\Phi_h\) is defined in terms of \(h\) using Majorana coordinates.
The free ground state \(\omega_0\) of \(\Phi_h\) is unique, pure, and quasifree.
In the number-conserving case, \(\omega_0\) is the Fermi sea obtained by filling the negative-energy one-particle states.
Adding a self-adjoint, even, spatially decaying interaction \(\mathcal V\) with coupling \(\lambda\) gives the interacting family
\[
\Phi_\lambda:=\Phi_h+\lambda\mathcal V.
\]
For a suitable constant \(C_*\ge1\), independent of \(\mathcal V\), we construct a cocycle \((\beta_{\lambda,\lambda'})\) of quasi-local automorphisms on the coupling interval
\[
I_*:=[-\lambda_*,\lambda_*],
\qquad
\lambda_*:=\frac{g}{4C_*\|\mathcal V\|_*},
\]
such that the states
\[
\omega_\lambda:=\omega_0\circ\beta_{0,\lambda}
\]
are pure gapped ground states of \(\Phi_\lambda\), whose GNS Hamiltonians satisfy
\[
\operatorname{spec}(L_\lambda)\subset\{0\}\cup
\bigl[g-C_*|\lambda|\|\mathcal V\|_*,\infty\bigr).
\]
Here \(L_{\lambda}\ge0\) is the GNS Hamiltonian of the pair \((\omega_{\lambda},\Phi_{\lambda})\), with one-dimensional kernel, and \(\|\cdot\|_*\) is the relevant weighted interaction norm.

We prove this result in two locality regimes.
In the polynomial case, both \(h\) and \(\cV\) have a finite weighted spatial norm at exponent \(p_\Phi>\dvol\).
The norm \eqref{eq:card-norm} for \(\cV\) includes an exponential weight in the support size, which imposes no additional restriction for finite-body interactions.
The metric space must also satisfy the cell-partition condition of Hypothesis~\ref{hyp:cell-partitions}, which holds, for example, for uniformly discrete Euclidean point sets.
For two-site terms satisfying \(\|\cV(\{x,y\})\|\le C(1+d(x,y))^{-\varpi}\) and hopping and pairing coefficients with the same pointwise decay, it suffices that
\[
  \varpi>2\dvol.
\]
The second regime is a short-range result requiring stretched-exponential decay of both \(h\) and \(\cV\).
If \(h\) and \(\cV\) have stretched-exponential decay of order \(0<\theta_{\rm in}\le1\), the spectral-flow generator and its automorphisms have locality of every order \(0<\theta<\theta_{\rm in}\).

Our proof uses a quasi-local Bogoliubov transformation of the doubled free system, within Araki's quasifree CAR framework \cite{Araki70} and with the explicit construction used by Hastings in his gap-stability proof \cite{Hastings19}.
We also follow Nachtergaele, Sims, and Young in restricting perturbations to finite regions within a fixed GNS representation \cite{NSY24}.
Subsection~\ref{sec:proof-roadmap} gives a detailed account of the proof strategy.
The key tools for dealing with slow polynomial decay are the interaction-propagation estimates of \cite{FermionicLR}, which control evolved interactions in weighted norms without loss of the spatial exponent and with polynomial growth in time.

The argument also applies to interacting fermion systems with an even pure product ground state.
We assume that excitations cost at least a fixed positive energy per site excited away from this product state.
Under this assumption, Theorem~\ref{thm:product-reference-stability} gives gap stability and quantitative transport in both decay regimes.
The unperturbed interactions may themselves be strong.
Applications include onsite-pinned occupation models with possibly strong repulsive interactions and arrays of uniformly gapped singlet blocks, such as Fidkowski--Kitaev blocks, Hubbard dimers, and interlayer singlets.
Subsection~\ref{sec:product-references} gives a sufficient condition for this energy assumption and specifies these examples.
For quantum spin systems with product ground states satisfying the same energy assumption, the theorem extends earlier short-range stability results \cite[Example~2]{Yarotsky06} and \cite[Theorem~2.8]{NSY24} to polynomially decaying interactions and provides the quantitative spectral flow.
Examples include Heisenberg ferromagnets in a positive field and arrays of antiferromagnetic dimers (see Subsection~\ref{sec:product-references}).

In \cite{LPPLV3} we use the ground-state transport constructed here to prove uniqueness and strong LPPL for weakly interacting gapped fermions.
Strong LPPL controls the effect of arbitrarily large regional perturbations on observables far from the perturbed region.
Earlier strong LPPL results for spin systems were proved by Henheik, Teufel, and Wessel \cite{HenheikTeufelWessel22} and by Bachmann, De Roeck, Donvil, and Fraas \cite{BachmannDeRoeckDonvilFraas22}.
In \cite{LPPLPT4}, a quantum comparison theorem of Dobrushin type shows that, for weakly interacting gapped fermions, the KMS state is unique at every positive temperature, its correlations decay uniformly in the temperature, and it satisfies strong LPPL; moreover, the gapped ground state is unique.
That paper also proves exponential local convergence of KMS states to the ground state at low temperature.

A bulk gap and a regular ground-state path are standard hypotheses in adiabatic and response theory.
Finite-system formulations \cite{BDF18,MonacoTeufel19,Teufel20} use volume-uniform gaps, while infinite-volume formulations \cite{BTWAd25,HenheikTeufel22,TeufelWesle26,WMMMT25} start from bulk-gapped states.
In infinite volume, pointwise GNS gaps alone do not select a continuous ground-state path.
The transport constructed here provides such a path, with continuously differentiable local expectations.
Bulk gaps and local transport are also used in constructions of Hall conductance and topological invariants \cite{BachmannCorbelliFraasOgata25,KapustinSopenko20, KapustinSopenko22}.
For rapidly decaying interactions, the automorphisms constructed here have the locality required by these results, showing that the interacting invariants equal their free-fermion values (see Remark~\ref{rem:Hall-response-indices}).

The bulk-gap assumption is particularly useful for topologically nontrivial band insulators, where open boundaries may support gapless edge states; see, e.g., \cite{KellendonkRichterSchulzBaldes02}.
We can also work directly with generalized Hofstadter models at irrational magnetic flux \cite{MMMTW26} and with quasicrystals.
For the latter, local-patch methods can establish one-particle gaps even when ordinary finite restrictions suffer from spectral pollution \cite{HMT22,HMT26}.
Such one-particle information thus suffices to obtain an interacting bulk gap and ground-state transport for these applications.

Section~\ref{sec:setting} introduces the mathematical setup, and Section~\ref{sec:results} states the results.
We recommend Subsection~\ref{sec:proof-roadmap} as a first reading, also for experts: it presents the proof's ideas and their logical dependencies before the technical details.
Sections~\ref{sec:atomicization}--\ref{sec:limit} give the proof for exponentially decaying \(h\) and \(\cV\) and deduce the stretched-exponential case by a change of metric.
Section~\ref{sec:card-extension} proves the polynomial result using the propagation estimates of \cite{FermionicLR} and the same normal form, bootstrap, and cutoff limit.
Section~\ref{sec:product-reference-proof} applies these arguments to interacting systems with product ground states.

\section{Setting}
\label{sec:setting}

\subsection{Algebraic and geometric preliminaries}

Let \((\Gamma,d)\) be a countable metric space, and write \(\ball_R(x):=\{y\in\Gamma\,|\,d(x,y)\le R\}\).
We assume uniform polynomial growth: there are \(\dvol\in[0,\infty)\) and \(C_\Gamma\ge1\) such that
\begin{equation}
 |\ball_R(x)|\le C_\Gamma(1+R)^{\dvol},
 \qquad x\in\Gamma,\quad R\ge0.
 \label{eq:volume-growth}
\end{equation}
For \(X\subset\Gamma\) and \(R\ge0\), set
\begin{align*}
 \dist(x,X):=\inf_{y\in X}d(x,y), \quad
  X^{(R)} :=\{y\in\Gamma\,|\,\dist(y,X)\le R\},\quad
 \diam X:=\sup_{x,y\in X}d(x,y).
\end{align*}
For nonempty sets \(X,Y\subset\Gamma\), define \(\dist(X,Y):=\inf_{x\in X}\dist(x,Y)\).
Fix the number of fermionic orbitals per site \(r\in\N\), and put
\begin{equation*}
 \Hsp:=\ell^2(\Gamma;\C^r),
 \qquad
 \A:=\CAR(\Hsp).
\end{equation*}
For \(X\subseteq\Gamma\), put
\begin{equation*}
 \A_X:=\CAR(\ell^2(X;\C^r))\subset\A,
 \qquad
 \A_{\loc}:=\bigcup_{X\Subset\Gamma}\A_X.
\end{equation*}
We write \(X\Subset\Gamma\) when \(X\) is a finite subset of \(\Gamma\).

Let \(a:\Hsp\to\A\) be the canonical antilinear annihilation map.
The parity automorphism \(\Theta\) is determined by \(\Theta(a(f))=-a(f)\) for all \(f\in\Hsp\), and
\begin{equation*}
 \A_X^\pm:=\{A\in\A_X\,|\,\Theta(A)=\pm A\}.
\end{equation*}

An \emph{interaction} is a map \(\Phi:X\mapsto\Phi(X)=\Phi(X)^*\in\A_X^+\) on the finite subsets of \(\Gamma\), where \(\Phi(\varnothing)=0\).
For \(p,k\ge0\), define the polynomial interaction norm
\begin{equation}
 \|\Phi\|^{\pol}_{p;k}
 :=\sup_{x\in\Gamma}
 \sum_{\substack{X\Subset\Gamma\\x\in X}}
 |X|^k(1+\diam X)^p\|\Phi(X)\|
 \label{eq:polynomial-interaction-norm}
\end{equation}
and, for \(\kappa>0\), the norm with an exponential support-size weight
\begin{equation}
 \|\Phi\|^{\#}_{p;\kappa}
 :=\sup_{x\in\Gamma}\sum_{\substack{X\Subset\Gamma\\x\in X}}
 \E^{\kappa|X|}(1+\diam X)^p\|\Phi(X)\|.
 \label{eq:card-norm}
\end{equation}
The norm \( \|\cdot\|^{\#}_{p;\kappa}\) controls \(\|\cdot\|^{\pol}_{p;k}\) for every fixed \(k\).

For \(\mu>0\), \(0<\theta\le1\), and \(k\ge0\), define the stretched-exponential interaction norm
\begin{equation}
 \|\Phi\|_{\mu,\theta;k}
 :=\sup_{x\in\Gamma}
 \sum_{\substack{X\Subset\Gamma\\x\in X}}
 |X|^k\E^{\mu\diam(X)^\theta}\|\Phi(X)\|.
 \label{eq:interaction-norm}
\end{equation}
A bound is called \emph{rooted} when it sums contributions whose supports contain a fixed site, uniformly in that site, as in \eqref{eq:polynomial-interaction-norm}.
The same norms will also be used for local operator families \(X\mapsto B(X)\in\A_X\), without assumptions of parity or self-adjointness.
A path \(s\mapsto\Phi_s\) is called \emph{componentwise norm-continuous} if \(s\mapsto\Phi_s(X)\) is norm-continuous for every finite \(X\subset\Gamma\).
For \(A\in\A_{\loc}\), set
\begin{equation}
 \delta_\Phi(A):=
 \sum_{X\Subset\Gamma}[\Phi(X),A].
 \label{eq:local-derivation}
\end{equation}
Under the decay hypotheses imposed below, \(\delta_\Phi\) is closable, and we write \(\overline\delta_\Phi\) for its closure, the closed derivation generated by \(\Phi\).
Then \(\I\overline\delta_\Phi\) generates the strongly continuous Heisenberg dynamics \((\tau_t^\Phi)_{t\in\R}\), and \(\A_{\loc}\) is a core; see Appendix~\ref{app:nonautonomous-background} and Lemma~\ref{lem:local-form-core}.

For a parameter-dependent interaction \((\eta_u)_{u\in J}\), where \(J\subseteq\R\) is an interval, we order the parameters of its cocycle as \emph{initial point, endpoint}.
Our convention is
\begin{equation}
 \partial_t\beta_{s,t}(A)
 =\I\,\beta_{s,t}(\delta_{\eta_t}(A)),
 \qquad \beta_{s,s}=\id,
 \qquad \beta_{s,u}\circ\beta_{u,t}=\beta_{s,t}\,,
 \label{eq:path-cocycle-convention}
\end{equation}
where the differential equation is understood for local \(A\); existence and uniqueness of solutions under the locality assumptions used here are recalled in Appendix~\ref{app:nonautonomous-background}.
This is the convention of \cite{BTW26} and the companion LPPL paper \cite{LPPLV3}; our \(\beta_{s,t}\) corresponds to \(\tau_{t,s}\) in \cite{NSY18} and \cite{NSY19}.

For a bounded linear map \(\alpha:\A\to\A\), \(\mu>0\), \(0<\theta\leq 1\), and \(q\ge0\), define the (stretched-) exponential quasi-locality norm
\begin{equation}
 \|\alpha\|^{\rm ql}_{\mu,\theta;q}
 :=\|\alpha\|_{\rm op}
 +\sup_{\substack{X\Subset\Gamma,\ 0\ne A\in\A_X\\
                    R\in\N_0}}
 \frac{\E^{\mu R^\theta}}{(1+|X|)^q\|A\|}
 \inf_{B\in\A_{X^{(R)}}}\|\alpha(A)-B\|.
 \label{eq:stretched-ql-norm}
\end{equation}
Appendix~\ref{app:basic-background} gives further geometric estimates and introduces the weighted interaction topologies and graded tensor-product conventions used in the proof.

\begin{definition}[Ground states and GNS gap]
A state \(\omega\) is a ground state of \(\delta_\Phi\) if
\begin{equation}
 \operatorname{Re}\omega(A^*\delta_\Phi(A))\ge0,
 \qquad A\in\A_{\loc}.
 \label{eq:ground-form}
\end{equation}
It has a GNS gap at least \(g'>0\) if it satisfies the ground-state gap inequality
\begin{equation}
 \operatorname{Re}\omega(A^*\delta_\Phi(A))
 \ge g'\left(\omega(A^*A)-|\omega(A)|^2\right),
 \qquad A\in\A_{\loc}.
 \label{eq:GNS-form-gap}
 \qedhere
\end{equation}
\end{definition}

Proposition~\ref{prop:form-spectrum}, together with the local core property of Lemma~\ref{lem:local-form-core}, shows that \eqref{eq:GNS-form-gap} is equivalent to a spectral gap of at least \(g'\) and a one-dimensional ground space of the GNS Hamiltonian.
Both results are proved in Appendix~\ref{app:GNS-background}.

\subsection{Free fermion systems}
\label{sec:quadratic-references}

We use Majorana coordinates to treat number-conserving band Hamiltonians and quadratic Hamiltonians with pairing in one notation.
For \(x\in\Gamma\), let \(\iota_x:\C^r\to\Hsp\) be the canonical isometry.
For \(x\in\Gamma\) and \(1\le j\le r\), set \(a_{x,j}:=a(\iota_xe_j)\) and define
\begin{equation}
 \gamma_{x,2j-1}:=a_{x,j}+a_{x,j}^*,\qquad
 \gamma_{x,2j}:=\I(a_{x,j}^*-a_{x,j}).
 \label{eq:physical-Majoranas}
\end{equation}
For labels \(u=(x,j)\) with \(1\le j\le2r\), these self-adjoint generators satisfy \(\{\gamma_u,\gamma_v\}=2\delta_{uv}\1\).
Their coefficient space is
\[
 \Hsp_{\mathrm M}:=\ell^2(\Gamma;\C^{2r}).
\]
For a bounded operator \(k\) on \(\ell^2(\Gamma;\C^m)\), with \(m\in\N\), let \(k_{xy}\) be its block from the fiber at \(y\) to that at~\(x\).
For \(\mu\ge0\) and \(0<\theta\le1\), define the Schur norm
\begin{equation}
 \|k\|_{\mu,\theta}^{\rm Sch}
 :=\max\left\{
 \sup_{x\in\Gamma}\sum_{y\in\Gamma} \E^{\mu d(x,y)^\theta}\|k_{xy}\|,
 \sup_{y\in\Gamma}\sum_{x\in\Gamma} \E^{\mu d(x,y)^\theta}\|k_{xy}\|
 \right\}.
 \label{eq:weighted-Schur}
\end{equation}
We write \(\|k\|_\mu^{\rm Sch}:=\|k\|_{\mu,1}^{\rm Sch}\) for the exponential norm.
For \(p\ge0\), define also
\begin{equation}
 \|k\|_p^{\mathrm{Sch,pol}}
 :=\max\left\{
 \sup_{x\in\Gamma}\sum_{y\in\Gamma}(1+d(x,y))^p\|k_{xy}\|,
 \sup_{y\in\Gamma}\sum_{x\in\Gamma}(1+d(x,y))^p\|k_{xy}\|
 \right\}.
 \label{eq:polynomial-weighted-Schur}
\end{equation}
The free Hamiltonian is specified by a Majorana coefficient operator \(h\in\cB(\Hsp_{\mathrm M})\) satisfying
\begin{equation}
 h=h^*,\qquad \overline h=-h.
 \label{eq:Majorana-structure}
\end{equation}
Here \(\overline h\) is entrywise complex conjugation in the Majorana coefficient basis.
Thus \(h\) is purely imaginary and \(h^{\mathsf T}=-h\).
Under these conditions, the second-quantization formula \eqref{eq:quadratic-Majorana-interaction} below defines an even self-adjoint quadratic interaction \(\Phi_h\).
Fix \(g>0\), and assume the one-particle gap
\begin{equation}
 \spec(h)\cap(-g,g)=\varnothing
 \label{eq:one-particle-gap}
\end{equation}
and one of the following locality assumptions, specified in the respective theorem.
For the stretched-exponential stability theorem, fix \(0<\theta_{\rm in}\le1\) and \(\mu_h>0\), and assume
\begin{equation}
 C_h:=\|h\|_{\mu_h,\theta_{\rm in}}^{\rm Sch}<\infty.
 \label{eq:h-locality}
\end{equation}
For the polynomial theorem, fix a common spatial exponent \(p_\Phi>\dvol>0\) for \(h\) and \(\cV\), and assume
\begin{equation}
 C_h:=\|h\|_{p_\Phi}^{\mathrm{Sch,pol}}<\infty.
 \label{eq:h-polynomial-locality}
\end{equation}

For \(u=(x,j)\), write \(x_u:=x\) and let \(h_{uv}\) denote the scalar matrix entries of \(h\).
Define the canonical quadratic interaction by summing over ordered pairs \((u,v)\) of Majorana labels.
For \(X\Subset\Gamma\) with \(1\le|X|\le2\), set
\begin{equation}
 \Phi_h(X):=\tfrac14
 \sum_{\substack{(u,v)\\\{x_u,x_v\}=X}}h_{uv}\gamma_u\gamma_v,
 \label{eq:quadratic-Majorana-interaction}
\end{equation}
and set \(\Phi_h(X):=0\) otherwise.
Since \(\overline h_{uv}=-h_{uv}=h_{vu}\), each component is self-adjoint and even.
Summing the finitely many matrix entries in each onsite block gives, for every fixed \(k\ge0\),
\begin{equation}
\begin{aligned}
 \|\Phi_h\|_{\mu_h,\theta_{\rm in};k}&\le 2^{k+1}r^2C_h
       &&\text{under \eqref{eq:h-locality}},\\
 \|\Phi_h\|^{\pol}_{p_\Phi;k}&\le 2^{k+1}r^2C_h
       &&\text{under \eqref{eq:h-polynomial-locality}}.
\end{aligned}
 \label{eq:quadratic-interaction-bounds}
\end{equation}

The CAR identity \([\gamma_v\gamma_w,\gamma_u] =2\delta_{wu}\gamma_v-2\delta_{vu}\gamma_w\) gives, for every label \(u\) and \(t\in\R\),
\begin{equation}
 \delta_{\Phi_h}(\gamma_u)=-\sum_v h_{uv}\gamma_v,
 \qquad
 \tau_t^{\Phi_h}(\gamma_u)
 =\sum_v(\E^{-\I th})_{uv}\gamma_v.
 \label{eq:Majorana-free-dynamics}
\end{equation}
Both series in \eqref{eq:Majorana-free-dynamics} converge absolutely in norm, uniformly for \(t\) in compact intervals, by the Schur bound on \(h\) and \(\|\E^{-\I th}\|_0^{\rm Sch}\le\E^{|t|C_h}\).

For comparison with the usual notation, let \(h_{\mathrm{nc}}=h_{\mathrm{nc}}^*\) and \(\Delta^{\mathsf T}=-\Delta\) be bounded number-conserving and pairing operators on \(\Hsp\), respectively.
Here \(h_{\mathrm{nc}}\) includes the chemical-potential term.
Define the Bogoliubov--de Gennes (BdG) matrix by
\begin{equation}
 h_{\mathrm{BdG}}:=
 \begin{pmatrix}
 h_{\mathrm{nc}}&\Delta\\
 \Delta^*&-\overline{h_{\mathrm{nc}}}
 \end{pmatrix}
 \quad\text{on }\Hsp\oplus\Hsp.
 \label{eq:BdG-matrix}
\end{equation}
The onsite unitary \(T:\Hsp\oplus\Hsp\to\Hsp_{\mathrm M}\) is defined, for \(\psi_1,\psi_2\in\Hsp\), \(x\in\Gamma\), and \(1\le j\le r\), by
\begin{equation}
 (T(\psi_1\oplus\psi_2))_{x,2j-1}:=\frac{(\psi_1)_{x,j}+(\psi_2)_{x,j}}{\sqrt2},
 \qquad
 (T(\psi_1\oplus\psi_2))_{x,2j}:=\frac{-\I(\psi_1)_{x,j}+\I(\psi_2)_{x,j}}{\sqrt2}.
 \label{eq:BdG-Majorana-unitary}
\end{equation}
The corresponding Majorana coefficient operator is \(Th_{\mathrm{BdG}}T^*\).
Conversely, every \(h\) satisfying \eqref{eq:Majorana-structure} has this form, with \(h_{\mathrm{BdG}}=T^*hT\).
The onsite unitary preserves spectra and block Schur norms.
To check the quadratic normalization, fix \(\Lambda\Subset\Gamma\).
Substituting \(h=Th_{\mathrm{BdG}}T^*\) and the definitions of the Majorana generators, and using \(a_{x,j}a_{y,k}^*=\delta_{xy}\delta_{jk}\1-a_{y,k}^*a_{x,j}\), gives
\begin{equation}
 \begin{aligned}
 \sum_{X\subset\Lambda}\Phi_h(X)
 ={}&\sum_{\substack{x,y\in\Lambda\\1\le j,k\le r}}
 ((h_{\mathrm{nc}})_{xy})_{jk}a_{x,j}^*a_{y,k}
 -\tfrac12\sum_{x\in\Lambda}
 \operatorname{Tr}_{\C^r}((h_{\mathrm{nc}})_{xx})\,\1\\
 &+\tfrac12\sum_{\substack{x,y\in\Lambda\\1\le j,k\le r}}
 \left((\Delta_{xy})_{jk}a_{x,j}^*a_{y,k}^*
 +\overline{(\Delta_{xy})_{jk}}\,a_{y,k}a_{x,j}\right)\,.
 \end{aligned}
 \label{eq:quadratic-normal-ordering}
\end{equation}
Setting \(\Delta=0\) gives the number-conserving Hamiltonian up to a scalar, which does not affect the derivation.
For nonzero \(\Delta\), the gap assumption concerns the full BdG matrix and pairing is part of the free Hamiltonian and need not be small.
We use here the quasifree CAR formalism of \cite{Araki70}.
The same Majorana formulation, with a different Hamiltonian normalization, is used in \cite{Hastings19}.

Set
\begin{equation}
 \begin{aligned}
 \sgnh&:=\operatorname{sgn}(h)
 =\1_{(0,\infty)}(h)-\1_{(-\infty,0)}(h),\\
 |h|&:=(h^2)^{1/2}
 =h\,\1_{(0,\infty)}(h)-h\,\1_{(-\infty,0)}(h).
 \end{aligned}
 \label{eq:free-sign}
\end{equation}
Then \(\overline{\sgnh}=-\sgnh\), whereas \(|h|\) is real symmetric.
Let \(\omega_0\) be the even quasifree state determined by
\begin{equation}
 \omega_0(\gamma_u\gamma_v)=(\1+\sgnh)_{uv}.
 \label{eq:free-Majorana-state}
\end{equation}
Its odd correlations vanish, and its even correlations are given by Wick's rule.
Since \(\sgnh=\sgnh^*\), \(\sgnh^2=\1\), and \(\overline{\sgnh}=-\sgnh\), the matrix \((\1+\sgnh)/2\) is a basis projection: it is an orthogonal projection whose conjugate is its complement.
The quasifree construction therefore gives a pure state \cite{Araki70}.
It is a gapped ground state of \(\Phi_h\): for every \(A\in\A_{\loc}\),
\begin{equation}
 \operatorname{Re}\omega_0(A^*\delta_{\Phi_h}(A))
 \ge g\bigl(\omega_0(A^*A)-|\omega_0(A)|^2\bigr).
 \label{eq:free-ground-gap}
\end{equation}
Proposition~\ref{prop:free-GNS-Hamiltonian} proves the ground-state property and the gap using the product-vacuum representation of the stacked system.
By \cite[Lemma~4.3 and Theorem~2]{Araki70}, \(\omega_0\) is the unique ground state of \(\Phi_h\).
In the number-conserving case \(\Delta=0\), \(\omega_0\) is the Fermi sea with one-particle density \(\1_{(-\infty,0)}(h_{\mathrm{nc}})\).

Under \eqref{eq:h-locality}, weighted functional calculus gives, for some \(\mu_1\in(0,\mu_h)\),
\begin{equation}
 \|\sgnh\|_{\mu_1,\theta_{\rm in}}^{\rm Sch}
 +\bigl\||h|\bigr\|_{\mu_1,\theta_{\rm in}}^{\rm Sch}<\infty.
 \label{eq:reference-functional-locality}
\end{equation}
Lemma~\ref{lem:weighted-free-functional-calculus} proves this weighted Schur estimate, with constants controlled by the geometric, one-particle locality, and gap data.
Under \eqref{eq:h-polynomial-locality}, Lemma~\ref{lem:polynomial-free-functional-calculus} gives \(\|\sgnh\|_{p_\Phi}^{\mathrm{Sch,pol}}+\||h|\|_{p_\Phi}^{\mathrm{Sch,pol}}<\infty\) at the same exponent \(p_\Phi\), with bounds controlled by the corresponding polynomial locality and gap data.

Finally, let \(\cV\ne0\) be an interaction and consider the interacting bulk systems
\begin{equation}
 \R\ni \lambda\mapsto\Phi_{\lambda}:=\Phi_h+\lambda\cV.
 \label{eq:physical-path}
\end{equation}

\section{Main results}
\label{sec:results}

In this section we first state the main stability results for weak perturbations of gapped free fermion systems.
Next, in Subsection~\ref{sec:product-references}, we formulate a more general result that also applies to perturbations of certain interacting systems with product ground states.
Since the proofs of the main results require several steps, some of which are quite technical, we explain the overall strategy of the proof in Subsection~\ref{sec:proof-roadmap}.
We end the section with a discussion of and comparison with related literature in Subsection~\ref{sec:comparison-gap-stability}.

\subsection{Gap stability and spectral flow}
\label{sec:stretched-exponential-theorem}

Our main results establish stability of the bulk gap under weak many-body interactions and provide quasi-local automorphisms connecting the free ground state to the interacting ground states.
We first state the stretched-exponential result and then the polynomial result.
In both cases the locality of the spectral flow depends on the decay of \(h\) and \(\cV\).

\begin{theorem}
\label{thm:main}

Fix \(0<\theta_{\rm in}\le1\) and \(\mu_h,\mu_{\cV}>0\), and assume \eqref{eq:one-particle-gap} and \eqref{eq:h-locality}.
There is a constant \(C_*\ge1\), depending only on \(C_\Gamma,\dvol,r,C_h/g,\mu_h,\mu_{\cV},\theta_{\rm in}\).
For each \(0<\theta<\theta_{\rm in}\), there are constants \(0<C_\beta<\infty\), \(\mu_\eta,\mu_\beta>0\), depending only on the same data, \(g\), and \(\theta\), and, for every \(k\ge0\), a constant \(C_{\eta,k}>0\) depending additionally on \(k\).

The following holds for every nonzero self-adjoint even interaction \(\cV\) satisfying \(\|\cV\|_{\mu_{\cV},\theta_{\rm in};2}<\infty\).
On the interval \(I_*:=[-\lambda_*,\lambda_*]\) defined by
\begin{equation*}
 \lambda_*:=\frac{g}{4C_*
 \|\cV\|_{\mu_{\cV},\theta_{\rm in};2}} \,,
\end{equation*}
there is a family \((\omega_{\lambda})_{\lambda\in I_*}\) of pure gapped ground states of \(\Phi_{\lambda}=\Phi_h+\lambda\cV\), independent of \(\theta\), such that for every \(\lambda\in I_*\) and every \(A\in\operatorname{Dom}(\overline\delta_{\Phi_{\lambda}})\)
\begin{equation}
 \operatorname{Re}\omega_{\lambda}\left(A^*\overline\delta_{\Phi_{\lambda}}(A)\right)
 \ge
 \left(g-C_*|\lambda|\|\cV\|_{\mu_{\cV},\theta_{\rm in};2}\right)
 \left(\omega_{\lambda}(A^*A)-|\omega_{\lambda}(A)|^2\right)\,.
 \label{eq:main-gap}
\end{equation}
For each such \(\theta\), there exists a componentwise norm-continuous path \(I_*\ni\lambda\mapsto\eta_\lambda\) of even self-adjoint interactions such that its unique strongly continuous, parity-preserving automorphism cocycle \((\beta_{\lambda,\lambda'})\) transports this family,
\begin{equation}
\forall \lambda\in I_*:\quad\omega_{\lambda}=\omega_0\circ\beta_{0,\lambda}\,,
 \label{eq:main-state-path}
\end{equation}
and satisfies the quantitative locality bounds
\begin{equation}
 \sup_{\lambda\in I_*}\|\eta_{\lambda}\|_{\mu_\eta,\theta;k}
 \le C_{\eta,k}\|\cV\|_{\mu_{\cV},\theta_{\rm in};2}
 \qquad(k\ge0),
 \label{eq:main-generator-locality}
\end{equation}
and
\begin{equation}
 \sup_{\lambda,\lambda'\in I_*}
 \|\beta_{\lambda,\lambda'}\|^{\rm ql}_{\mu_\beta,\theta;1}
 \le C_\beta.
 \label{eq:main-flow-locality}
\end{equation}
\keepwithproofpointer
\end{theorem}
The proof is given in Sections~\ref{sec:atomicization}--\ref{sec:limit}.
\smallskip

We now turn to polynomially decaying interactions.
For the corresponding stability theorem, we require that \(\Gamma\) can be partitioned into cells whose number of neighbors is bounded uniformly in the scale.

\begin{hypothesis}[Partitions into cells]
\label{hyp:cell-partitions}
Assume \eqref{eq:volume-growth} with \(\dvol>0\).
There are constants \(\Ccell\ge1\) and \(\Dcell\in\N_0\) such that, for every \(\ell\ge1\), there exist an index set \(\cB_\ell\) and pairwise disjoint, nonempty finite cells \(Q_b^{(\ell)}\subseteq\Gamma\), \(b\in\cB_\ell\), satisfying
\[
 \Gamma=\bigcup_{b\in\cB_\ell}Q_b^{(\ell)}
 \quad\text{with}\quad
 \diam Q_b^{(\ell)}\le \Ccell\ell
 \quad(b\in\cB_\ell).
\]
On the vertex set \(\cB_\ell\), define a graph by joining distinct \(b,c\in\cB_\ell\) precisely when \(\dist(Q_b^{(\ell)},Q_c^{(\ell)})\le\ell\).
We require this graph to have maximum degree at most \(\Dcell\).
\end{hypothesis}

Note that no compatibility between partitions at different scales is required.
Hypothesis~\ref{hyp:cell-partitions} holds for a broad class of physically relevant geometries, including all Euclidean point sets satisfying \eqref{eq:volume-growth}: intersecting half-open cubes of side length \(\ell\) with \(\Gamma\) gives finite cells with \(\Ccell=\sqrt n\) and \(\Dcell=3^n-1\).
Here \(n\) need not equal \(\dvol\); see \cite[Remark~2.2]{FermionicLR}.
The cell condition also holds if every metric ball can be covered by a uniformly bounded number of balls of half its radius; see \cite[Remark~2.3]{FermionicLR}.

\begin{theorem}
\label{thm:card-stability}
Fix \(p_\Phi>\dvol\), and assume \eqref{eq:one-particle-gap}, \eqref{eq:h-polynomial-locality}, and Hypothesis~\ref{hyp:cell-partitions}.
Choose \(\kappa_0=\kappa_0(\Dcell)>0\) as in \cite[Theorem~2.6]{FermionicLR}.
There is a constant \(C_*\ge1\), depending only on \(p_\Phi,\dvol,C_\Gamma,\Ccell,\Dcell,r,C_h/g\).
There are also constants \(C_\beta,c_{\#},C_{\#}>0\), depending only on these data and \(g\), and, for every \(k\ge0\), a constant \(C_{\eta,k}>0\) depending additionally on~\(k\).

The following holds for every nonzero self-adjoint even interaction \(\cV\) with \(\|\cV\|^{\#}_{p_\Phi;\kappa_0}<\infty\).
On the interval \(I_*:=[-\lambda_*,\lambda_*]\) defined by
\begin{equation}
 \lambda_*:=\frac{g}{4C_*
 \|\cV\|^{\#}_{p_\Phi;\kappa_0}},
 \label{eq:card-lambda-star}
\end{equation}
there is a family \((\omega_{\lambda})_{\lambda\in I_*}\) of pure gapped ground states of \(\Phi_{\lambda}=\Phi_h+\lambda\cV\) satisfying
\begin{equation}
 \operatorname{Re}\omega_{\lambda}
 \left(A^*\overline\delta_{\Phi_{\lambda}}(A)\right)
 \ge\left(g-C_*|\lambda|\|\cV\|^{\#}_{p_\Phi;\kappa_0}\right)
 \left(\omega_{\lambda}(A^*A)-|\omega_{\lambda}(A)|^2\right)
 \label{eq:polynomial-main-gap}
\end{equation}
for every \(\lambda\in I_*\) and \(A\in\operatorname{Dom}(\overline\delta_{\Phi_{\lambda}})\).

There also exists a path \(I_*\ni\lambda\mapsto\eta_\lambda\) of even self-adjoint interactions, \(C^\infty\) in each norm \(\|\cdot\|^{\pol}_{p_\Phi;k}\) for \(k\ge0\), such that its unique strongly continuous, parity-preserving automorphism cocycle \((\beta_{\lambda,\lambda'})\) transports this family,
\begin{equation}
\forall \lambda\in I_*:\quad\omega_{\lambda}=\omega_0\circ\beta_{0,\lambda}\,,
 \label{eq:polynomial-state-transport}
\end{equation}
and satisfies the quantitative locality bounds
\begin{equation}
 \sup_{\lambda\in I_*}\|\eta_{\lambda}\|^{\pol}_{p_\Phi;k}
 \le C_{\eta,k}\|\cV\|^{\#}_{p_\Phi;\kappa_0}
 \qquad(k\ge0),
 \label{eq:card-generator}
\end{equation}
and, for every \(X\Subset\Gamma\), \(A\in\A_X\), and \(R\ge0\),
\begin{equation}
 \sup_{\lambda,\lambda'\in I_*}\inf_{B\in\A_{X^{(R)}}}
 \|\beta_{\lambda,\lambda'}(A)-B\|
 \le C_\beta|X|\|A\|(1+R)^{-p_\Phi}.
 \label{eq:card-flow}
\end{equation}
The generator also satisfies
\begin{equation}
 \sup_{\lambda\in I_*}\sup_{x\in\Gamma}
 \sum_{X\ni x}\E^{c_{\#}|X|^{\theta_{\#}}}(1+\diam X)^{p_\Phi}
       \|\eta_{\lambda}(X)\|
 \le C_{\#}\|\cV\|^{\#}_{p_\Phi;\kappa_0},
 \label{eq:card-generator-support-weight}
\end{equation}
where
\[
 \theta_{\#}:=\frac{p_\Phi-\dvol}
 {p_\Phi-\dvol+2\dvol(p_\Phi+1)}\in(0,1).
\]
For every local \(A\) and fixed \(\lambda_0\in I_*\), the map \(\lambda\mapsto\beta_{\lambda_0,\lambda}(A)\) is norm-\(C^\infty\), and \(\lambda\mapsto\omega_{\lambda}(A)\) is \(C^\infty\).
\keepwithproofpointer
\end{theorem}
The proof follows that of Theorem~\ref{thm:main}, with the polynomial locality estimates proved in \hyperref[proof:card-stability]{Section~\ref*{sec:card-extension}}.

\begin{remark}[Smoothness of the transport]
\label{rem:transport-regularity}
Under Hypothesis~\ref{hyp:cell-partitions}, if \(\cV\) is a finite-body interaction, the transport constructed in Theorem~\ref{thm:main} also has smooth coupling dependence: for every local observable \(A\), the maps \(\lambda\mapsto\beta_{0,\lambda}(A)\) and \(\lambda\mapsto\omega_{\lambda}(A)\) are \(C^\infty\), the former in operator norm, on the same interval \(I_*\).
The proof is given in Section~\ref{sec:limit}, using \cite[Corollary~6.1(iii)]{FermionicLR}.
\end{remark}

For perturbations of bounded body order, say \(|X|\le m\) on the nonzero terms, the support-size condition follows from
\[
 \|\cV\|^{\#}_{p_\Phi;\kappa_0}
 \le\E^{\kappa_0 m}\|\cV\|^{\pol}_{p_\Phi;0}.
\]
In particular, uniformly bounded onsite terms and two-site terms with \(\|\cV(\{x,y\})\|\le C(1+d(x,y))^{-\varpi}\) give \(\|\cV\|^{\#}_{p_\Phi;\kappa_0}<\infty\) whenever \(p_\Phi<\varpi-\dvol\).
If the hopping and pairing blocks also obey a pointwise bound \(\|h_{xy}\|\le C(1+d(x,y))^{-\varpi_h}\), the one-particle norm is finite for \(p_\Phi<\varpi_h-\dvol\).
Thus \(\varpi,\varpi_h>2\dvol\) allow a common exponent \(\dvol<p_\Phi<\min\{\varpi_h,\varpi\}-\dvol\), so the theorem applies under Hypothesis~\ref{hyp:cell-partitions}.

Both theorems construct a quasi-adiabatic spectral flow.
For a suitable filter \(W\) from Subsection~\ref{sec:compatible-filter-pair},
its generator satisfies, for every local observable \(A\),
\begin{equation}
 \delta_{\eta_{\lambda}}(A)
 =\sum_{X\Subset\Gamma}
 \left[\int_\R W(t)\tau_t^{\Phi_{\lambda}}(\cV(X))\,\dd t,A\right],
 \qquad \lambda\in I_*,
 \label{eq:spectral-flow-derivation}
\end{equation}
where the series converges absolutely in norm.
The spectral-flow construction also applies on a compact coupling interval containing zero if the cutoff Hamiltonians \(K_{\Lambda,\lambda}\) defined in \eqref{eq:K-Lambda-s} have simple ground-state eigenvalues separated from the rest of the spectrum by a common positive gap, uniformly in \(\Lambda\) and \(\lambda\).
Choose the filter threshold below this gap.
The locality and cutoff-limit arguments then give quasi-local transport of the limiting ground states throughout that interval.

Write \(\|\cV\|_*\) for the norm of \(\cV\) in the respective theorem.
Let \((\pi_{\lambda},\cH_{\lambda},\Omega_{\lambda})\) be the GNS representation of \(\omega_{\lambda}\), and let \(L_{\lambda}\) be its positive GNS Hamiltonian.
By Proposition~\ref{prop:form-spectrum} and Lemma~\ref{lem:local-form-core}, the gap inequalities \eqref{eq:main-gap} and \eqref{eq:polynomial-main-gap} are equivalently expressed as
\[
 \ker L_{\lambda}=\C\Omega_{\lambda},
 \qquad
 \spec(L_{\lambda})\subseteq\{0\}\cup[g_{\lambda},\infty),
 \qquad g_{\lambda}:=g-C_*|\lambda|\|\cV\|_*.
\]

Applying the constructions with gap parameter one to \(h/g\) and \(\cV/g\), and multiplying the resulting gap inequality by \(g\), shows that \(C_*\) can be chosen to depend on \(C_h/g\), rather than on \(C_h\) and \(g\) separately.
The factor \(g\) in the formula for \(\lambda_*\) therefore expresses the behaviour under a common rescaling of the energy.
Under a common rescaling of \(h\), \(\cV\), and \(g\), the locality rates \(\mu_\eta,\mu_\beta\) and the automorphism constant \(C_\beta\) can be chosen unchanged, while the generator constants \(C_{\eta,k}\) scale inversely with the energy scale.
As \(g\) tends to zero at fixed \(C_h\), the admissible perturbation strength also depends on the behaviour of \(C_*\) as \(C_h/g\) increases.

The generator bounds \eqref{eq:main-generator-locality} and \eqref{eq:card-generator} also quantify closeness to the identity.
Write \(C_\eta:=C_{\eta,2}\) in the stretched-exponential case and \(C_\eta:=C_{\eta,1}\) in the polynomial case.
For \(A\in\A_X\), rooted summation gives \(\|\delta_{\eta_u}(A)\|\le2C_\eta|X|\|A\|\|\cV\|_*\).
Integrating \eqref{eq:path-cocycle-convention} therefore yields
\begin{equation}
 \|\beta_{\lambda,\lambda'}(A)-A\|
 \le2C_\eta|\lambda'-\lambda|\,|X|\,\|A\|\,\|\cV\|_*.
 \label{eq:transport-close-to-identity}
\end{equation}
Since \(\omega_{\lambda'}=\omega_{\lambda}\circ\beta_{\lambda,\lambda'}\), the same right-hand side bounds \(|\omega_{\lambda'}(A)-\omega_{\lambda}(A)|\).
For local \(A\), the cocycle equation \eqref{eq:path-cocycle-convention} also gives
\[
 \frac{\dd}{\dd\lambda}\beta_{0,\lambda}(A)
 =\I\beta_{0,\lambda}(\delta_{\eta_\lambda}(A)).
\]
Componentwise continuity and the uniform interaction bounds make the right-hand side norm-continuous, so \(\lambda\mapsto\beta_{0,\lambda}(A)\) and \(\lambda\mapsto\omega_\lambda(A)\) are continuously differentiable.

\begin{remark}[Decay threshold and propagation bounds]
\label{rem:card-transport-threshold}
The condition \(p_\Phi>\dvol\) in Theorem~\ref{thm:card-stability} is required by the interaction-propagation theorem \cite[Theorem~2.6]{FermionicLR} and by the excitation-number comparison: the bound \eqref{eq:common-vacuum-root-bound} contains the sum \(\sum_{n\ge0}2^{n(\dvol-p_\Phi)}\).
For two-site perturbations, the sufficient norm condition \(p_\Phi<\varpi-\dvol\) gives the threshold \(\varpi>2\dvol\).
This threshold also appears in the algebraic light-cone bounds of Foss-Feig, Gong, Clark, and Gorshkov~\cite{FossFeigEtAl15}.

In \cite[Theorem~6, Proposition~8, and Theorem~9]{TeufelWessel25} we prove polynomial-time Lieb--Robinson bounds for long-range fermions and construct a polynomially local spectral flow for finite-volume Hamiltonian paths with a uniformly isolated spectral sector.
That construction requires only polynomial support-size moments, but gives a spectral-flow generator with weaker spatial decay than the Hamiltonian.
In our gap proof here, this loss would require faster decay to close the excitation-number comparison.
The lossless interaction-propagation estimates from \cite{FermionicLR} we use here instead give locality of the generator and its automorphisms at the same exponent \(p_\Phi>\dvol\).

For integer \(\dvol\ge1\), the construction in \cite[Proposition~B.1]{FermionicLR}, based on the protocols of Tran et al.\ \cite[Eq.~(3) and Sections~III--IV]{TranStateTransfer21}, gives smooth time-dependent even two-site interactions on \(\Z^{\dvol}\) with uniformly bounded \(\|\cdot\|^{\#}_{p;\kappa_0}\)-norms for \(0<p\le\dvol\), but which transfer information over distance \(R\) in times growing more slowly than every positive power of \(R\).
The restriction \(p>\dvol\) is therefore sharp for general time-dependent propagation.
\end{remark}

\begin{remark}[Uniqueness]
\label{rem:ground-state-uniqueness}
The theorems above construct pure gapped ground-states \(\omega_\lambda\) and cocycles transporting them.
The companion paper \cite{LPPLV3} proves a new criterion called \emph{local excitation detectability} and deduces from it uniqueness of \(\omega_\lambda\) among all algebraic ground states \cite[Corollary~3.2]{LPPLV3}.
Under the assumptions of Theorems~\ref{thm:main} and~\ref{thm:card-stability}, detectability holds on a possibly smaller coupling interval about zero \cite[Proposition~2.6]{LPPLV3}; under those of Theorem~\ref{thm:main}, this interval may depend on the chosen order \(0<\theta<\theta_{\rm in}\).
A positive GNS gap and quasi-local spectral flow from the free ground state do not by themselves imply uniqueness among algebraic ground states.
In the example of \cite[Section~3.6]{LPPLV3}, the unique free ground state remains unchanged and gapped along the path, so the identity cocycle transports it.
At coupling \(\lambda=1\) of that example, a second gapped ground state exists.
\end{remark} 

\begin{remark}[Hall response and topological indices]
\label{rem:Hall-response-indices}
Hall conductance, Hall conductivity, and the topological indices studied in \cite{KapustinSopenko20, WMMMT25,TeufelWesle26, MMMTW26, BachmannShapiroTauber25,MazharShapiro26} are constant in phases defined by locally generated automorphic equivalence, under the symmetry, geometry, and locality assumptions of the respective constructions.
Theorem~\ref{thm:main} provides the locality required by these results, so that, under their remaining hypotheses, the interacting invariants coincide with their free-fermion values.
\end{remark}

\begin{remark}[Symmetry covariance of the transport]
\label{rem:transport-symmetry}
The transport also respects onsite symmetries of the interaction.
More precisely, let \(\zeta\) be a parity-preserving automorphism with \(\zeta(\A_X)=\A_X\), \(\zeta(\Phi_h(X))=\Phi_h(X)\), and \(\zeta(\cV(X))=\cV(X)\) for every finite \(X\).
The generator \(\eta_\lambda\) can be chosen so that
\[
 \zeta(\eta_{\lambda}(Y))=\eta_{\lambda}(Y),
 \qquad \zeta\circ\beta_{\lambda,\lambda'}=\beta_{\lambda,\lambda'}\circ\zeta
 \quad(\lambda,\lambda'\in I_*,\ Y\Subset\Gamma).
\]
In particular, number-conserving free Hamiltonians with gauge-invariant perturbations admit gauge-equivariant transport.

The flow also commutes with a conjugate-linear, multiplicative, adjoint-preserving bijection \(\zeta:\A\to\A\) satisfying the same support, parity, and interaction-invariance conditions.

We prove this at the end of \hyperref[proof:transport-symmetry]{Section~\ref*{sec:card-extension}}.
\keepwithproofpointer
\end{remark}

\subsection{Interacting systems with product ground states}
\label{sec:product-references}

The proofs of Theorems~\ref{thm:main} and~\ref{thm:card-stability} use an onsite product ground state and a GNS Hamiltonian bounded below by a positive multiple of the number of excitations.
For free fermion systems, a quasi-local change of coordinates maps the stacked free ground state to such a product state, and the one-particle gap gives the required lower bound.
The following theorem assumes these two properties directly and gives gap stability and spectral flow, in both locality regimes, also for interacting systems that have them.

Let \(\Phi_0\) be an even self-adjoint interaction with a ground state \(\omega_0\) that is the graded product of even pure onsite states.
Its restriction to each finite region is the graded product of the corresponding onsite states.
Write \(P_x\in\A_{\{x\}}\) for their even rank-one projections.
In the GNS representation \((\pi_0,\cH_0,\Omega_0)\), define
\begin{equation}
 N:=\sum_{x\in\Gamma}\pi_0(\1-P_x)
 \label{eq:product-reference-number}
\end{equation}
by the increasing sum of nonnegative quadratic forms.
Thus \(N\) counts sites departing from the specified product state, rather than physical particles.
Let \(K_0\ge0\) be the GNS Hamiltonian of \((\omega_0,\Phi_0)\), and assume
\begin{equation}
 \operatorname{Dom}(K_0^{1/2})\subseteq\operatorname{Dom}(N^{1/2}),
 \qquad
 \|K_0^{1/2}\psi\|^2\ge g\|N^{1/2}\psi\|^2
 \quad(\psi\in\operatorname{Dom}(K_0^{1/2})).
 \label{eq:product-reference-coercivity}
\end{equation}
Since \(N\ge\1-|\Omega_0\rangle\langle\Omega_0|\), assumption \eqref{eq:product-reference-coercivity} implies that \(K_0\) has a simple ground state eigenvalue zero and a spectral gap at least \(g\); see the beginning of Section~\ref{sec:product-reference-proof}.

\begin{theorem}
\label{thm:product-reference-stability}
Fix \(g,C_0>0\) and the geometric and onsite data \(C_\Gamma,\dvol,r\).
For pairs \((\Phi_0,\omega_0)\) as above satisfying \eqref{eq:product-reference-coercivity}, the following hold.
\begin{enumerate}[label=(\alph*),leftmargin=*]
\item Fix \(0<\theta_{\rm in}\le1\) and \(\mu_0,\mu_{\cV}>0\).
There is \(C_*\ge1\), depending only on \(C_\Gamma,\dvol,r,g,C_0,\mu_0,\mu_{\cV}, \theta_{\rm in}\).
For each \(0<\theta<\theta_{\rm in}\), there are constants \(C_\beta,\mu_\eta,\mu_\beta>0\), depending only on the same data and \(\theta\), and, for every \(k\ge0\), a constant \(C_{\eta,k}>0\) depending additionally on \(k\).
The conclusions of Theorem~\ref{thm:main} hold for every such pair \((\Phi_0,\omega_0)\) with
\[
 \|\Phi_0\|_{\mu_0,\theta_{\rm in};2}\le C_0
\]
and every nonzero even self-adjoint \(\cV\) with \(\|\cV\|_{\mu_{\cV},\theta_{\rm in};2}<\infty\), with \(\Phi_{\lambda}=\Phi_0+\lambda\cV\) and
\[
 I_*:=[-\lambda_*,\lambda_*],\qquad
 \lambda_*:=\frac{g}{4C_*
       \|\cV\|_{\mu_{\cV},\theta_{\rm in};2}}.
\]

\item Assume Hypothesis~\ref{hyp:cell-partitions}, fix \(p_\Phi>\dvol\), and choose \(\kappa_0=\kappa_0(\Dcell)\) as in Theorem~\ref{thm:card-stability}.
There are \(C_*\ge1\) and \(C_\beta,c_{\#},C_{\#}>0\), depending only on \(p_\Phi,\dvol,C_\Gamma,\Ccell,\Dcell,r,g,C_0\).
For every \(k\ge0\), there is also \(C_{\eta,k}>0\), depending only on these data and \(k\).
With these constants, the conclusions of Theorem~\ref{thm:card-stability} hold for every such pair \((\Phi_0,\omega_0)\) with
\[
 \|\Phi_0\|^{\#}_{p_\Phi;\kappa_0}\le C_0
\]
and every nonzero even self-adjoint \(\cV\) with \(\|\cV\|^{\#}_{p_\Phi;\kappa_0}<\infty\), with \(\Phi_{\lambda}=\Phi_0+\lambda\cV\) and
\[
 I_*:=[-\lambda_*,\lambda_*],\qquad
 \lambda_*:=\frac{g}{4C_*
                  \|\cV\|^{\#}_{p_\Phi;\kappa_0}}.
\]
\end{enumerate}
\keepwithproofpointer
\end{theorem}
The proof is given in \hyperref[proof:product-reference-stability]{Section~\ref*{sec:product-reference-proof}}.
\smallskip

In both parts, the transport is the quasi-adiabatic spectral flow generated by time-filtering the perturbation terms with the full interacting dynamics of \(\Phi_{\lambda}=\Phi_0+\lambda\cV\), as in \eqref{eq:spectral-flow-derivation}.

A sufficient condition for \eqref{eq:product-reference-coercivity}, under the decay assumptions of Theorem~\ref{thm:product-reference-stability}, is that each onsite term \(\Phi_0(\{x\})\) has ground state energy \(e_x\), rank-one ground projection \(P_x\), and gap at least \(g\), while every term \(\Phi_0(X)\) with \(|X|\ge2\) is nonnegative and satisfies \(\Phi_0(X)\prod_{x\in X}P_x=0\).
Here \(\omega_0\) is the product state determined by these onsite ground projections.
After shifting each onsite term by its ground state energy \(e_x\), the interaction is frustration-free, with product ground state \(\omega_0\).
For every finite \(\Lambda\), this gives
\[
 \sum_{X\subseteq\Lambda}\Phi_0(X)
       -\sum_{x\in\Lambda}e_x\1
 \ge g\sum_{x\in\Lambda}(\1-P_x).
\]
In the product GNS representation, the operator on the left is nonnegative and annihilates \(\Omega_0\).
For local \(A\), absolute convergence of the local derivation therefore gives
\[
 \omega_0(A^*\delta_{\Phi_0}(A))
 =\lim_{\Lambda\nearrow\Gamma}
 \left\langle\pi_0(A)\Omega_0,
 \pi_0\!\left(\sum_{X\subseteq\Lambda}\Phi_0(X)
       -\sum_{x\in\Lambda}e_x\1\right)
 \pi_0(A)\Omega_0\right\rangle\ge0.
\]
Thus \(\omega_0\) is a ground state, and its positive GNS Hamiltonian \(K_0\) is the limit of these operators on local GNS vectors.
The inequality therefore passes to local vectors, and the form-core property of Lemma~\ref{lem:local-form-core}, together with closedness of the excitation form, gives \eqref{eq:product-reference-coercivity}.
This criterion covers pinned occupation models with nonnegative repulsive terms that annihilate the chosen product state.
It also applies to arrays of isolated blocks with nondegenerate ground states and a uniform gap, after grouping clusters of uniformly bounded size and diameter into sites.
Such blocks include the onsite singlets used in symmetric-mass-generation models \cite{HeEtAl16}, Fidkowski--Kitaev blocks \cite{FidkowskiKitaev10}, and Hubbard dimers or interlayer blocks in parameter regimes with an isolated singlet ground state.

Theorem~\ref{thm:product-reference-stability} also holds, under the same decay and excitation-number hypotheses and with parity qualifications omitted, for quantum spin systems with local algebras \(\A_X:=\cB(\bigotimes_{x\in X}\cH_x)\), where \(\sup_x\dim\cH_x<\infty\).
In this formulation, \(\omega_0\) is the ordinary tensor product of pure onsite states, the interactions are self-adjoint, and \(\sup_x\dim\cH_x\) replaces \(r\) in the constant dependencies.
The proof in Section~\ref{sec:product-reference-proof} includes this extension.

For example, on \(\Gamma=\Z^{\dvol}\) with integer \(\dvol\ge1\), consider the spin-\(1/2\) Heisenberg ferromagnet in a positive field, with formal Hamiltonian
\[
 H_0:=B\sum_{x\in\Gamma}\bigl(\tfrac12\1-S_x^z\bigr)
      +\sum_{\{x,y\}}J_{xy}
       \bigl(\tfrac14\1-\mathbf S_x\cdot\mathbf S_y\bigr),
 \qquad B>0,\quad J_{xy}\ge0.
\]
The second sum is over unordered pairs of distinct sites, and \(\mathbf S_x\) denotes the spin-\(1/2\) operators.
Assume \(J_{xy}\le C(1+d(x,y))^{-\varpi}\) with \(\varpi>2\dvol\).
Each exchange term is a nonnegative multiple of the singlet projection and annihilates the fully polarized product state, whose onsite projections are \(P_x=\tfrac12\1+S_x^z\).
The preceding criterion therefore gives \(K_0\ge BN\), even when the exchange is strong relative to the field.
Theorem~\ref{thm:product-reference-stability}(b) gives stability and polynomially local spectral flow under general weak two-site perturbations with the same decay.

Arrays of identical antiferromagnetic spin-\(1/2\) Heisenberg dimers give another example: grouping each dimer into a site gives a gapped singlet product state, to which weak inter-dimer interactions with the same decay may be added.
Stability for finite-range versions of these models with short-range perturbations is covered by Yarotsky \cite[Example~2]{Yarotsky06} and Nachtergaele, Sims, and Young \cite[Theorem~2.8]{NSY24}.

\subsection{Proof strategy}
\label{sec:proof-roadmap}

The seven steps below explain the common proof strategy of Theorems~\ref{thm:main} and~\ref{thm:card-stability}.
We give enough detail for readers to follow the full argument without having to work through every technical estimate in the later sections.
Write \(\|\cV\|_*\) for the perturbation norm in the respective theorem.

\paragraph{1. Construct the candidate physical spectral flow.}
Here and in the following \emph{physical} refers to the system defined by \(\Phi_\lambda\) on \(\A\).
For a bounded perturbation of \(\Phi_h\), quasi-adiabatic flow transports the free ground state \(\omega_0\) to ground states along any coupling interval about zero on which the ground-state eigenvalue remains simple and its spectral gap exceeds the filter threshold \cite{HastingsWen05,BMNS12}.
In our argument, we fix an odd inverse-Liouvillian filter \(W\) with threshold \(g/2\) and rapid time decay.
The resulting flow is defined for all \(\lambda\in\R\), and its locality estimates do not require a gap for the perturbed system.
We first apply it to bounded perturbations obtained by restricting the physical interaction to a region \(\Lambda\Subset\Gamma\), following the approach of Nachtergaele, Sims, and Young \cite{NSY24}:
\[
 \cV_\Lambda(X):=\1_{\{X\subset\Lambda\}}\cV(X),
 \qquad
 V_\Lambda^\ph:=\sum_{X\subset\Lambda}\cV(X).
\]
The free interaction \(\Phi_h\) is left unchanged and no boundaries or edges are introduced.

Write \(\tau_t^{\ph,\Lambda,\lambda}\) for the dynamics of \(\Phi_h+\lambda\cV_\Lambda\).
The corresponding quasi-adiabatic inverse Liouvillian gives the bounded generator
\[
 \mathscr S_{\Lambda,\lambda}^\ph
 :=\cI_{\Lambda,\lambda}^\ph(V_\Lambda^\ph)
 =\int_\R W(t)\,\tau_t^{\ph,\Lambda,\lambda}(V_\Lambda^\ph)\,\dd t.
\]
It defines an inner cocycle \(\beta_{\Lambda,\lambda,\lambda'}\), with
\[
 \partial_{\lambda}u_{\Lambda,\lambda}
 =-\I\mathscr S_{\Lambda,\lambda}^\ph u_{\Lambda,\lambda},
 \qquad u_{\Lambda,0}=\1,
 \qquad
 \beta_{\Lambda,0,\lambda}(A):=u_{\Lambda,\lambda}^*Au_{\Lambda,\lambda}.
\]
Define the transported state
\[
 \omega_{\Lambda,\lambda}:=\omega_0\circ\beta_{\Lambda,0,\lambda}.
\]
Bounded perturbation theory \cite[Chapters~IV and~VII]{Kato} keeps the ground-state eigenvalue simple and the gap above \(g/2\) for sufficiently small \(|\lambda|\).
The ground-state transport identity of Proposition~\ref{prop:bootstrap-flow} then shows that \(\omega_{\Lambda,\lambda}\) is a ground state of \(\Phi_h+\lambda\cV_\Lambda\).
This interval may depend on \(\Lambda\), since \(\|V_\Lambda^\ph\|\) can grow with \(|\Lambda|\).
It remains to prove that the gap stays above \(g/2\) on an interval independent of \(\Lambda\), and then remove the cutoff.

\paragraph{2. Pass to product-vacuum coordinates and establish the excitation-number bound.}
Topological obstructions can prevent a quasi-local transformation of the original ground state to an onsite product vacuum.
For a number-conserving band Hamiltonian \(h_{\mathrm{nc}}\), the occupied projection \(P:=\1_{(-\infty,0)}(h_{\mathrm{nc}})\) may be topologically nontrivial, obstructing an exponentially localized orthonormal basis of occupied states.
One can remove this obstruction by adjoining a free system with one-particle Hamiltonian \(-h_{\mathrm{nc}}\), whose occupied projection is \(\1-P\).

Hastings' explicit transformation of the doubled free system \cite{Hastings19} implements this idea for general quadratic Hamiltonians, including pairing terms.
It maps the doubled ground state to an onsite product vacuum by a quasi-local change of coordinates, in which the Hamiltonian describes positive-energy excitations above that vacuum.

To apply this construction, let \(\Hsp_{\rm aux}\) be a copy of \(\Hsp\), with auxiliary algebra \(\A_{\rm aux}:=\CAR(\Hsp_{\rm aux})\).
Let \(J_{\rm aux}:\Hsp_{\mathrm M}\to\Hsp_{\mathrm{M,aux}}\) be the canonical unitary identification of their Majorana coefficient spaces.
On this copy, introduce \(h_{\rm aux}:=-J_{\rm aux}hJ_{\rm aux}^*\), and denote its free ground state by \(\omega_{\rm aux}\).
An explicit real orthogonal transformation \(O_\st\) of the Majorana generators induces a parity-preserving Bogoliubov automorphism \(\sigma_\st\) on \(\A^\st:=\A\widehat\otimes\A_{\rm aux}\), where \(\st\) denotes the stacked system.
It sends the stacked free dynamics to a free dynamics with positive excitation matrix \(|h|\).
In these coordinates, \(\omega_\at:=(\omega_0\widehat\otimes\omega_{\rm aux}) \circ\sigma_\st^{-1}\) is an onsite product vacuum state, where the subscript \(\at\) stands for product-vacuum.
Figure~\ref{fig:stable-product-vacuum} depicts this change of coordinates.
We use one-particle functional calculus in Proposition~\ref{prop:exponential-stable-locality} to control the locality of \(\sigma_\st\).

Let \((\pi_\at,\cH_\at,\Omega_\at)\) denote the GNS representation of \(\omega_\at\).
The free excitation Hamiltonian in this representation is
\begin{equation}
 K_0=\dd\Gamma(|h|),
 \qquad g\1\le |h|\le\|h\|\1,
 \qquad\text{and hence}\qquad K_0\ge gN.
 \label{eq:abstract-roadmap-free-domination}
\end{equation}
Here \(\dd\Gamma(|h|)\) denotes fermionic second quantization, and \(N:=\sum_x n_x\) counts excitations above the product vacuum, with \(n_x\) the excitation number at site \(x\).
The bounds on \(|h|\) mean that each excitation carries energy between \(g\) and \(\|h\|\).
They therefore give \(gN\le K_0\le\|h\|N\) as quadratic forms and \(\operatorname{Dom}(K_0^{1/2})=\operatorname{Dom}(N^{1/2})\).
Since \(\ker N=\C\Omega_\at\), the lower bound also shows that \(\Omega_\at\) spans the ground space of \(K_0\).

On this fixed GNS Hilbert space, now set
\[
 V_\Lambda:=\pi_\at\sigma_\st(V_\Lambda^\ph\widehat\otimes\1),
 \qquad K_{\Lambda,\lambda}:=K_0+\lambda V_\Lambda.
\]
Since \(V_\Lambda\) is bounded, all \(K_{\Lambda,\lambda}\) are self-adjoint on \(\operatorname{Dom}(K_0)\).
Let \(I_\Lambda\) be the connected component containing zero of the set of parameters \(\lambda\) for which the lowest eigenvalue of \(K_{\Lambda,\lambda}\) is simple and is separated from the rest of the spectrum by more than the \emph{bootstrap threshold} \(g/2\).
For \(\lambda\in I_\Lambda\), denote this lowest eigenvalue by \(E_{\Lambda,\lambda}\) and its rank-one spectral projection by \(P_{\Lambda,\lambda}\).

We now explain the \emph{excitation-counting argument} that allows us to control the spectral gap uniformly in \(\Lambda\).
In this comparison, we identify local observables with their representatives under \(\pi_\at\).
Suppose that on the form domain of \(K_0\) a Hamiltonian has the normal form
\begin{equation}\label{eq:roadmap-normal-form}
 K_{\lambda}^\nf=K_0+\lambda\sum_j R_{j,\lambda},
\end{equation}
where each \(R_{j,\lambda}=R_{j,\lambda}^*\) is supported in a finite set \(X_j\), but need not be positive.
Let \(P^\vac\) project onto the global product vacuum, and set \(P_X^\vac:=\1_{\{0\}}\!\left(\sum_{x\in X}n_x\right)\), the projection onto states with no excitations in \(X\).
Assume further that each \(R_{j,\lambda}\) annihilates the local vacuum on its support \(X_j\) on both sides,
\begin{equation}
 R_{j,\lambda}P_{X_j}^{\vac}=P_{X_j}^{\vac}R_{j,\lambda}=0,
 \qquad\text{and}\quad
 f_{\lambda}:=\sup_{x\in\Gamma}\sum_{j:\,x\in X_j}\|R_{j,\lambda}\|<\infty.
 \label{eq:abstract-roadmap-rooted-strength}
\end{equation}
The first relation in \eqref{eq:abstract-roadmap-rooted-strength} implies, for every \(j\),
\begin{equation*}
 -\|R_{j,\lambda}\|(\1-P_{X_j}^{\vac})\le R_{j,\lambda}
 \le\|R_{j,\lambda}\|(\1-P_{X_j}^{\vac}),
 \qquad
 \1-P_{X_j}^{\vac}\le N_{X_j}:=\sum_{x\in X_j}n_x\,.
\end{equation*}
The second inequality holds because every vector orthogonal to the local vacuum contains at least one excitation in \(X_j\).

Define \(R_{\lambda}:=\sum_jR_{j,\lambda}\) as a quadratic form on \(\operatorname{Dom}(N^{1/2})\).
Summing the estimates \(|\langle\psi,R_{j,\lambda}\psi\rangle|\le\|R_{j,\lambda}\|\langle\psi,N_{X_j}\psi\rangle\) over \(j\) and regrouping the resulting nonnegative contributions according to the sites \(x\in X_j\) gives, for every \(\psi\in\operatorname{Dom}(N^{1/2})\),
\begin{equation}
 \begin{aligned}
 |\langle\psi,R_{\lambda}\psi\rangle|
  \le \sum_j\|R_{j,\lambda}\|\langle\psi,N_{X_j}\psi\rangle =\sum_{x\in\Gamma}
   \Big(\sum_{j:\,x\in X_j}\|R_{j,\lambda}\|\Big)
   \langle\psi,n_x\psi\rangle
 \le f_{\lambda}\langle\psi,N\psi\rangle.
\end{aligned}
 \label{eq:abstract-roadmap-relative-form-bound}
\end{equation}
The series \(\sum_j\langle\psi,R_{j,\lambda}\psi\rangle\) therefore converges absolutely for every \(\psi\in\operatorname{Dom}(N^{1/2})\), so \(R_\lambda\) defines a quadratic form on this domain.
Using \(K_0\ge gN\), we obtain
\begin{equation}
  K_{\lambda}^\nf=K_0+\lambda R_{\lambda}\ge(g-|\lambda|f_{\lambda})N.
 \label{eq:abstract-roadmap-form-bound}
\end{equation}
The vacuum has energy zero because \(K_0\) and every \(R_{j,\lambda}\) annihilate it.
On the orthogonal complement of the vacuum we have \(N\ge\1\).
If \(|\lambda|f_\lambda<g\), the vacuum is therefore the unique ground state of \(K_\lambda^\nf\) and the gap is at least \(g-|\lambda|f_\lambda\).

Hence, to apply this estimate, we must bring \(K_{\Lambda,\lambda}-E_{\Lambda,\lambda}\) unitarily into the normal form \eqref{eq:roadmap-normal-form}, with local terms satisfying \eqref{eq:abstract-roadmap-rooted-strength} uniformly in \(\Lambda\) and along the relevant coupling interval.

\paragraph{3. Differentiate the centered Hamiltonian to construct the normal form.}
We express the physical spectral flow of Step~1 in the fixed GNS representation by setting
\[
 \mathscr S_{\Lambda,\lambda}
 :=\pi_\at\sigma_\st(\mathscr S_{\Lambda,\lambda}^\ph\widehat\otimes\1),
 \qquad
 U_{\Lambda,\lambda}:=\pi_\at\sigma_\st(u_{\Lambda,\lambda}\widehat\otimes\1),\qquad \sigma_{\Lambda,\lambda}
 :=\sigma_\st\circ
       (\beta_{\Lambda,0,\lambda}\widehat\otimes\id_{\rm aux}).
\]
By \eqref{eq:algebra-GNS-dynamics}, \(\mathscr S_{\Lambda,\lambda}\) is given by the inverse-filter formula \eqref{eq:spectral-flow-generator} for \(K_{\Lambda,\lambda}\) and \(V_\Lambda\).
The definitions above and the ground-state transport explained in Step~1 therefore give
\begin{align}
 \pi_\at(\sigma_{\Lambda,\lambda}(A))
 &=U_{\Lambda,\lambda}^*\pi_\at(\sigma_\st(A))U_{\Lambda,\lambda},
 && A\in\A^\st,\ \lambda\in\R,
 \label{eq:roadmap-standard-spectral-flow}\\
 U_{\Lambda,\lambda}^*P_{\Lambda,\lambda}U_{\Lambda,\lambda}
 &=P^\vac,
 &&\lambda\in I_\Lambda.
 \label{eq:roadmap-ground-state-transport}
\end{align}
On \(I_\Lambda\), the centered moving-frame Hamiltonian
\[
 K_{\Lambda,\lambda}^\nf
 :=U_{\Lambda,\lambda}^*(K_{\Lambda,\lambda}-E_{\Lambda,\lambda})U_{\Lambda,\lambda}
\]
has the same gap as \(K_{\Lambda,\lambda}\), with ground state energy zero and fixed ground state vector \(\Omega_\at\).

It remains to show that \(K_{\Lambda,\lambda}^\nf\) can be written in the form \eqref{eq:roadmap-normal-form}, with local terms satisfying \eqref{eq:abstract-roadmap-rooted-strength}.
Since \(K_{\Lambda,0}^\nf=K_0\), it suffices to decompose its derivative with respect to \(\lambda\) into quasi-local terms that individually annihilate the global vacuum.
The spectral-flow generator \(\mathscr S_{\Lambda,\lambda}\) is chosen so that the commutator below removes the couplings between the ground state of \(K_{\Lambda,\lambda}\) and its orthogonal complement.
The resulting filtered terms act as scalars on the ground-state subspace, so subtracting their individual expectations makes them annihilate it.
Differentiation gives the first equality below, and the inverse-Liouvillian identity \eqref{eq:filtered-commutator} gives the second \cite[Proposition~6.9]{NSY19}:
\begin{align*}
 \frac{\dd}{\dd \lambda}K_{\Lambda,\lambda}^\nf
  =U_{\Lambda,\lambda}^*
   \bigl(V_\Lambda+\I[\mathscr S_{\Lambda,\lambda},K_{\Lambda,\lambda}]
                         -\dot E_{\Lambda,\lambda}\1\bigr)U_{\Lambda,\lambda} =\sum_{Z\subset\Lambda}
       \pi_\at\sigma_{\Lambda,\lambda}\bigl(
       \cV_\Lambda(Z)^{\D_{\Lambda,\lambda}^\ph}\widehat\otimes\1\bigr)
       -\dot E_{\Lambda,\lambda}\1.
\end{align*}
Here the filtered perturbation terms are
\begin{equation}
 \cV_\Lambda(Z)^{\D_{\Lambda,\lambda}^\ph}
 :=\int_\R w(t)\tau_t^{\ph,\Lambda,\lambda}
                   (\cV_\Lambda(Z))\,\dd t,
 \label{eq:roadmap-filtered-perturbation}
\end{equation}
where the real filter \(w\) satisfies \(W'=w-\delta_0\) distributionally, \(\widehat w(0)=1\), and \(\supp\widehat w\subset(-g/2,g/2)\).
On \(I_\Lambda\), this support condition makes each term \(\pi_\at\sigma_\st(\cV_\Lambda(Z)^{\D_{\Lambda,\lambda}^\ph} \widehat\otimes\1)\) block diagonal between \(\ran P_{\Lambda,\lambda}\) and its orthogonal complement; see \eqref{eq:D-ground-blocks}.
Since \(\widehat w(0)=1\), filtering preserves ground-state expectations.
Together with \(\omega_\at\circ\sigma_\st=\omega_0\widehat\otimes\omega_{\rm aux}\), \eqref{eq:roadmap-standard-spectral-flow} shows that \(\omega_{\Lambda,\lambda}(A)\) equals the expectation of \(\pi_\at\sigma_\st(A\widehat\otimes\1)\) in \(U_{\Lambda,\lambda}\Omega_\at\), for every physical observable \(A\).
The Hellmann--Feynman formula therefore gives
\[
 \dot E_{\Lambda,\lambda}
 =\omega_{\Lambda,\lambda}(V_\Lambda^\ph)
 =\sum_{Z\subset\Lambda}\omega_{\Lambda,\lambda}
       (\cV_\Lambda(Z)^{\D_{\Lambda,\lambda}^\ph}).
\]
For every real \(\lambda\) let
\[
 G_{\Lambda,\lambda}(Z)
 :=\sigma_{\Lambda,\lambda}\bigl(
   \cV_\Lambda(Z)^{\D_{\Lambda,\lambda}^\ph}\widehat\otimes\1\bigr)
   -\omega_{\Lambda,\lambda}\bigl(
   \cV_\Lambda(Z)^{\D_{\Lambda,\lambda}^\ph}\bigr)\1.
\]
On \(I_\Lambda\), subtracting the individual expectations removes the ground-state block of each filtered term.
Ground-state transport \eqref{eq:roadmap-ground-state-transport} then gives
\[
 \pi_\at(G_{\Lambda,\lambda}(Z))P^\vac
 =P^\vac\pi_\at(G_{\Lambda,\lambda}(Z))=0.
\]
Thus the centered derivative is a sum of quasi-local terms that individually annihilate the fixed vacuum:
\begin{equation}
 \frac{\dd}{\dd \lambda}K_{\Lambda,\lambda}^\nf
 =\sum_{Z\subset\Lambda}\pi_\at(G_{\Lambda,\lambda}(Z)),
 \qquad \lambda\in I_\Lambda.
 \label{eq:abstract-roadmap-centered-derivative}
\end{equation}

After integrating \eqref{eq:abstract-roadmap-centered-derivative}, we localize each summand with the vacuum conditional expectation \cite{Araki04}, which replaces the degrees of freedom outside a region by their vacuum state rather than averaging them with the normalized trace.
If a self-adjoint observable annihilates the global vacuum, its vacuum conditional expectation onto a finite region annihilates the local vacuum on both sides; see Lemma~\ref{lem:vacuum-localization}.
For a nonempty set \(Z\), use its neighborhoods \(Z^{(R_n)}\) at the dyadic radii \(R_0:=0\), \(R_n:=2^{n-1}\) for \(n\ge1\), and define
\begin{equation*}
 \Delta_{Z,0}^\vac:=\EE_Z^\vac,
 \qquad
 \Delta_{Z,n}^\vac
 :=\EE_{Z^{(R_n)}}^\vac-\EE_{Z^{(R_{n-1})}}^\vac
 \quad(n\ge1).
\end{equation*}
For every real \(\lambda\), set
\begin{equation*}
 R_{\Lambda,\lambda}(Z,n)
 :=\Delta_{Z,n}^\vac\!\left(
   \int_0^1G_{\Lambda,t\lambda}(Z)\,\dd t\right).
\end{equation*}
Integrating \eqref{eq:abstract-roadmap-centered-derivative} and summing the telescopic shells gives the following normal form on \(I_\Lambda\cap J_{\rm loc}\), with \(J_{\rm loc}\) defined in Step~4:
\begin{equation}
 K_{\Lambda,\lambda}^\nf
 =K_0+\lambda\sum_{Z\subset\Lambda}\sum_{n\ge0}
   \pi_\at(R_{\Lambda,\lambda}(Z,n)),
 \qquad
 R_{\Lambda,\lambda}(Z,n)P_{Z^{(R_n)}}^{\vac}
 =P_{Z^{(R_n)}}^{\vac}R_{\Lambda,\lambda}(Z,n)=0.
 \label{eq:abstract-roadmap-vacuum-normal-form}
\end{equation}
The estimates in Step~4 prove absolute norm convergence of the sums over \(Z\) and \(n\) in \eqref{eq:abstract-roadmap-vacuum-normal-form}, for each fixed \(\Lambda\), and give the bound on the sum of norms in \eqref{eq:abstract-roadmap-rooted-strength}, uniformly in \(\Lambda\).

\paragraph{4. Prove locality uniformly in \(\Lambda\).}
The uniform locality estimates will be proved on the \(\Lambda\)-independent coupling domain
\begin{equation*}
 J_{\rm loc}:=
 \left\{\lambda\in\R\,|\,
 |\lambda|\|\cV\|_*\le g\right\}.
\end{equation*}
We decompose the physical generator \(\mathscr S_{\Lambda,\lambda}^\ph\) into local terms \(\eta_{\Lambda,\lambda}(Y)\), whose sum converges in norm, and obtain local decompositions of the filtered perturbation terms \(\cV_\Lambda(Z)^{\D_{\Lambda,\lambda}^\ph}\) defined in \eqref{eq:roadmap-filtered-perturbation}.
For stretched-exponential decay, we integrate standard Lieb--Robinson bounds \cite{LR72,NSY18} for \(\tau_t^{\ph,\Lambda,\lambda}\) against \(W\) and \(w\), respectively, and use tracial conditional expectations to form local shells.
For polynomial decay, we use the lossless interaction representatives of \cite[Theorems~2.6--2.7 and Corollary~6.1]{FermionicLR} directly.
Section~\ref{sec:card-extension} shows that these decompositions preserve the spatial exponent \(p_\Phi\).
In both decay regimes, the interaction \(\eta_{\Lambda,\lambda}\) obtained from these decompositions represents the physical spectral-flow generator \(\mathscr S_{\Lambda,\lambda}^\ph\) and has weighted norm bounded by a constant times \(\|\cV\|_*\).
Since \(J_{\rm loc}\) has length \(2g/\|\cV\|_*\), the integral of this norm is bounded independently of \(\Lambda\) and the perturbation strength.
Applying Lieb--Robinson bounds to the evolution generated by \(\eta_{\Lambda,\lambda}\) therefore gives uniform locality estimates for the physical flow \(\beta_{\Lambda,0,\lambda}\).
Combining these estimates with the locality of the fixed coordinate transformation \(\sigma_\st\) gives locality of \(\sigma_{\Lambda,\lambda}\); see Lemma~\ref{lem:uniform-alpha-locality}.
Apply the locality bound for \(\sigma_{\Lambda,\lambda}\) to the local decomposition of \(\cV_\Lambda(Z)^{\D_{\Lambda,\lambda}^\ph}\) in \eqref{eq:roadmap-filtered-perturbation}.
Subtracting its expectation in \(\omega_{\Lambda,\lambda}\) gives the locality estimate for the centered observable \(G_{\Lambda,\lambda}(Z)\).
The vacuum conditional expectations from Step 3 then give decay of the norms \(\|R_{\Lambda,\lambda}(Z,n)\|\) as the shell radius \(R_n\) increases.

These locality estimates yield a constant \(C_*\ge1\), depending only on the fixed locality, geometric, and free data, but not on \(\Lambda\), \(\lambda\), or the perturbation strength \(\|\cV\|_*\), such that
\begin{equation}
 f_{\Lambda,\lambda}:=
 \sup_{y\in\Gamma}
 \sum_{\substack{Z\subset\Lambda,\ n\ge0\\y\in Z^{(R_n)}}}
 \|R_{\Lambda,\lambda}(Z,n)\|
 \le C_*\|\cV\|_*,
 \label{eq:abstract-roadmap-rooted-normal-form}
\end{equation}
for all \(\Lambda\Subset\Gamma\) and \(\lambda\in J_{\rm loc}\).

\paragraph{5. Closing the bootstrap.}
Applying the excitation-counting estimate \eqref{eq:abstract-roadmap-form-bound} to the vacuum normal form \eqref{eq:abstract-roadmap-vacuum-normal-form} gives
\begin{equation}
 K_{\Lambda,\lambda}^\nf
 \ge\bigl(g-|\lambda|C_*
 \|\cV\|_*\bigr)N
 \label{eq:abstract-roadmap-gap-improvement}
\end{equation}
for \(\lambda\in I_\Lambda\cap J_{\rm loc}\).
Define, as in the main theorems,
\begin{equation*}
 \lambda_*:=
 \frac{g}{4C_*\|\cV\|_*},
 \qquad
 I_*:=[-\lambda_*,\lambda_*].
\end{equation*}
Since \(C_*\ge1\), one has \(I_*\subset J_{\rm loc}\), and the right-hand side of \eqref{eq:abstract-roadmap-gap-improvement} is at least \((3g/4)N\) on \(I_\Lambda\cap I_*\).
If \(I_\Lambda\) had an endpoint in \(I_*\), the margin \(3g/4>g/2\) would allow bounded perturbation theory to extend the simple ground state eigenvalue beyond it while keeping the gap above \(g/2\).
The spectral flow with our fixed filter would therefore still satisfy the ground-state transport identity \eqref{eq:roadmap-ground-state-transport} on the extended interval.
This continuation argument gives \(I_*\subset I_\Lambda\) for every finite \(\Lambda\), and the cutoff Hamiltonians have a simple ground state eigenvalue and a gap at least \(3g/4\) throughout \(I_*\).
Figure~\ref{fig:bootstrap} illustrates this continuity argument.

\begin{figure}[htbp]
\centering
\begin{tikzpicture}[stability diagram,x=.9cm,y=3cm]
 \def\xL{-4.7}\def\xRt{4.7}
 \fill[gsblk!16] (-1,0) rectangle (1,1.1);
 \draw[black!45,dash pattern=on 3pt off 2pt] (\xL,.5) -- (\xRt,.5)
   node[fig note,anchor=west] {bootstrap threshold};
 \draw[black!45,dash pattern=on 3pt off 2pt] (\xL,.75) -- (\xRt,.75);
 \draw[black!30,dash pattern=on 1pt off 1.5pt] (\xL,1) -- (0,1);
 \foreach \s in {-1,1} \draw[black!40,dash pattern=on 1pt off 1.5pt] (\s,0) -- (\s,.75);
 \draw[phys,line width=.9pt] (-3.2,.2) -- (0,1) -- (2.7,.325);
 \draw[phys,line width=.9pt,dash pattern=on 1.5pt off 1.5pt]
   (-4,0) -- (-3.2,.2) (2.7,.325) -- (4,0);
 \foreach \s/\v in {-1/.75,1/.75,-2/.5,2/.5}
   \filldraw[fill=white,draw=phys,line width=.6pt] (\s,\v) circle[radius=1.6pt];
 \node[phys!85!black,font=\footnotesize,anchor=west,align=left] at (3.35,.28)
   {$g-C_*|\lambda|\,\|\cV\|_*$\\[-1pt]proved on $I_\Lambda\cap J_{\rm loc}$};
 \draw[aux,line width=.8pt,smooth] plot coordinates
   {(-3.2,.5) (-2.6,.62) (-1.9,.78) (-1.2,.9) (-.5,.975) (0,1)
    (.6,.972) (1.3,.885) (1.9,.72) (2.3,.605) (2.7,.5)};
 \draw[aux,line width=.8pt,dash pattern=on 1.5pt off 1.5pt,smooth] plot coordinates
   {(-3.75,.37) (-3.45,.44) (-3.2,.5)};
 \draw[aux,line width=.8pt,dash pattern=on 1.5pt off 1.5pt,smooth] plot coordinates
   {(2.7,.5) (2.95,.43) (3.2,.37)};
 \foreach \s in {-3.2,2.7} \fill[aux] (\s,.5) circle[radius=1.6pt];
 \node[aux!85!black,font=\footnotesize,anchor=south west] at (1.55,.9)
   {gap of $K_{\Lambda,\lambda}$ (schematic)};
 \draw[->] (\xL,0) -- (\xRt+.3,0) node[below left=1pt] {$\lambda$};
 \draw[->] (\xL,0) -- (\xL,1.2);
 \foreach \v/\t in {.5/{g/2},.75/{3g/4},1/g}
   \draw (\xL,\v) -- ++(-.1,0) node[left,font=\footnotesize] {$\t$};
 \foreach \s/\t in {-4/-4\lambda_*,-2/-2\lambda_*,-1/-\lambda_*,0/0,1/\lambda_*,2/2\lambda_*,4/4\lambda_*}
   \draw (\s,0) -- ++(0,-.025) node[below=-1pt,font=\footnotesize] {$\t$};
 \draw[gsblk!75!black,line width=.9pt] (-1,-.2) -- (1,-.2);
 \foreach \s in {-1,1} \fill[gsblk!75!black] (\s,-.2) circle[radius=1.6pt];
 \node[gsblk!70!black,font=\footnotesize,anchor=west] at (1.15,-.2) {$I_*$};
 \draw[aux,line width=.9pt] (-3.2,-.3) -- (2.7,-.3);
 \foreach \s in {-3.2,2.7}
   \filldraw[fill=white,draw=aux,line width=.6pt] (\s,-.3) circle[radius=1.6pt];
 \node[aux!85!black,font=\footnotesize,anchor=west] at (2.85,-.3) {$I_\Lambda$};
 \draw[phys,line width=.9pt] (-4,-.4) -- (4,-.4);
 \draw[phys,line width=.9pt,dash pattern=on 1.5pt off 1.5pt]
   (\xL,-.4) -- (-4,-.4) (4,-.4) -- (\xRt,-.4);
 \node[phys!85!black,font=\footnotesize,anchor=west] at (\xRt+.15,-.4) {$J_{\rm loc}$};
\end{tikzpicture}
\caption{Closing the bootstrap for a fixed cutoff region \(\Lambda\).
The orange curve shows a schematic gap of \(K_{\Lambda,\lambda}\), and the blue curve shows the lower bound \(g-C_*|\lambda|\|\cV\|_*\), proved on \(I_\Lambda\cap J_{\rm loc}\).
Since \(C_*\ge1\), \([-4\lambda_*,4\lambda_*]\subseteq J_{\rm loc}\), so the locality estimates apply throughout the range of the blue curve.
The lower blue line marks this guaranteed interval, while its dashed extensions indicate that \(J_{\rm loc}=[-4C_*\lambda_*,4C_*\lambda_*]\) may be larger.
The shaded interval is \(I_*=[-\lambda_*,\lambda_*]\), where the gap bound is at least \(3g/4\).
By continuity from within \(I_\Lambda\), an endpoint of \(I_\Lambda\) in \(I_*\) would have gap at least \(3g/4>g/2\), allowing continuation beyond it and proving \(I_*\subset I_\Lambda\); see Proposition~\ref{prop:abstract-bootstrap} and Corollary~\ref{cor:cutoff-gap-form}.}
\label{fig:bootstrap}
\end{figure}
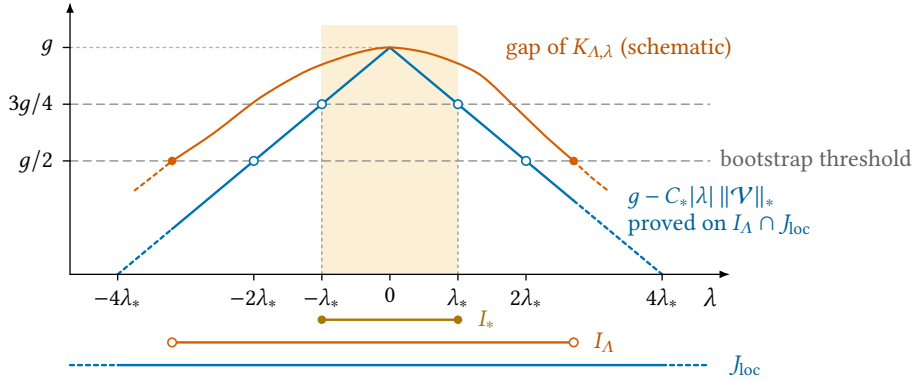

\paragraph{6. Transfer the gap to the physical cutoff state.}
Recall that \(\pi_\at\) is the product-vacuum GNS representation with vacuum vector \(\Omega_\at\).
On \(I_*\), the unitary \(U_{\Lambda,\lambda}\) maps this vacuum to a ground state vector of \(K_{\Lambda,\lambda}\).
Its vector state
\[
 \widehat\omega_{\Lambda,\lambda}(A)
 :=\langle U_{\Lambda,\lambda}\Omega_\at,
              \pi_\at(A)U_{\Lambda,\lambda}\Omega_\at\rangle,
 \qquad A\in\A^\st,
\]
satisfies
\begin{equation}
 \widehat\omega_{\Lambda,\lambda}\circ\sigma_\st
 =\omega_{\Lambda,\lambda}\widehat\otimes\omega_{\rm aux}.
 \label{eq:abstract-roadmap-flow-descent}
\end{equation}
Here \(\omega_{\Lambda,\lambda}=\omega_0\circ\beta_{\Lambda,0,\lambda}\) is the transported physical state from Step~1, and \(\omega_{\rm aux}\) is the auxiliary free ground state.
The factorization follows from the flow identity \eqref{eq:roadmap-standard-spectral-flow} and \(\omega_\at\circ\sigma_\st=\omega_0\widehat\otimes\omega_{\rm aux}\).
To transfer the gap bound for \(K_{\Lambda,\lambda}\) to the physical state \(\omega_{\Lambda,\lambda}\), fix \(A\in\A_{\loc}\) and set
\[
 \xi:=\pi_\at\sigma_\st(A\widehat\otimes\1)
             U_{\Lambda,\lambda}\Omega_\at.
\]
The transformation \(\sigma_\st\) intertwines the physical dynamics with the dynamics generated by \(K_{\Lambda,\lambda}\).
Applying the gap bound from Step~5 and the state factorization \eqref{eq:abstract-roadmap-flow-descent} gives
\begin{align*}
 \operatorname{Re}\omega_{\Lambda,\lambda}
       \bigl(A^*\delta_{\Phi_h+\lambda\cV_\Lambda}(A)\bigr)
 &=\langle\xi,(K_{\Lambda,\lambda}-E_{\Lambda,\lambda})\xi\rangle\\
 &\ge(g-C_*|\lambda|\|\cV\|_*)
       \bigl(\|\xi\|^2-|\langle U_{\Lambda,\lambda}\Omega_\at,\xi\rangle|^2\bigr)\\
 &=(g-C_*|\lambda|\|\cV\|_*)
       \bigl(\omega_{\Lambda,\lambda}(A^*A)-|\omega_{\Lambda,\lambda}(A)|^2\bigr).
\end{align*}
This proves the ground-state gap inequality for the physical state \(\omega_{\Lambda,\lambda}\), uniformly in \(\Lambda\) on \(I_*\).
Proposition~\ref{prop:decay-independent-descent} proves that \(\xi\in\operatorname{Dom}(K_{\Lambda,\lambda})\), which justifies applying the spectral gap inequality.

\paragraph{7. Remove the perturbation cutoff.}
We construct the full spectral-flow generator \(\eta_{\lambda}\) by applying Step~4 to the full physical dynamics, with \(\cV\) in place of \(\cV_\Lambda\).
To compare \(\eta_{\lambda}\) with the cutoff generators, choose any increasing sequence \((\Lambda_n)_{n\in\N}\) of finite subsets of \(\Gamma\) with \(\bigcup_n\Lambda_n=\Gamma\).
Along this exhaustion, the cutoff dynamics \(\tau_t^{\ph,\Lambda_n,\lambda}\) of \(\Phi_h+\lambda\cV_{\Lambda_n}\) converge in norm on each local observable to the full dynamics \(\tau_t^{\ph,\lambda}\) of \(\Phi_h+\lambda\cV\), uniformly in \(\lambda\in I_*\) and on compact time intervals.
Combining this convergence with the time decay of \(W\) and the uniform locality estimates of Step~4 gives
\[
 \eta_{\Lambda_n,\lambda}\stackrel{n\to\infty}{\longrightarrow}\eta_{\lambda}
 \qquad\text{locally and uniformly for }\lambda\in I_*\,.
\]
The convergence holds in the local weighted interaction norms of Appendix~\ref{app:nonautonomous-background}, with a slightly weaker spatial weight.
Proposition~\ref{prop:nonautonomous-local-limit-package} gives existence of the full physical cocycle \(\beta_{\lambda,\lambda'}\) generated by \(\eta_{\lambda}\) and convergence of the cutoff cocycles:
\[
 \sup_{\lambda,\lambda'\in I_*}
 \|\beta_{\Lambda_n,\lambda,\lambda'}(A)-\beta_{\lambda,\lambda'}(A)\|\stackrel{n\to\infty}{\longrightarrow}0
 \qquad(A\in\A).
\]
Hence the physical cutoff states converge pointwise to \(\omega_{\lambda}:=\omega_0\circ\beta_{0,\lambda}\), independently of the exhaustion.
This state is pure because \(\omega_0\) is pure and \(\beta_{0,\lambda}\) is an automorphism.

Passing the uniform gap inequality from Step 6 to the limit shows that \(\omega_\lambda\) is a ground state of \(\Phi_h+\lambda\mathcal V\) with the claimed gap bound.
All admissible filters give the same cutoff ground states and hence the same limiting ground-state path; see \eqref{eq:filter-independent-cutoff-state}.
The corresponding cocycles may differ.

Sections~\ref{sec:atomicization}--\ref{sec:limit} carry out the full proof first for exponential decay and then deduce the stretched-exponential case by a change of metric.
Section~\ref{sec:card-extension} proves the polynomial case, using the same algebraic and bootstrap arguments with the appropriate locality estimates.
For the interacting systems of Theorem~\ref{thm:product-reference-stability}, the unperturbed ground state is already a product state and its GNS Hamiltonian satisfies \(K_0\ge gN\); the coordinate map is the identity, and the subsequent normal-form, bootstrap, and limit arguments apply directly.

\subsection{Discussion of related literature}
\label{sec:comparison-gap-stability}

\paragraph{Stability of free fermion systems.}
Hastings \cite{Hastings19} and Koma \cite{Koma20} work on finite cubic lattices, assuming exponential spatial decay of the free Hamiltonian and perturbation and a free spectral gap bounded below independently of the volume.
Hastings allows free Hamiltonians with pairing and general even perturbations.
Koma develops Hastings' approach for number-conserving free Hamiltonians and assumes either bounded interaction order or an exponential support-size weight on the perturbation \cite[Assumptions~2.1--2.3]{Koma20}.
The use of doubling to overcome the topological obstruction to a product-state reduction goes back to Hastings \cite{Hastings08} and we use the explicit Majorana transformation of \(h\oplus(-h)\) from his gap-stability proof \cite[Section~III, (8)--(15)]{Hastings19}.
Hastings' argument already treats unperturbed Hamiltonians with dispersion by comparison with \(H_{\rm proj}=\sum_j n_j\): his Theorem~4 assumes \(H_0\ge H_{\rm proj}\) with the same ground space and gives linear gap loss \cite[Sections~III--IV]{Hastings19}.
The relative-bound argument for block-diagonal perturbations appears in the work of Bravyi, Hastings, and Michalakis \cite[Section~3]{BHM10}.
The spectral-flow bootstrap and the passage from quasi-local terms preserving the ground space to local terms annihilating local ground states follow Bravyi and Hastings \cite[Sections~4--7]{BH11}.
Hastings estimates \(U_{\lambda}^*H_0U_{\lambda}-H_0\) and \(\lambda U_{\lambda}^*VU_{\lambda}\) separately, using differentiation for the first term, and then filters the contributions from \(H_0\) and \(\lambda V\) \cite[(30)--(31) and Lemma~2]{Hastings19}.
We simplify this construction by differentiating the full centered moving-frame Hamiltonian: the inverse-Liouvillian identity expresses its derivative directly through centered, filtered perturbation terms.

\paragraph{Product-state perturbation theory.}
Yarotsky \cite{Yarotsky04} proves stability of uniformly gapped onsite spin Hamiltonians with nondegenerate ground states under weak bounded finite-range perturbations, using a Kirkwood--Thomas-type expansion \cite{KirkwoodThomas83}.
An alternative is the local Lie--Schwinger block-diagonalization method introduced by Fr\"ohlich and Pizzo \cite{FrohlichPizzo20} and developed by Del Vecchio, Fr\"ohlich, Pizzo, and Rossi \cite{DelVecchioFrohlichPizzoRossi22}, which gives volume-independent gap stability for weak short-range perturbations of gapped onsite Hamiltonians.
Yarotsky's later result \cite{Yarotsky06} also treats interacting Hamiltonians with product ground states under relative bounds; its Example~2 allows noncommuting finite-range terms in the unperturbed interaction.
In our free-fermion proof, the reduction produces a product ground state, but the GNS Hamiltonian \(K_0=\dd\Gamma(|h|)\) generally includes hopping between distinct sites.
The onsite results therefore do not apply directly, and we obtain the required relative form estimate using \(K_0\ge gN\).

\paragraph{Frustration-free systems.}
Michalakis and Zwolak \cite[Theorem~1]{MichalakisZwolak13} prove stability of gapped frustration-free Hamiltonians under local topological quantum order (LTQO) and a local-gap condition, using spectral flow and a relative form bound.
Jackson, Nachtergaele, and Young prove exponential local indistinguishability and hence LTQO for AKLT models on the hexagonal and Lieb lattices \cite[Theorem~1.1 and Corollary~1.2]{JacksonNachtergaeleYoung26}; their Corollary~1.3 then gives gap stability under a uniform finite-volume gap assumption for the unperturbed system.
Moon and Nachtergaele \cite[Theorems~3.11 and~4.4]{MoonNachtergaele18} prove stability for one-dimensional spin and fermion chains with uniformly gapped finite-range frustration-free interactions, under the respective LTQO conditions and weak stretched-exponentially decaying perturbations.
The bulk-gap theorem of Nachtergaele, Sims, and Young \cite{NSY24}, building on their earlier work \cite{NSY22}, assumes a finite-range, uniformly bounded, frustration-free interaction with a bulk gap, LTQO, and polynomial lower bounds on suitable local finite-region gaps; the perturbation decays stretched-exponentially.
Its Theorem~2.8 also gives a transported ground-state path.
We use their construction of the transport in a fixed GNS representation.
Their theorem does not cover our general setting, since the free interaction \(\Phi_h\) need not admit a finite-range, frustration-free presentation satisfying their local-gap and LTQO assumptions.
At sufficiently weak coupling, the companion paper \cite{LPPLV3} proves an LTQO-type estimate for these fermion systems: local observables act approximately as scalars on finite-sample ground spaces, with an error that decays with distance to the boundary.

\paragraph{Stability of quantum codes.}
Yin and Lucas \cite[Theorem~2]{YinLucas25} prove gap stability for commuting low-density parity-check code Hamiltonians under quantitative conditions on the code distance and local checks.
They allow spatially nonlocal perturbations, whose strength is measured by a norm with an exponential support-size weight.
Using iterative Schrieffer--Wolff transformations, they bound the splitting within the possibly degenerate low-energy sector and show that a gap separates it from the rest of the spectrum.
Their construction also gives a unitary relating the unperturbed and perturbed spectral projections up to an error that vanishes as the system size grows; its generator satisfies a support-size-weighted bound.
For perturbations local on a graph, Lavasani, Gullans, Albert, and Barkeshli \cite{LavasaniEtAl24} and De Roeck, Khemani, Li, O'Dea, and Rakovszky \cite{DeRoeckEtAl25} prove related stability results under assumptions of local indistinguishability or consistency of local ground spaces.
The latter result also applies to expander graphs.

\paragraph{Constructive fermionic expansions.}
De Roeck and Salmhofer \cite[Theorems~2--3]{DRS19} use a Grassmann tree--determinant expansion to prove correlation decay and a spectral gap uniform in the volume for weakly interacting fermions.
Their argument requires a uniformly gapped finite-volume free Hamiltonian, weighted summability of its imaginary-time covariance, and sufficiently small interaction kernels.
These assumptions allow polynomial decay and non-self-adjoint perturbations.
For compatible finite-volume families, the uniform gap passes to the bulk by \cite[Lemma~G.1]{WMMMT25}; see also \cite[Lemma~3.1]{MMMTW26}.
Our operator argument constructs the bulk ground state together with its quasi-local transport directly.
In their exponentially localized lattice example \cite[Theorem~1]{DRS19}, the stated gap loss is of H\"older order in the coupling, whereas ours is linear.

Giuliani, Mastropietro, and Porta \cite{GMP17} consider two-dimensional periodic band insulators, with hopping and weak density--density interactions decaying faster than any power.
Their multiscale Grassmann expansion proves that the interacting Hall conductivity equals its free-fermion value.\footnote{The authors explained to the present author (personal communication) that exponential imaginary-time estimates within this expansion, followed by finite-dimensional continuation, also yield a gap uniform in the torus size; this consequence is not stated as a theorem in \cite{GMP17}.}

Rayudu et al.\ \cite[Theorem~3]{RayuduEtAl26} prove that adding a positive field to interacting fermion Hamiltonians in the Lee--Yang class gives a gap bounded below by a fixed multiple of the field strength, uniformly in the volume and without a weak-coupling assumption.

\paragraph{Long-range perturbations of Ising phases.}
Lapa and Levin \cite{LapaLevin23} use a convergent polymer expansion to prove gap stability and exponentially small ground-sector splitting for Ising systems with symmetry breaking under weak symmetry-preserving long-range perturbations.
Their summability condition includes one-dimensional two-site interactions decaying as \((1+|x-y|)^{-\varpi}\) with \(\varpi>1\).
Via the Jordan--Wigner transformation, their result also applies to the Kitaev wire near its exactly solvable point, with weak long-range density--density interactions \cite[Section~II\,C]{LapaLevin23}.

\section{The stacked free system and the product-vacuum representation}
\label{sec:atomicization}

In Sections~\ref{sec:atomicization}--\ref{sec:limit}, we first prove Theorem~\ref{thm:main} for \(\theta_{\rm in}=1\); the final part of its proof treats \(\theta_{\rm in}<1\).
This section carries out the first part of Step~2 of Subsection~\ref{sec:proof-roadmap}: the transformation to product-vacuum coordinates, its locality, and the free GNS Hamiltonian.

In Majorana coordinates, \(\A\) already has Araki's self-dual CAR structure, with coefficientwise conjugation on \(\Hsp_{\mathrm M}\).
Araki shows that any two basis projections are related by a Bogoliubov transformation and that their Fock states are related by the induced automorphism \cite[Sections~2 and~4]{Araki70}.
Thus the free ground state can be transformed algebraically into a prescribed product vacuum, but this statement provides no spatial locality estimates.
Writing creation and annihilation operators together in the self-dual formalism does not introduce independent fermions.

Following Hastings' doubling argument, we adjoin an independent auxiliary free system with the opposite Majorana coefficient operator to obtain a spatially controlled transformation without assuming topological triviality.
For a band insulator, this supplies the complementary occupied band.
We use Hastings' explicit transformation of \(h\oplus(-h)\), whose matrix depends on \(\sgnh\), from his gap-stability proof \cite[Section~III, (8)--(15)]{Hastings19}.
We implement this transformation on the infinite-volume CAR algebra; Proposition~\ref{prop:exponential-stable-locality} gives the quantitative locality estimates used below.

We use the graded tensor-product and composite-label metric conventions of Appendix~\ref{app:CAR-geometry}.
The sign \(\sgnh\) and modulus \(|h|\) defined in \eqref{eq:free-sign} satisfy
\[
 \sgnh^2=\1,\qquad \overline{\sgnh}=-\sgnh,\qquad
 h=\sgnh|h|,\qquad \overline{|h|}=|h|,\qquad |h|\ge g\1.
\]
Their locality follows from Lemmas~\ref{lem:weighted-free-functional-calculus} and~\ref{lem:polynomial-free-functional-calculus}.

\subsection{The auxiliary free system and the change of coordinates}

Let \(\Hsp_{\rm aux}\) be a labeled copy of the physical space \(\Hsp\), and set \(\A_{\rm aux}:=\CAR(\Hsp_{\rm aux})\).
Let \(J_{\rm aux}:\Hsp_{\mathrm M}\to\Hsp_{\mathrm{M,aux}}\) be the canonical identification of their Majorana coefficient spaces, and define
\begin{equation}
 h_{\rm aux}:=-J_{\rm aux}hJ_{\rm aux}^*.
 \label{eq:auxiliary}
\end{equation}
Its interaction \(\Phi_{h_{\rm aux}}\) is defined by \eqref{eq:quadratic-Majorana-interaction}, and its pure quasifree ground state is denoted by \(\omega_{\rm aux}\).
The subscript \(\st\) denotes the stacked system.
For \(X\subseteq\Gamma\), set
\[
 \A^\st:=\CAR(\Hsp\oplus\Hsp_{\rm aux})=\A\widehat\otimes\A_{\rm aux},\qquad
 \A_X^\st:=\CAR(\ell^2(X;\C^{2r})),\qquad
 \A_{\loc}^\st:=\bigcup_{X\Subset\Gamma}\A_X^\st.
\]
Let \(\gamma_{x,j}^{\rm aux}\) be the Majorana generators of \(\A_{\rm aux}\), defined by \eqref{eq:physical-Majoranas} with the auxiliary annihilation operators.
For \(x\in\Gamma\) and \(1\le j\le2r\), set
\[
 \gamma_{x,j,1}:=\gamma_{x,j}\widehat\otimes\1,
 \qquad
 \gamma_{x,j,2}:=\1\widehat\otimes\gamma_{x,j}^{\rm aux}.
\]
In block operator matrices on the stacked coefficient space, we identify the auxiliary copy with \(\Hsp_{\mathrm M}\) using \(J_{\rm aux}\).
Define
\begin{equation}
 O_\st:=\frac1{\sqrt2}
 \begin{pmatrix}\1&\I\sgnh\\ \I\sgnh&\1\end{pmatrix}.
 \label{eq:O-st}
\end{equation}

\begin{lemma}
\label{lem:stable-atomicization}
The operator \(O_\st\) is real orthogonal and satisfies
\begin{equation}
 O_\st^{\mathsf T}
 \begin{pmatrix}h&0\\0&-h\end{pmatrix}O_\st
 =
 \begin{pmatrix}0&\I|h|\\-\I|h|&0\end{pmatrix}.
 \label{eq:hamiltonian-atomicization}
\end{equation}
It defines a parity-preserving Bogoliubov automorphism of \(\A^\st\) by
\begin{equation}
 \sigma_\st(\gamma_u):=\sum_v(O_\st)_{uv}\gamma_v,
 \label{eq:Majorana-atomicizing-action}
\end{equation}
where \(u,v\in\Gamma\times\{1,\ldots,2r\}\times\{1,2\}\) label the site, Majorana generator, and layer.
The inverse automorphism has coefficient matrix \(O_\st^{\mathsf T}\).
\end{lemma}

\begin{proof}
The matrix \(\I\sgnh\) is real antisymmetric and satisfies \((\I\sgnh)^2=-\1\), so direct multiplication gives orthogonality of \(O_\st\) and, using \(h\sgnh=\sgnh h=|h|\), \eqref{eq:hamiltonian-atomicization}.
The real orthogonal operator \(O_\st\) defines a parity-preserving Bogoliubov automorphism, with inverse induced by \(O_\st^{\mathsf T}\) \cite[Section~2]{Araki70}.
\end{proof}

The weighted Schur bound on \(\sgnh\) also controls the action of \(\sigma_\st\) on arbitrary local observables.

\begin{proposition}
\label{prop:exponential-stable-locality}
Define
\begin{equation}
 k_\st:=-\frac\pi4\begin{pmatrix}0&\sgnh\\\sgnh&0\end{pmatrix}.
 \label{eq:atomicizing-path-generator}
\end{equation}
Under either locality hypothesis below, \(\Phi_{k_\st}\) is an even self-adjoint interaction with \(\sigma_\st=\tau_1^{\Phi_{k_\st}}\).
If \eqref{eq:h-locality} holds, then, for the rate \(\mu_1\) in Lemma~\ref{lem:weighted-free-functional-calculus},
\[
 \|\Phi_{k_\st}\|_{\mu_1,\theta_{\rm in};2}
 \le 2\|\Phi_{k_\st}\|_{\mu_1,\theta_{\rm in};1}
 \le 8\pi r^2\|\sgnh\|_{\mu_1,\theta_{\rm in}}^{\rm Sch}<\infty.
\]
There are also \(a_\sigma,C_\sigma>0\), depending only on \(\dvol,C_\Gamma,r,\mu_h,C_h/g, \theta_{\rm in}\), such that
\begin{equation}
 \|O_\st\|_{a_\sigma,\theta_{\rm in}}^{\rm Sch}\le C_\sigma,
 \qquad
 \|\sigma_\st^{\pm1}\|_{a_\sigma,\theta_{\rm in};1}^{\rm ql}\le C_\sigma.
 \label{eq:stable-atomicization-quantitative}
\end{equation}
If instead \(\dvol>0\), \(p>\dvol\), and \(C_h:=\|h\|^{\mathrm{Sch,pol}}_p<\infty\), then
\[
 \|\Phi_{k_\st}\|^{\pol}_{p;1}
 \le 4\pi r^2\|\sgnh\|^{\mathrm{Sch,pol}}_p<\infty,
\]
and
\begin{equation}
 \|O_\st\|^{\mathrm{Sch,pol}}_p\le C_\sigma,
 \qquad
 \inf_{B\in\A_{X^{(R)}}^\st}\|\sigma_\st^{\pm1}(A)-B\|
 \le C_\sigma|X|\|A\|(1+R)^{-p}
 \label{eq:stable-atomicization-polynomial}
\end{equation}
for every nonempty finite \(X\), \(A\in\A_X^\st\), and \(R\ge0\).
Here \(C_\sigma\) depends only on \(p,\dvol,C_\Gamma,r,C_h/g\).
\end{proposition}

\begin{proof}
Lemmas~\ref{lem:weighted-free-functional-calculus} and~\ref{lem:polynomial-free-functional-calculus} give the respective weighted Schur bounds for \(\sgnh\), hence for \(O_\st\) and its transpose.
To control the automorphism on local observables, consider the real orthogonal path
\[
 O_\st(t):=\exp\!\left[
 \frac{\pi t}{4}
 \begin{pmatrix}0&\I \sgnh\\ \I \sgnh&0\end{pmatrix}\right],
 \qquad 0\le t\le1.
\]
The block matrix in the exponent is real antisymmetric.
Using \(\sgnh^2=\1\) to evaluate the exponential gives \(O_\st(1)=O_\st\).
By \eqref{eq:Majorana-free-dynamics}, this path is the Majorana action of the dynamics generated by \(\Phi_{k_\st}\), which proves \(\sigma_\st=\tau_1^{\Phi_{k_\st}}\).
The quadratic-interaction bounds \eqref{eq:quadratic-interaction-bounds}, applied with \(2r\) modes per site, give both stated interaction-norm bounds; here \(\Phi_{k_\st}(X)=0\) for \(|X|>2\) gives the comparison between the first and second support moments.
Apply \cite[Theorem~2.9 and Remark~2.10]{FermionicLR} to the constant path \(\Phi_{k_\st}\) on the unit time interval, with spatial weight \(\E^{\mu_1R^{\theta_{\rm in}}}\) in the stretched-exponential case and \((1+R)^p\) in the polynomial case.
The first-support-moment bounds above give the automorphism estimates in \eqref{eq:stable-atomicization-quantitative} and \eqref{eq:stable-atomicization-polynomial}; reversing the path gives the same bounds for the inverse.
\end{proof}

Define
\begin{equation}
 h_\at:=
 \begin{pmatrix}0&\I|h|\\-\I|h|&0\end{pmatrix}.
 \label{eq:atomic-one-particle-Hamiltonian}
\end{equation}
Let \(\Phi_{h_\at}\) be the corresponding quadratic interaction on the stacked algebra, defined by \eqref{eq:quadratic-Majorana-interaction}.
Write \(\tau_t^{\at,0}:=\tau_t^{\Phi_{h_\at}}\), and let \(\tau_t^{\rm st,0}\) be the stacked free dynamics generated by \(\Phi_h\widehat\otimes\1+\1\widehat\otimes\Phi_{h_{\rm aux}}\).
The field action \eqref{eq:Majorana-atomicizing-action} sends a quadratic coefficient matrix \(k\) to \(O_\st^{\mathsf T}kO_\st\).
Combining \eqref{eq:Majorana-free-dynamics} with \eqref{eq:hamiltonian-atomicization} gives
\begin{equation}
 \tau_t^{\at,0}\circ\sigma_\st
 =\sigma_\st\circ\tau_t^{\rm st,0},
 \qquad t\in\R.
 \label{eq:free-atomic-dynamics-covariance}
\end{equation}
We also write \(\tau_t^{\rm aux}:=\tau_t^{\Phi_{h_{\rm aux}}}\).

Differentiating \eqref{eq:free-atomic-dynamics-covariance} on local observables gives
\begin{equation}
 \sigma_\st(\A_{\loc}^\st)
 \subset\operatorname{Dom}(\overline\delta_{\Phi_{h_\at}}),
 \qquad
 \overline\delta_{\Phi_{h_\at}}(\sigma_\st(B))
 =\sigma_\st\!\left(
 \delta_{\Phi_h\widehat\otimes\1+
 \1\widehat\otimes\Phi_{h_{\rm aux}}}(B)\right),
 \quad B\in\A_{\loc}^\st.
 \label{eq:free-atomic-derivation-covariance}
\end{equation}

\subsection{The product vacuum and the excitation Hamiltonian}
\label{sec:excitation-GNS}

We show that the transformed state is an onsite product vacuum and prove the comparison \(K_0\ge gN\) used in the excitation count.
Put
\begin{equation}
 \omega_\at:=(\omega_0\widehat\otimes\omega_{\rm aux})
 \circ\sigma_\st^{-1}.
 \label{eq:atomic-state}
\end{equation}
The inverse field transformation and \eqref{eq:free-Majorana-state} give its two-point matrix
\begin{equation}
 O_\st^{\mathsf T}
 \left[\1+\begin{pmatrix}\sgnh&0\\0&-\sgnh\end{pmatrix}\right]O_\st
 =\1+\begin{pmatrix}0&\I\1\\-\I\1&0\end{pmatrix}.
 \label{eq:atomic-Majorana-covariance}
\end{equation}
Define the onsite excitation annihilation operators by
\begin{equation}
 c_{x,j}:=\frac{\gamma_{x,j,1}+\I\gamma_{x,j,2}}2,
 \qquad 1\le j\le2r.
 \label{eq:atomic-fields}
\end{equation}
They satisfy the usual creation--annihilation CAR, and \eqref{eq:atomic-Majorana-covariance} gives \(\omega_\at(c_{x,j}^*c_{y,k})=0\).
Since the state is quasifree, it is the empty product state for these onsite modes, hence a pure onsite product state on \(\A^\st\).
In its GNS triple \((\pi_\at,\cH_\at,\Omega_\at)\),
\begin{equation}
 \pi_\at(c_{x,j})\Omega_\at=0.
 \label{eq:vacuum-annihilation}
\end{equation}
Define
\begin{equation}
 n_x:=\sum_{j=1}^{2r}c_{x,j}^*c_{x,j},
 \qquad N_X:=\sum_{x\in X}n_x,
 \qquad N:=\sum_{x\in\Gamma}\pi_\at(n_x).
 \label{eq:number-operators}
\end{equation}
The last sum, understood in the sense of nonnegative quadratic forms, defines the number operator on \(\cH_\at\).
For \(x\in\Gamma\) and finite \(X\subset\Gamma\), define the vacuum projections
\begin{equation}
 P_x:=\prod_{j=1}^{2r}(\1-c_{x,j}^*c_{x,j}),
 \qquad P_X^\vac:=\prod_{x\in X}P_x.
 \label{eq:local-vacuum-projection}
\end{equation}

We use the standard Fock implementation of positive quasifree dynamics \cite[Sections~5.2.3--5.2.4]{BR97} to determine \(K_0\) and its form domain.

\begin{proposition}
\label{prop:free-GNS-Hamiltonian}
The GNS representation of \(\omega_\at\) is the Fock representation over \(\Hsp_{\mathrm M}\), with annihilation operators \eqref{eq:atomic-fields}.
In these annihilation variables, the one-particle energy operator is \(|h|\), so the free dynamics is implemented by
\begin{equation}
 K_0=\dd\Gamma(|h|)
 \label{eq:K0}
\end{equation}
and the following inequalities hold as quadratic forms:
\begin{equation}
 gN\le K_0\le\|h\|N,
 \qquad
 K_0\ge g\left(\1-|\Omega_\at\rangle\langle\Omega_\at|\right).
 \label{eq:K0-domination}
\end{equation}
In particular, \(\operatorname{Dom}(K_0^{1/2})=\operatorname{Dom}(N^{1/2})\), with equivalent form norms.
The linear span of vectors obtained by applying finitely many local creation operators \(\pi_\at(c_{x,j}^*)\) to \(\Omega_\at\) is a common form core.
The physical free ground state \(\omega_0\) satisfies the gap inequality \eqref{eq:free-ground-gap}.
\end{proposition}

\begin{proof}
Since \(\omega_\at\) is the vacuum state for the annihilation operators \(c_{x,j}\), its GNS representation is the Fock representation over \(\Hsp_{\mathrm M}\).
For a one-particle wave function \(\psi=(\psi_{x,j})\in\Hsp_{\mathrm M}\), define \(c(\psi):=\sum_{x\in\Gamma}\sum_{j=1}^{2r}\overline{\psi_{x,j}}c_{x,j}\).
The CAR imply convergence in norm.
The second-quantization operator \(\dd\Gamma(|h|)\) annihilates \(\Omega_\at\).
For \(n\ge1\) and \(\psi_1,\ldots,\psi_n\in\Hsp_{\mathrm M}\), its action on an \(n\)-particle vector is
\[
 \dd\Gamma(|h|)(\psi_1\wedge\cdots\wedge\psi_n)
 =\sum_{j=1}^n\psi_1\wedge\cdots\wedge(|h|\psi_j)\wedge\cdots\wedge\psi_n.
\]
Equations~\eqref{eq:Majorana-free-dynamics} and \eqref{eq:atomic-one-particle-Hamiltonian}, together with the real symmetry of \(|h|\), give the field dynamics, which second quantization implements as follows \cite[Sections~5.2.3--5.2.4]{BR97}:
\[
 \pi_\at(\tau_t^{\at,0}(c(\psi)))
 =\pi_\at(c(\E^{\I t|h|}\psi)) =\E^{\I t\dd\Gamma(|h|)}\pi_\at(c(\psi))\E^{-\I t\dd\Gamma(|h|)}.
\]
Since this unitary group fixes \(\Omega_\at\), the GNS Hamiltonian is \(K_0=\dd\Gamma(|h|)\).
Under this representation, \(N=\dd\Gamma(\1)\) and \(N\ge\1-|\Omega_\at\rangle\langle\Omega_\at|\), so \eqref{eq:K0-domination} follows from \(g\1\le|h|\le\|h\|\1\).
The two-sided form bound gives equality of the form domains and equivalence of the form norms.
Finite-particle vectors with finitely supported wave functions form a form core for \(N\) and therefore also for \(K_0\).

For \(A\in\A_\loc\), set \(B:=\sigma_\st(A\widehat\otimes\1)\).
By \eqref{eq:free-atomic-derivation-covariance}, \(B\in\operatorname{Dom}(\overline\delta_{\Phi_{h_\at}})\).
The identity
\[
 \E^{\I tK_0}\pi_\at(B)\Omega_\at
 =\pi_\at(\tau_t^{\at,0}(B))\Omega_\at
\]
implies \(\pi_\at(B)\Omega_\at\in\operatorname{Dom}(K_0)\), because its right-hand side is differentiable at \(t=0\).
Differentiating at zero and applying the gap bound in \eqref{eq:K0-domination} gives
\[
 \omega_\at(B^*\overline\delta_{\Phi_{h_\at}}(B))
 =\langle\pi_\at(B)\Omega_\at,K_0\pi_\at(B)\Omega_\at\rangle
 \ge g\bigl(\omega_\at(B^*B)-|\omega_\at(B)|^2\bigr).
\]
Using \eqref{eq:free-atomic-dynamics-covariance} and \eqref{eq:atomic-state}, we rewrite the energy expectation as \(\omega_0(A^*\delta_{\Phi_h}(A))\) and the variance as \(\omega_0(A^*A)-|\omega_0(A)|^2\).
This proves \eqref{eq:free-ground-gap}.
\end{proof}

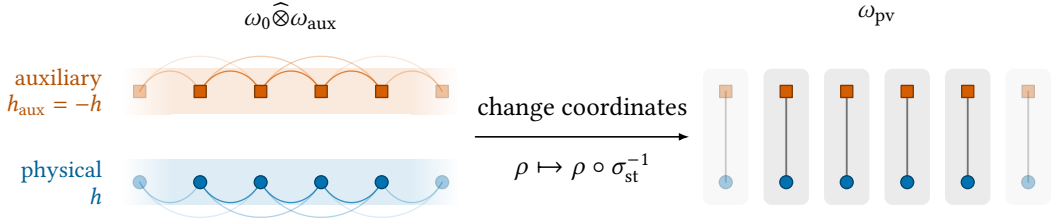
\begin{figure}[htbp]
\centering
\begin{tikzpicture}[stability diagram]
 \def\dx{.8}   
 \def\yP{0}     
 \def\yA{1.2}   
 \def\hb{.3}    
 \def\xR{7.75}   
 \foreach \y/\c in {\yP/phys,\yA/aux}{
  \shade[left color=white,right color=\c!13]
    (-1.3*\dx,\y-\hb) rectangle (-.5*\dx+.01,\y+\hb);
  \fill[\c!13] (-.5*\dx,\y-\hb) rectangle (3.5*\dx,\y+\hb);
  \shade[left color=\c!13,right color=white]
    (3.5*\dx-.01,\y-\hb) rectangle (4.3*\dx,\y+\hb);
 }
 \foreach \i/\op in {-1/.3,0/.75,1/.75,2/.75,3/.3}{
  \pgfmathtruncatemacro\j{\i+1}
  \draw[phys,opacity=\op] (\i*\dx,\yP) .. controls +(.2,-.36) and +(-.2,-.36) .. (\j*\dx,\yP);
  \draw[aux,opacity=\op]  (\i*\dx,\yA) .. controls +(.2,.36)  and +(-.2,.36)  .. (\j*\dx,\yA);
 }
 \foreach \i/\op in {-1/.14,0/.3,1/.3,2/.14}{
  \pgfmathtruncatemacro\j{\i+2}
  \draw[phys,opacity=\op] (\i*\dx,\yP) .. controls +(.35,-.62) and +(-.35,-.62) .. (\j*\dx,\yP);
  \draw[aux,opacity=\op]  (\i*\dx,\yA) .. controls +(.35,.62)  and +(-.35,.62)  .. (\j*\dx,\yA);
 }
 \foreach \i/\op in {-1/.35,0/1,1/1,2/1,3/1,4/.35}{
  \begin{scope}[opacity=\op]
   \filldraw[fill=phys,draw=phys!60!black] (\i*\dx,\yP) circle[radius=.085];
   \filldraw[fill=aux,draw=aux!60!black]
     (\i*\dx-.08,\yA-.08) rectangle (\i*\dx+.08,\yA+.08);
  \end{scope}
 }
 \node[anchor=east,align=right,text=aux!85!black,font=\footnotesize]
   at (-1.45*\dx,\yA) {auxiliary\\[-1pt]$h_{\rm aux}=-h$};
 \node[anchor=east,align=right,text=phys!85!black,font=\footnotesize]
   at (-1.45*\dx,\yP) {physical\\[-1pt]$h$};
 \node[fig title] at (1.5*\dx,\yA+1.0)
   {$\omega_0\widehat\otimes\omega_{\rm aux}$};
 \draw[stability arrow] (4.55*\dx,.5*\yA) -- (\xR-1.6*\dx,.5*\yA)
   node[midway,above=2pt] {change coordinates}
   node[midway,below=2pt] {$\rho\mapsto\rho\circ\sigma_\st^{-1}$};
 \foreach \i/\op in {-1/.35,0/1,1/1,2/1,3/1,4/.35}{
  \begin{scope}[opacity=\op,shift={(\xR+\i*\dx,0)}]
   \fill[black!8,rounded corners=3pt] (-.3,\yP-\hb) rectangle (.3,\yA+\hb);
   \draw[black!55,line width=.7pt] (0,\yP) -- (0,\yA);
   \filldraw[fill=phys,draw=phys!60!black] (0,\yP) circle[radius=.085];
   \filldraw[fill=aux,draw=aux!60!black] (-.08,\yA-.08) rectangle (.08,\yA+.08);
  \end{scope}
 }
 \node[fig title] at (\xR+1.5*\dx,\yA+1.0) {$\omega_\at$};
\end{tikzpicture}
\caption{Stacking and the product vacuum.
Circles represent the physical layer and squares the auxiliary layer at the same sites; shaded regions indicate factors of the state.
Initially, the two free layers are independent, while the curves indicate possible entanglement between sites within each layer.
After the change of coordinates \eqref{eq:atomic-state}, the state factorizes over sites, with entanglement between the physical and auxiliary degrees of freedom at each site.
The transformed state is the empty state of the \(c\)-fields.
The GNS Hamiltonian \(K_0=\dd\Gamma(|h|)\) generally still couples distinct sites, but satisfies \eqref{eq:K0-domination}.}
\label{fig:stable-product-vacuum}
\end{figure}

\section{Spectral flow in a fixed GNS representation}
\label{sec:local-cutoffs}

We write \(\Phi_0\) for the unperturbed physical interaction and set \(\Phi_{\lambda}=\Phi_0+\lambda\cV\).
For the weakly interacting fermion systems of Theorems~\ref{thm:main} and~\ref{thm:card-stability}, we take \(\Phi_0=\Phi_h\) and use the auxiliary system and the map \(\sigma_\st\) from Section~\ref{sec:atomicization}.
For the interacting systems with product ground states in Theorem~\ref{thm:product-reference-stability}, we take the auxiliary algebra to be \(\C\), with its unique state and identity dynamics, and set \(\sigma_\st=\id\).
In this case \((\A^\st,\omega_\at,\pi_\at,\cH_\at,\Omega_\at) =(\A,\omega_0,\pi_0,\cH_0,\Omega_0)\).
In both cases the unperturbed dynamics in product-vacuum coordinates is implemented by \(K_0\), and \(\omega_\at\circ\sigma_\st=\omega_0\widehat\otimes\omega_{\rm aux}\).

For free fermions, use the excitation operators \eqref{eq:number-operators} and the vacuum projections \eqref{eq:local-vacuum-projection}.
For interacting systems with product ground states, let \(P_x\) project onto the prescribed pure state at site \(x\), and use the same definitions of \(P_X^\vac,N_X,N\) with \(n_x:=\1-P_x\).
In both cases, the operators \(n_x\) are nonnegative and even, and satisfy \(n_xP_x=0\) and \(n_x\ge\1-P_x\).
In each application,
\[
 N\ge\1-P^\vac,\qquad \ker N=\C\Omega_\at,\qquad K_0\ge gN,
 \quad P^\vac=|\Omega_\at\rangle\langle\Omega_\at|.
\]
The last inequality includes \(\operatorname{Dom}(K_0^{1/2})\subseteq\operatorname{Dom}(N^{1/2})\).
Proposition~\ref{prop:free-GNS-Hamiltonian} proves these facts for free fermion systems; the corresponding assertions for interacting systems with product ground states are proved in Section~\ref{sec:product-reference-proof}.

We restrict the physical perturbation to a finite region while keeping the product-vacuum GNS representation \(\pi_\at\) fixed.
The corresponding Hamiltonians define the bootstrap interval \(I_\Lambda\) introduced in Step~2 of the roadmap.

\subsection{Cutoff Hamiltonians and the bootstrap interval}
\label{sec:bootstrap-interval}

For a finite perturbation region \(\Lambda\), define
\begin{align}
 \cV_\Lambda(X)&:=\1_{\{X\subset\Lambda\}}\cV(X),
 \qquad V_\Lambda^\ph:=\sum_{X\subset\Lambda}\cV(X),
 \nonumber\\
 V_\Lambda&:=\pi_\at\sigma_\st(V_\Lambda^\ph\widehat\otimes\1).
 \label{eq:V-Lambda}
\end{align}
On \(\cH_\at\), set
\begin{equation}
 K_{\Lambda,\lambda}:=K_0+\lambda V_\Lambda,\qquad \lambda\in\R.
 \label{eq:K-Lambda-s}
\end{equation}
Write \(\tau_t^{\ph,\Lambda,\lambda}\) for the dynamics of \(\Phi_0+\lambda\cV_\Lambda\).
Its expression in product-vacuum coordinates is
\begin{equation}
 \tau_t^{\Lambda,\lambda}
 :=\sigma_\st\circ
  (\tau_t^{\ph,\Lambda,\lambda}\widehat\otimes\tau_t^{\rm aux})
       \circ\sigma_\st^{-1}.
 \label{eq:atomic-cutoff-dynamics}
\end{equation}

\begin{lemma}
\label{lem:operator-facts}
For every finite \(\Lambda\), the following hold.
\begin{enumerate}[label=(\alph*),leftmargin=*]
\item \(V_\Lambda\) is bounded and self-adjoint, and \(K_{\Lambda,\lambda}\) is self-adjoint on \(\operatorname{Dom}(K_0)\) for every \(\lambda\in\R\).
\item For every \(\lambda\in\R\), the dynamics \(\tau^{\Lambda,\lambda}\) is implemented by \(K_{\Lambda,\lambda}\) in the representation \(\pi_\at\):
  \begin{equation}
   \pi_\at(\tau_t^{\Lambda,\lambda}(A))
   =\E^{\I tK_{\Lambda,\lambda}}\pi_\at(A)\E^{-\I tK_{\Lambda,\lambda}},
   \qquad A\in\A^\st.
   \label{eq:algebra-GNS-dynamics}
  \end{equation}
\item \(K_{\Lambda,0}=K_0\) has the simple ground state eigenvalue zero and gap at least \(g\).
\end{enumerate}
\end{lemma}

\begin{proof}
Part~(a) follows from the bounded perturbation theorem.
The interaction dynamics \(\tau^{\ph,\Lambda,\lambda}\) agrees with the bounded perturbation of \(\tau^{\Phi_0}\) by \(\lambda V_\Lambda^\ph\), since their generators agree on the local algebra, a common core by Lemma~\ref{lem:local-form-core}.
The same core argument shows that \(\tau^{\ph,\Lambda,\lambda}\widehat\otimes\tau^{\rm aux}\) is the dynamics generated by the graded tensor sum of the interactions.
For part~(b), conjugate its bounded-perturbation integral equation by \(\sigma_\st\).
The unperturbed dynamics in these coordinates is implemented by \(K_0\), and the perturbation becomes \(\lambda\pi_\at\sigma_\st(V_\Lambda^\ph\widehat\otimes\1) =\lambda V_\Lambda\).
Uniqueness of the bounded-perturbation evolution \cite[Section~5.4]{BR97} proves \eqref{eq:algebra-GNS-dynamics}.
Part~(c) follows from \(K_0\Omega_\at=0\) and \(K_0\ge gN\ge g(\1-P^\vac)\).
\end{proof}

For fixed \(\Lambda\), let \(\mathcal O_\Lambda\subset\R\) be the set of \(\lambda\) for which \(E_{\Lambda,\lambda}:=\inf\spec K_{\Lambda,\lambda}\) is a simple eigenvalue and
\begin{equation}
 \dist\!\left(E_{\Lambda,\lambda},
 \spec(K_{\Lambda,\lambda})\setminus\{E_{\Lambda,\lambda}\}\right)
 >\frac g2.
 \label{eq:bootstrap-gap}
\end{equation}
Lemma~\ref{lem:operator-facts}(c) gives \(0\in\mathcal O_\Lambda\).
Norm-resolvent continuity and stability of Riesz projections under bounded perturbations imply that \(\mathcal O_\Lambda\) is open \cite[Chapter~IV]{Kato}.
Its connected component containing zero is therefore an open interval; we denote it by \(I_\Lambda\) and call it the \emph{connected bootstrap interval}.
For \(\lambda\in I_\Lambda\), let \(P_{\Lambda,\lambda}:=\1_{\{E_{\Lambda,\lambda}\}}(K_{\Lambda,\lambda})\) be the rank-one ground-state projection on \(\cH_\at\).
The eigenvalue \(E_{\Lambda,\lambda}\) and projection \(P_{\Lambda,\lambda}\) depend real-analytically on \(\lambda\in I_\Lambda\), by analytic perturbation theory for the bounded perturbation \(K_0+\lambda V_\Lambda\) \cite[Chapter~VII]{Kato}.

\subsection{The physical spectral flow}
\label{sec:compatible-filter-pair}

Quasi-adiabatic continuation was introduced by Hastings \cite{Hastings04} and developed by Hastings and Wen \cite{HastingsWen05}.
We use the exact spectral-flow construction of Bachmann, Michalakis, Nachtergaele, and Sims \cite{BMNS12}, with the physical generator defined by the inverse Liouvillian of \cite[Appendix~B]{BTW26}.
A Gevrey filter gives the stretched-exponential time decay needed for the locality estimates; Gaussian filters, as in \cite{BachmannDuFraasWessel25}, give faster time decay at the expense of approximate rather than exact spectral separation.

Fix the Fourier convention
\begin{equation*}
 \widehat f(\xi):=\int_\R \E^{\I t\xi}f(t)\,\dd t.
\end{equation*}
For every \(0<\theta<1\), choose a standard real even Gevrey profile \(\widehat\chi_\theta\in G_c^{1/\theta}(\R)\) with \(\widehat\chi_\theta=1\) on \([-1/4,1/4]\) and \(\supp(\widehat\chi_\theta)\subset(-1/2,1/2)\).
Here the Gevrey condition means that, for suitable \(A,C<\infty\), \(\|\widehat\chi_\theta^{(n)}\|_\infty \le CA^n(n!)^{1/\theta}\) for every \(n\).
This choice depends on \(\theta\) but not on the model.
Given such a \(\theta\), set \(\widehat w(\xi):=\widehat\chi_\theta(\xi/g)\).
The locality estimates below hold for every \(0<\theta<1\).
For the gap bootstrap we choose \(\theta=1/2\) in Section~\ref{sec:form-bound}; Section~\ref{sec:limit} then constructs transport at every exponent \(0<\theta<1\) on the resulting interval.
The polynomial proof also uses \(\theta=1/2\).
The inverse Fourier transform \(w\) is real and even, and
\begin{equation}
 \widehat w(\xi)=1\quad (|\xi|\le g/4),
 \qquad
 \supp(\widehat w)\subset(-g/2,g/2).
 \label{eq:low-pass-filter}
\end{equation}
Define the real odd function \(W\) from \(w\) by
\begin{equation}
 W(t):=
 \begin{cases}
  -\displaystyle\int_t^\infty w(u)\,\dd u,&t>0,\\[4pt]
  0,&t=0,\\[2pt]
  \displaystyle\int_{-\infty}^t w(u)\,\dd u,&t<0.
 \end{cases}
 \label{eq:inverse-filter-kernel}
\end{equation}
Then, distributionally, \(W'=w-\delta_0\), and hence
\begin{equation}
 \widehat W(\xi)=\frac{1-\widehat w(\xi)}{\I\xi}.
 \label{eq:compatible-filter-identity}
\end{equation}
The quotient is continuously extended by zero at \(\xi=0\).
In particular, \(\widehat W(\xi)=-\I/\xi\) for \(|\xi|\ge g/2\).
The construction of compactly supported Gevrey cutoffs of order \(1/\theta>1\) and the Gevrey Paley--Wiener estimate are given in \cite[Chapter~1]{Rodino93}.
Applied to \(\widehat w\), the Paley--Wiener estimate gives constants \(C_w,c_w>0\) such that
\begin{equation}
 |w(t)|\le C_w\E^{-c_w|t|^\theta}.
 \label{eq:w-stretched-time}
\end{equation}
The associated function \(W\) has the same type of decay, with a possible loss of rate:
\begin{equation}
 |W(t)|\le C_W\E^{-c_W|t|^\theta}
 \label{eq:W-stretched-time}
\end{equation}
for some \(C_W,c_W>0\).
Indeed, for \(t>0\),
\begin{equation*}
 |W(t)|\le C_w\int_t^\infty \E^{-c_wu^\theta}\,\dd u
 \le C(1+t)^{1-\theta}\E^{-c_wt^\theta}
 \le C_W\E^{-c_Wt^\theta},
\end{equation*}
where \(0<c_W<c_w\); the case \(t<0\) follows by oddness.
The decay bounds \eqref{eq:w-stretched-time}--\eqref{eq:W-stretched-time} imply that both \(w\) and \(W\) have every absolute polynomial moment.
On the physical algebra, define
\begin{equation}
 A^{\D_{\Lambda,\lambda}^\ph}
 :=\int_\R w(t)\tau_t^{\ph,\Lambda,\lambda}(A)\,\dd t,
 \qquad
 \cI_{\Lambda,\lambda}^\ph(A)
 :=\int_\R W(t)\tau_t^{\ph,\Lambda,\lambda}(A)\,\dd t.
 \label{eq:physical-inverse-filter}
\end{equation}
The bounded physical generator
\begin{equation}
 \mathscr S_{\Lambda,\lambda}^\ph
 :=\cI_{\Lambda,\lambda}^\ph(V_\Lambda^\ph)
 \label{eq:physical-cutoff-generator}
\end{equation}
is even, self-adjoint, and norm-continuous in \(\lambda\).
Indeed, bounded-perturbation comparison and \(\int|tW(t)|\,\dd t<\infty\) give norm continuity.
Let \(u_{\Lambda,\lambda}\) solve \(\partial_{\lambda}u_{\Lambda,\lambda}=-\I\mathscr S_{\Lambda,\lambda}^\ph u_{\Lambda,\lambda}\), \(u_{\Lambda,0}=\1\), and define \(\beta_{\Lambda,0,\lambda}(A):=u_{\Lambda,\lambda}^*Au_{\Lambda,\lambda}\).
The unitary \(u_{\Lambda,\lambda}\) is even, since its generator is even.
Its two-parameter extension follows the convention \eqref{eq:path-cocycle-convention}.
In the fixed GNS representation, set
\begin{align}
 \mathscr S_{\Lambda,\lambda}
 &:=\pi_\at\sigma_\st(\mathscr S_{\Lambda,\lambda}^\ph\widehat\otimes\1)
   =\int_\R W(t)\E^{\I tK_{\Lambda,\lambda}}V_\Lambda
                    \E^{-\I tK_{\Lambda,\lambda}}\,\dd t,
 \label{eq:spectral-flow-generator}
\\
 U_{\Lambda,\lambda}
 &:=\pi_\at\sigma_\st(u_{\Lambda,\lambda}\widehat\otimes\1),
 \qquad
 \partial_{\lambda}U_{\Lambda,\lambda}=-\I\mathscr S_{\Lambda,\lambda}U_{\Lambda,\lambda},
 \quad U_{\Lambda,0}=\1.
 \label{eq:U-flow}
\end{align}
The integral identity follows from \eqref{eq:algebra-GNS-dynamics}.
Both flows are defined for every real \(\lambda\), independently of the bootstrap gap.

For \(\lambda\in I_\Lambda\), exact quasi-adiabatic continuation gives
\begin{equation}
 U_{\Lambda,\lambda}^*P_{\Lambda,\lambda}U_{\Lambda,\lambda}
 =P^\vac=|\Omega_\at\rangle\langle\Omega_\at|.
 \label{eq:projection-transport}
\end{equation}
The identity \eqref{eq:projection-transport} is proved in the following proposition.

Define the automorphisms
\begin{equation}
 \sigma_{\Lambda,\lambda}
 :=\sigma_\st\circ
       (\beta_{\Lambda,0,\lambda}\widehat\otimes\id_{\rm aux}),
 \qquad
 \alpha_{\Lambda,\lambda,\lambda'}
 :=\sigma_\st\circ
       (\beta_{\Lambda,\lambda,\lambda'}\widehat\otimes\id_{\rm aux})
           \circ\sigma_\st^{-1},
 \label{eq:cutoff-flow-descent}
\end{equation}
and write \(\alpha_{\Lambda,\lambda}:=\alpha_{\Lambda,0,\lambda}\).
Thus \(\sigma_{\Lambda,\lambda}=\alpha_{\Lambda,\lambda}\circ\sigma_\st\).

\begin{proposition}
\label{prop:bootstrap-flow}
For every finite \(\Lambda\), \(\alpha_{\Lambda,\lambda}\) is a parity-preserving inner automorphism satisfying
\begin{equation}
 \pi_\at(\alpha_{\Lambda,\lambda}(A))
 =U_{\Lambda,\lambda}^*\pi_\at(A)U_{\Lambda,\lambda},
 \qquad A\in\A^\st.
 \label{eq:flow-automorphism}
\end{equation}
For \(\lambda\in I_\Lambda\), the ground-state transport identity \eqref{eq:projection-transport} holds.
\end{proposition}

\begin{proof}
The bounded-perturbation family \(K_0+\lambda V_\Lambda\) has a simple ground-state eigenvalue with gap greater than \(g/2\) on \(I_\Lambda\).
We first prove \eqref{eq:projection-transport} in the GNS representation and follow the standard proof strategy, e.g.\ \cite[Proposition~2.4]{BMNS12} and \cite[Theorem~6.3]{NSY19}.
Fix \(\lambda\in I_\Lambda\), write \(K,E,P,\mathscr S\) for \(K_{\Lambda,\lambda},E_{\Lambda,\lambda},P_{\Lambda,\lambda},\mathscr S_{\Lambda,\lambda}\), and put \(P^\perp:=\1-P\).
By the gap assumption, \(K-E\) has spectrum in \((g/2,\infty)\) on \(\ran P^\perp\) and \((K-E)^{-1}\) denotes its inverse there.

Let \(\cC\) be the positively oriented circle of radius \(g/4\) about \(E\).
Differentiating the Riesz formula \(P_{\Lambda,\lambda}=-(2\pi\I)^{-1}\oint_\cC(K_{\Lambda,\lambda}-z)^{-1}\,\dd z\) gives \(\partial_\lambda P_{\Lambda,\lambda}=(2\pi\I)^{-1}\oint_\cC(K-z)^{-1}V_\Lambda(K-z)^{-1}\,\dd z\).
Since \((K-z)^{-1}P=(E-z)^{-1}P\) and \(z\mapsto P^\perp(K-z)^{-1}\) is analytic inside \(\cC\), the residue at \(z=E\) gives
\[
 P^\perp(\partial_\lambda P_{\Lambda,\lambda})P=-(K-E)^{-1}P^\perp V_\Lambda P .
\]
On the other hand, \(\E^{-\I tK}P=\E^{-\I tE}P\), so \eqref{eq:spectral-flow-generator} and the spectral theorem give \(P^\perp\mathscr SP=\widehat W(K-E)P^\perp V_\Lambda P\).
Since \(\widehat W(\xi)=-\I/\xi\) for \(\xi\ge g/2\),
\[
 P^\perp(-\I[\mathscr S,P])P=-\I P^\perp\mathscr SP=-(K-E)^{-1}P^\perp V_\Lambda P .
\]

Both \(\partial_\lambda P_{\Lambda,\lambda}\) and \(-\I[\mathscr S,P]\) are self-adjoint, so the two identities also match their blocks \(P(\cdot)P^\perp\).
Their diagonal blocks vanish, for \(\partial_\lambda P_{\Lambda,\lambda}\) because \(P^2=P\).
Hence \(\partial_\lambda P_{\Lambda,\lambda}=-\I[\mathscr S_{\Lambda,\lambda},P_{\Lambda,\lambda}]\), and \eqref{eq:U-flow} gives
\[
 \partial_\lambda\bigl(U_{\Lambda,\lambda}^*P_{\Lambda,\lambda}U_{\Lambda,\lambda}\bigr)
 =U_{\Lambda,\lambda}^*\bigl(\partial_\lambda P_{\Lambda,\lambda}+\I[\mathscr S_{\Lambda,\lambda},P_{\Lambda,\lambda}]\bigr)U_{\Lambda,\lambda}=0,
 \qquad\lambda\in I_\Lambda .
\]
Since \(I_\Lambda\) is an interval containing \(0\), \(U_{\Lambda,0}=\1\), and \(P_{\Lambda,0}=P^\vac\) because \(K_0\Omega_\at=0\) and \(K_0\ge gN\ge g(\1-P^\vac)\), this proves \eqref{eq:projection-transport}.
By \eqref{eq:cutoff-flow-descent}, \(\alpha_{\Lambda,\lambda}\) is conjugation by the even unitary \(\sigma_\st(u_{\Lambda,\lambda}\widehat\otimes\1)\).
Applying \(\pi_\at\) and using \eqref{eq:U-flow} gives \eqref{eq:flow-automorphism}.
\end{proof}

For \(A\in\A^\st\), define the filtered observable in product-vacuum coordinates by
\begin{equation}
 A^{\D_{\Lambda,\lambda}}
 :=\int_\R w(t)\tau_t^{\Lambda,\lambda}(A)\,\dd t.
 \label{eq:diagonal-filter}
\end{equation}
If \(\lambda\in I_\Lambda\) and \(P_{\Lambda,\lambda}^\perp=\1-P_{\Lambda,\lambda}\), then the support condition \eqref{eq:low-pass-filter} and \(\widehat w(0)=1\) give
\begin{align}
 P_{\Lambda,\lambda}\pi_\at(A^{\D_{\Lambda,\lambda}})P_{\Lambda,\lambda}^\perp
 &=P_{\Lambda,\lambda}^\perp\pi_\at(A^{\D_{\Lambda,\lambda}})P_{\Lambda,\lambda}
 =0,\nonumber\\
 P_{\Lambda,\lambda}\pi_\at(A^{\D_{\Lambda,\lambda}})P_{\Lambda,\lambda}
 &=P_{\Lambda,\lambda}\pi_\at(A)P_{\Lambda,\lambda}.
 \label{eq:D-ground-blocks}
\end{align}
These identities follow from \(\widehat w(K_{\Lambda,\lambda}-E_{\Lambda,\lambda})=P_{\Lambda,\lambda}\) and \eqref{eq:algebra-GNS-dynamics}; see \cite[Lemmas~3.1--3.2]{MoonOgata20} and \cite[proof of Proposition~4.1]{BTW26}.
The excited--excited block may also change under this filter.

\section{The vacuum normal form}
\label{sec:normal-form}

\subsection{Differentiating the moving-frame Hamiltonian}

For \(\lambda\in I_\Lambda\), define the centered moving-frame Hamiltonian
\begin{equation}
 K_{\Lambda,\lambda}^\nf
 :=U_{\Lambda,\lambda}^*
 (K_{\Lambda,\lambda}-E_{\Lambda,\lambda})U_{\Lambda,\lambda}.
 \label{eq:K-nf}
\end{equation}
In this subsection we express the \(\lambda\)-derivative of \(K_{\Lambda,\lambda}^\nf\) as a sum of centered filtered perturbation terms; see Proposition~\ref{prop:centered-derivative}.
Let
\begin{equation*}
 \widehat\omega_{\Lambda,\lambda}(A)
 :=\langle U_{\Lambda,\lambda}\Omega_\at,
 \pi_\at(A)U_{\Lambda,\lambda}\Omega_\at\rangle
\end{equation*}
and set \(\omega_{\Lambda,\lambda}:=\omega_0\circ\beta_{\Lambda,0,\lambda}\).
By \eqref{eq:flow-automorphism} and \(\omega_\at\circ\sigma_\st=\omega_0\widehat\otimes\omega_{\rm aux}\),
\begin{equation}
 \widehat\omega_{\Lambda,\lambda}\circ\sigma_\st
 =\omega_{\Lambda,\lambda}\widehat\otimes\omega_{\rm aux}.
 \label{eq:cutoff-state-factorization}
\end{equation}
For \(\lambda\in I_\Lambda\), \eqref{eq:projection-transport} shows that \(\widehat\omega_{\Lambda,\lambda}\) is the ground-state vector state of \(K_{\Lambda,\lambda}\).

\begin{figure}[htbp]
\centering
\begin{tikzcd}[column sep=3.3em,row sep=2.8em,panel label={(a) observables}]
 \A\widehat\otimes\A_{\rm aux} \arrow[r,"\sigma_\st"]
   \arrow[d,"{\beta_{\Lambda,\lambda,\lambda'}\widehat\otimes\id_{\rm aux}}"']
 & \A^\st \arrow[d,"{\alpha_{\Lambda,\lambda,\lambda'}}"]\\
 \A\widehat\otimes\A_{\rm aux} \arrow[r,"\sigma_\st"'] & \A^\st
\end{tikzcd}%
\hspace{2.6em}%
\begin{tikzcd}[column sep=3.3em,row sep=2.8em,panel label={(b) states}]
 \omega_0\widehat\otimes\omega_{\rm aux} \arrow[r,mapsto,"\circ\,\sigma_\st^{-1}"]
   \arrow[d,mapsto,"\circ\,(\beta_{\Lambda,0,\lambda}\widehat\otimes\id_{\rm aux})"']
 & \omega_\at \arrow[d,mapsto,"\circ\,\alpha_{\Lambda,\lambda}"]\\
 \omega_{\Lambda,\lambda}\widehat\otimes\omega_{\rm aux} \arrow[r,mapsto,"\circ\,\sigma_\st^{-1}"']
 & \widehat\omega_{\Lambda,\lambda}
\end{tikzcd}
\caption{The cutoff flow on observables (a) and states (b).
In each square, the left column uses physical coordinates and the right column product-vacuum coordinates.
In (b), \(\circ\,\gamma\) denotes \(\rho\mapsto\rho\circ\gamma\), and \(\widehat\omega_{\Lambda,\lambda}\) is the vector state of \(U_{\Lambda,\lambda}\Omega_\at\).
The diagrams summarize \eqref{eq:cutoff-flow-descent} and \eqref{eq:cutoff-state-factorization}.}
\label{fig:physical-flow-descent}
\end{figure}
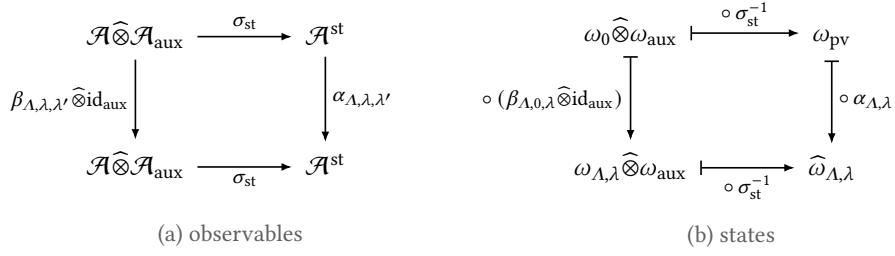

For fixed \(\Lambda\), \(U_{\Lambda,\lambda}\) and \(U_{\Lambda,\lambda}^*\) are norm continuously differentiable, and \(E_{\Lambda,\lambda}\) is continuously differentiable on \(I_\Lambda\).
The difficulty in differentiating \(K_{\Lambda,\lambda}^\nf\) in \eqref{eq:K-nf} is the unbounded factor \(K_{\Lambda,\lambda}=K_0+\lambda V_\Lambda\).
For \(\psi\in\operatorname{Dom}(K_0)\), norm differentiability of \(U_{\Lambda,\lambda}\) alone does not justify
\begin{equation*}
 \frac{\dd}{\dd \lambda}\bigl(K_{\Lambda,\lambda}U_{\Lambda,\lambda}\psi\bigr)
 =V_\Lambda U_{\Lambda,\lambda}\psi
  +K_{\Lambda,\lambda}\dot U_{\Lambda,\lambda}\psi.
\end{equation*}
We must show that \(U_{\Lambda,\lambda}\psi\) stays in \(\operatorname{Dom}(K_0)\) and is differentiable in the graph norm \(\|\psi\|+\|K_0\psi\|\).
In particular, \(\dot U_{\Lambda,\lambda}\psi\) then also belongs to this domain.

The inverse-filter identity \eqref{eq:abstract-filtered-commutator} gives a bounded extension of \([K_{\Lambda,\lambda},\mathscr S_{\Lambda,\lambda}]\); compare \cite[Proposition~6.9]{NSY19} and \cite[Lemma~B.3 and Theorem~3.2]{BTW26}.
Lemma~\ref{lem:filtered-commutator} uses this identity and \(K_{\Lambda,\lambda}=K_0+\lambda V_\Lambda\) to prove norm continuity of \(\mathscr S_{\Lambda,\lambda}\) and the bounded commutator \([K_0,\mathscr S_{\Lambda,\lambda}]\).
Lemma~\ref{lem:nonautonomous-domain-invariance} then gives the required graph-norm differentiability.

For a bounded operator \(B\), we write \(B\in C^1(K)\) when \(t\mapsto\E^{\I tK}B\E^{-\I tK}\) is strongly continuously differentiable.
Equivalently, the commutator form \(\langle K\varphi,B\psi\rangle- \langle B^*\varphi,K\psi\rangle\), initially defined for \(\varphi,\psi\in\operatorname{Dom}(K)\), extends to a bounded form.
In particular, \(B\in C^1(K)\) implies \(B\operatorname{Dom}(K)\subset\operatorname{Dom}(K)\); see \cite[Section~6.2]{ABG96}.

\begin{samepage}
\begin{lemma}
\label{lem:abstract-filtered-commutator}
Let \(K=K^*\) on a Hilbert space, let \(V=V^*\) be bounded, and assume that \(t\mapsto\tau_t(V):=\E^{\I tK}V\E^{-\I tK}\) is norm-continuous.
Let \(W\in L^1(\R)\cap BV(\R)\) be real and suppose that its distributional derivative is the finite signed measure
\begin{equation*}
 \dd W=w(t)\,\dd t-\delta_0,
 \qquad w\in L^1(\R).
\end{equation*}
Define the bounded operators
\begin{equation*}
 \cI_W^K(V):=\int_\R W(t)\tau_t(V)\,\dd t,
 \qquad
 V^{\D_w^K}:=\int_\R w(t)\tau_t(V)\,\dd t
\end{equation*}
by Bochner integration.
Then \(\cI_W^K(V)=\cI_W^K(V)^*\) and \(\cI_W^K(V)\in C^1(K)\).
More precisely,
\begin{equation}
 \I[K,\cI_W^K(V)]=V-V^{\D_w^K},
 \qquad
 V+\I[\cI_W^K(V),K]=V^{\D_w^K}.
 \label{eq:abstract-filtered-commutator}
\end{equation}
\end{lemma}
\end{samepage}

\begin{proof}
Transferring the time shift to \(W\) and using its distributional derivative gives
\begin{align*}
 \frac{\tau_\varepsilon(\cI_W^K(V))-\cI_W^K(V)}{\varepsilon}
 =-\int_\R\frac1\varepsilon\int_0^\varepsilon\tau_{t+u}(V)\,\dd u\,\dd W(t) \stackrel{\varepsilon\to0}{\longrightarrow}
 -\int_\R\tau_t(V)\,\dd W(t)=V-V^{\D_w^K}.
\end{align*}
The convergence is in norm because the isometric orbit \(t\mapsto\tau_t(V)\) is uniformly norm-continuous and \(|\dd W|\) is finite.
The bounded-commutator criterion then gives \(\cI_W^K(V)\in C^1(K)\) and \eqref{eq:abstract-filtered-commutator}; compare \cite[Proposition~6.9]{NSY19}.
Since \(W\) is real and \(V\) is self-adjoint, \(\cI_W^K(V)\) is self-adjoint.
\end{proof}

To control \(U_{\Lambda,\lambda}\) on \(\operatorname{Dom}(K_0)\), we will use the following graph-norm evolution lemma with \(H=K_0\) and generator \(\mathscr S_{\Lambda,\lambda}\).

\begin{lemma}
\label{lem:nonautonomous-domain-invariance}
Let \(H=H^*\), let \(J\subset\R\) be an interval, and suppose \(s\mapsto \mathscr S_s=\mathscr S_s^*\) is norm-continuous, \(\mathscr S_s\in C^1(H)\), and \(s\mapsto[H,\mathscr S_s]^\circ\) is norm-continuous, where the superscript denotes the bounded extension.
If
\begin{equation*}
 \partial_tU_{s,t}=-\I \mathscr S_tU_{s,t},
 \qquad U_{s,s}=\1,
\end{equation*}
then
\begin{equation*}
 U_{s,t}\operatorname{Dom}(H)=\operatorname{Dom}(H),
 \qquad s,t\in J,
\end{equation*}
and, for every \(\psi\in\operatorname{Dom}(H)\), the map \(t\mapsto U_{s,t}\psi\) is continuously differentiable in the graph norm of \(H\), with derivative
\begin{equation*}
 \frac{\dd}{\dd t}U_{s,t}\psi=-\I \mathscr S_tU_{s,t}\psi.
\end{equation*}
\end{lemma}

\begin{proof}
On \(\operatorname{Dom}(H)\) with its graph norm, the identity \(H\mathscr S_s\psi=\mathscr S_sH\psi+[H,\mathscr S_s]^\circ\psi\) and the continuity assumptions make \(s\mapsto\mathscr S_s\) a norm-continuous family of bounded operators.
The Banach-space evolution theorem \cite[Chapter~5, Theorem~5.1, p.~127]{Pazy83} gives a graph-norm continuously differentiable propagator satisfying
\begin{equation}
 U_{s,t}=\1-\I\int_s^t\mathscr S_uU_{s,u}\,\dd u.
 \label{eq:bounded-propagator-integral}
\end{equation}
Uniqueness of the Hilbert-space evolution shows that this propagator equals \(U_{s,t}\), so \(U_{s,t}\operatorname{Dom}(H)\subseteq\operatorname{Dom}(H)\).
Applying the same argument to \(U_{t,s}=U_{s,t}^{-1}\) gives the reverse inclusion.
\end{proof}

For every real \(\lambda\) and nonempty \(Z\Subset\Gamma\), write
\[
 F_{\Lambda,\lambda,Z}:=\cV_\Lambda(Z)^{\D_{\Lambda,\lambda}^\ph}
\]
for the filtered physical perturbation term.

\begin{lemma}
\label{lem:filtered-commutator}
For every finite \(\Lambda\) and \(\lambda\in\R\), \(\mathscr S_{\Lambda,\lambda}\in C^1(K_{\Lambda,\lambda})\cap C^1(K_0)\), and
\begin{equation}
 V_\Lambda+\I[\mathscr S_{\Lambda,\lambda},K_{\Lambda,\lambda}]
 =\sum_{Z\subset\Lambda}\pi_\at\!\left(
 \sigma_\st(F_{\Lambda,\lambda,Z}\widehat\otimes\1)\right).
 \label{eq:filtered-commutator}
\end{equation}
The commutator in \eqref{eq:filtered-commutator} denotes its bounded extension.
For fixed \(\Lambda\), the maps \(\lambda\mapsto\mathscr S_{\Lambda,\lambda}\) and \(\lambda\mapsto[K_0,\mathscr S_{\Lambda,\lambda}]^\circ\) are norm-continuous, where the superscript again denotes the bounded extension.
Moreover,
\[
 U_{\Lambda,\lambda}\operatorname{Dom}(K_0)
 =\operatorname{Dom}(K_0),\qquad \lambda\in\R,
\]
and, for every \(\psi\in\operatorname{Dom}(K_0)\), the map \(\lambda\mapsto U_{\Lambda,\lambda}\psi\) is continuously differentiable in the graph norm of \(K_0\).
\end{lemma}

\begin{proof}
The covariance identity \eqref{eq:algebra-GNS-dynamics} makes \(t\mapsto \E^{\I tK_{\Lambda,\lambda}}V_\Lambda \E^{-\I tK_{\Lambda,\lambda}}\) norm-continuous: \(V_\Lambda\) is a represented algebra element and the algebraic dynamics is pointwise norm-continuous.
The function \(W\) in \eqref{eq:inverse-filter-kernel} belongs to \(L^1\cap BV\) and has distributional derivative \(w(t)\,\dd t-\delta_0\).
Lemma~\ref{lem:abstract-filtered-commutator}, with \(K=K_{\Lambda,\lambda}\) and \(V=V_\Lambda\), therefore gives
\begin{equation*}
 V_\Lambda+\I[\mathscr S_{\Lambda,\lambda},K_{\Lambda,\lambda}]
 =\int_\R w(t)\E^{\I tK_{\Lambda,\lambda}}V_\Lambda
 \E^{-\I tK_{\Lambda,\lambda}}\,\dd t.
\end{equation*}
Expanding \(V_\Lambda\) by \eqref{eq:V-Lambda} under this integral and using \eqref{eq:atomic-cutoff-dynamics} gives \eqref{eq:filtered-commutator}.

Since \(K_{\Lambda,\lambda}=K_0+\lambda V_\Lambda\), on \(\operatorname{Dom}(K_0)\) we have
\begin{equation*}
 [\mathscr S_{\Lambda,\lambda},K_0]
 =[\mathscr S_{\Lambda,\lambda},K_{\Lambda,\lambda}]
 -\lambda[\mathscr S_{\Lambda,\lambda},V_\Lambda].
\end{equation*}
Both commutators on the right have bounded extensions: the first by Lemma~\ref{lem:abstract-filtered-commutator}, and the second because \(\mathscr S_{\Lambda,\lambda}\) and \(V_\Lambda\) are bounded.
Hence \(\mathscr S_{\Lambda,\lambda}\in C^1(K_0)\).
At fixed \(\Lambda\), the bounded-perturbation Duhamel formula \cite[Section~2.3.2]{NSY19} gives
\begin{equation*}
 \|\tau_t^{\Lambda,\lambda}(A)-\tau_t^{\Lambda,\lambda'}(A)\|
 \le2|\lambda-\lambda'|\,|t|\,\|V_\Lambda\|\,\|A\|.
\end{equation*}
Norm continuity of \(\mathscr S_{\Lambda,\lambda}\) follows from \eqref{eq:spectral-flow-generator} and norm continuity of the physical generator \eqref{eq:physical-cutoff-generator}.
Integrating the Duhamel bound against \(|w(t)|\,\dd t\) proves norm continuity of the right side of \eqref{eq:filtered-commutator}.
Equation~\eqref{eq:filtered-commutator} therefore gives norm continuity of \([\mathscr S_{\Lambda,\lambda},K_{\Lambda,\lambda}]^\circ\), and the displayed commutator identity gives norm continuity of \([K_0,\mathscr S_{\Lambda,\lambda}]^\circ\).
Lemma~\ref{lem:nonautonomous-domain-invariance}, with \(H=K_0\) and generator \(\mathscr S_{\Lambda,\lambda}\), now proves the domain equality and graph-norm differentiability.
\end{proof}

For every real \(\lambda\) and nonempty \(Z\Subset\Gamma\), define the centered transported term
\begin{equation}
 G_{\Lambda,\lambda}(Z)
 :=\sigma_{\Lambda,\lambda}(F_{\Lambda,\lambda,Z}\widehat\otimes\1)
       -\omega_{\Lambda,\lambda}(F_{\Lambda,\lambda,Z})\1.
 \label{eq:G-Lambda-s}
\end{equation}
Set \(G_{\Lambda,\lambda}(\varnothing):=0\).
Integrating the change-of-coordinates identity \eqref{eq:atomic-cutoff-dynamics} against the filter \(w\) gives
\[
 \bigl(\sigma_\st(\cV_\Lambda(Z)\widehat\otimes\1)\bigr)^{\D_{\Lambda,\lambda}}
 =\sigma_\st(F_{\Lambda,\lambda,Z}\widehat\otimes\1).
\]
Thus the ground-state block identities \eqref{eq:D-ground-blocks} apply to these physical filtered terms in the fixed GNS representation.
Moreover, \(\omega_\at(G_{\Lambda,\lambda}(Z))=0\) for every real \(\lambda\).

Figure~\ref{fig:filtering-centering} illustrates the filtering, scalar subtraction, and transport used in the following proof.

\begin{samepage}
\begin{lemma}
\label{lem:termwise-block}
For every \(\lambda\in I_\Lambda\) and \(Z\Subset\Gamma\),
\begin{equation}
 \pi_\at(G_{\Lambda,\lambda}(Z))\Omega_\at=0,
 \qquad
 \pi_\at(G_{\Lambda,\lambda}(Z)^*)\Omega_\at=0.
 \label{eq:G-vacuum-annihilation}
\end{equation}
\end{lemma}

\begin{proof}
Since \(P_{\Lambda,\lambda}\) has rank one, applying \eqref{eq:D-ground-blocks} with \(A=\sigma_\st(\cV_\Lambda(Z)\widehat\otimes\1)\) gives
\begin{equation*}
 P_{\Lambda,\lambda}\pi_\at\sigma_\st(F_{\Lambda,\lambda,Z}\widehat\otimes\1)P_{\Lambda,\lambda}
 =\omega_{\Lambda,\lambda}(F_{\Lambda,\lambda,Z})P_{\Lambda,\lambda}
 =\omega_{\Lambda,\lambda}(\cV_\Lambda(Z))P_{\Lambda,\lambda}.
\end{equation*}
Here we used \eqref{eq:cutoff-state-factorization}.
Together with the vanishing off-diagonal blocks, this shows that
\[
 \pi_\at\sigma_\st(F_{\Lambda,\lambda,Z}\widehat\otimes\1)
       -\omega_{\Lambda,\lambda}(F_{\Lambda,\lambda,Z})\1
\]
has zero product with \(P_{\Lambda,\lambda}\) on either side.
Conjugation by \(U_{\Lambda,\lambda}^*\), using \eqref{eq:projection-transport} and \(\sigma_{\Lambda,\lambda}=\alpha_{\Lambda,\lambda}\circ\sigma_\st\), gives \eqref{eq:G-vacuum-annihilation}.
\end{proof}
\end{samepage}

\begin{figure}[htbp]
\centering
\begin{tikzpicture}[stability diagram]
 \fill[framecol,rounded corners=4pt] (-1.5,-2.25) rectangle (8.55,1.5);
 \node[fig note] at (3.525,-1.9)
   {$\ran P_{\Lambda,\lambda}\oplus\ran P_{\Lambda,\lambda}^\perp$};
 \fill[framecol,rounded corners=4pt] (9.9,-2.25) rectangle (13.3,1.5);
 \node[fig note] at (11.6,-1.9)
   {$\C\Omega_\at\oplus\Omega_\at^\perp$};
 \blockmat{0}{gsblk!35}{odblk!35}{odblk!35}{exblk!35}
   {$b$}{$v^*$}{$v$}{$C$}
 \blockmat{3.45}{gsblk!35}{blkzero}{blkzero}{exblk!35}
   {$b$}{\blkz}{\blkz}{$C_\D$}
 \blockmat{6.9}{blkzero}{blkzero}{blkzero}{exblk!35}
   {\blkz}{\blkz}{\blkz}{$C_\D-b\1$}
 \blockmat{11.6}{blkzero}{blkzero}{blkzero}{exblk!35}
   {\blkz}{\blkz}{\blkz}{$\ast$}
 \foreach \x/\t in {0/perturbation term,3.45/filtered term,
                    6.9/centered term,11.6/transported term}
   \node[fig title] at (\x,1.17) {\t};
 \node[font=\footnotesize] at (0,-1.3)
   {$\pi_\at\sigma_\st(\cV_\Lambda(Z)\widehat\otimes\1)$};
 \node[font=\footnotesize] at (3.45,-1.3)
   {$\pi_\at\sigma_\st(F_{\Lambda,\lambda,Z}\widehat\otimes\1)$};
 \node[font=\footnotesize] at (6.9,-1.3)
   {$\pi_\at\sigma_\st(F_{\Lambda,\lambda,Z}\widehat\otimes\1)-b\1$};
 \node[font=\footnotesize] at (11.6,-1.3) {$\pi_\at(G_{\Lambda,\lambda}(Z))$};
 \draw[stability arrow] (.98,0) -- (2.47,0)
   node[midway,above=1pt,font=\footnotesize] {filter}
   node[midway,below=1pt,font=\footnotesize] {$\D_{\Lambda,\lambda}$};
 \draw[stability arrow] (4.43,0) -- (5.92,0)
   node[midway,above=1pt,font=\footnotesize] {center}
   node[midway,below=1pt,font=\footnotesize] {$-\,b\1$};
 \draw[stability arrow] (7.88,0) -- (10.62,0)
   node[midway,above=1pt,font=\footnotesize] {transport}
   node[midway,below=1pt,font=\footnotesize]
     {$U_{\Lambda,\lambda}^*(\,\cdot\,)\,U_{\Lambda,\lambda}$};
\end{tikzpicture}
\caption{Filtering, centering, and transport of a single perturbation term, for \(\lambda\in I_\Lambda\).
In the first three panels, block matrices are taken with respect to the one-dimensional moving ground space \(\ran P_{\Lambda,\lambda}\) and its orthogonal complement; in the last panel, with respect to the fixed vacuum \(\C\Omega_\at\) and its orthogonal complement.
The filter \(A\mapsto\int_\R w(t)\tau_t^{\Lambda,\lambda}(A)\,\dd t\) removes the ground--excited blocks and preserves the ground-state expectation \(b=\omega_{\Lambda,\lambda}(\cV_\Lambda(Z))\); \(C_{\D}\) denotes the resulting excited block.
Subtracting \(b\1\) removes the ground-state block.
Since \(U_{\Lambda,\lambda}^*P_{\Lambda,\lambda}U_{\Lambda,\lambda}=P^\vac\), the transported term \(\pi_\at(G_{\Lambda,\lambda}(Z))\) has zero product with \(P^\vac\) on either side, as proved in Lemma~\ref{lem:termwise-block}.}
\label{fig:filtering-centering}
\end{figure}
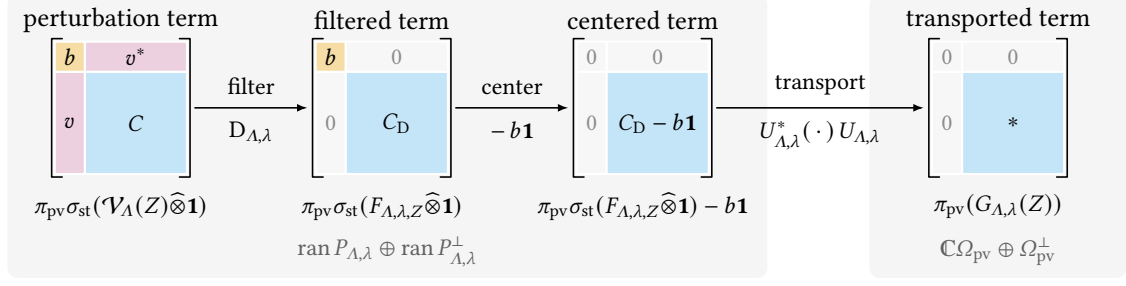

For every \(\lambda\in\R\), define the integrated centered term
\begin{equation}
 \cG_{\Lambda,\lambda}(Z):=\int_0^{\lambda}G_{\Lambda,u}(Z)\,\dd u.
 \label{eq:C-Lambda-s}
\end{equation}

\begin{samepage}
\begin{proposition}
\label{prop:centered-derivative}
For \(\lambda\in I_\Lambda\), the operator identity
\begin{equation}
 K_{\Lambda,\lambda}^\nf
 =K_0+\sum_{Z\subset\Lambda}\pi_\at(\cG_{\Lambda,\lambda}(Z))
 \label{eq:bounded-normal-form}
\end{equation}
holds on \(\operatorname{Dom}(K_0)\).
For every \(\psi\in\operatorname{Dom}(K_0)\), the map \(\lambda\mapsto K_{\Lambda,\lambda}^\nf\psi\) is continuously differentiable on \(I_\Lambda\), with
\begin{equation}
 \frac{\dd}{\dd\lambda}\bigl(K_{\Lambda,\lambda}^\nf\psi\bigr)
 =\sum_{Z\subset\Lambda}\pi_\at(G_{\Lambda,\lambda}(Z))\psi.
 \label{eq:centered-derivative}
\end{equation}
The correction in \eqref{eq:bounded-normal-form} is bounded.
Every \(\cG_{\Lambda,\lambda}(Z)\) is even and self-adjoint, and \(\pi_\at(\cG_{\Lambda,\lambda}(Z))\Omega_\at=0\).
Moreover,
\[
 \operatorname{Dom}\bigl((K_{\Lambda,\lambda}^\nf)^{1/2}\bigr)
 =\operatorname{Dom}(K_0^{1/2}),\qquad
 U_{\Lambda,\lambda}\operatorname{Dom}(K_0^{1/2})
 =\operatorname{Dom}(K_0^{1/2}),
\]
and \eqref{eq:bounded-normal-form} also holds as a quadratic-form identity on this domain.
\end{proposition}
\end{samepage}

\begin{proof}
Lemma~\ref{lem:filtered-commutator} gives continuous differentiability of \(\lambda\mapsto U_{\Lambda,\lambda}\psi\) in the \(K_0\)-graph norm for every \(\psi\in\operatorname{Dom}(K_0)\).
The affine family \(K_{\Lambda,\lambda}=K_0+\lambda V_\Lambda\) is continuously differentiable as a family of bounded maps from this graph-norm space into \(\cH_\at\).
Together with differentiability of \(U_{\Lambda,\lambda}^*\) and \(E_{\Lambda,\lambda}\), this justifies the product rule in Hilbert-space norm on the fixed domain \(\operatorname{Dom}(K_0)\).
Differentiating \eqref{eq:K-nf} and using \eqref{eq:U-flow} gives the first equality below; Lemma~\ref{lem:filtered-commutator} gives the second:
\begin{align*}
 \frac{\dd}{\dd \lambda}K_{\Lambda,\lambda}^\nf
  =U_{\Lambda,\lambda}^*
 \bigl(V_\Lambda+\I[\mathscr S_{\Lambda,\lambda},K_{\Lambda,\lambda}]
 -\dot E_{\Lambda,\lambda}\1\bigr)U_{\Lambda,\lambda} =\sum_{Z\subset\Lambda}
 \pi_\at\sigma_{\Lambda,\lambda}(F_{\Lambda,\lambda,Z}\widehat\otimes\1)
 -\dot E_{\Lambda,\lambda}\1.
\end{align*}
Differentiating the eigenvalue equation for the normalized ground state vector \(U_{\Lambda,\lambda}\Omega_\at\) gives the Hellmann--Feynman identity
\[
 \dot E_{\Lambda,\lambda}
 =\omega_{\Lambda,\lambda}(V_\Lambda^\ph)
 =\sum_{Z\subset\Lambda}\omega_{\Lambda,\lambda}(F_{\Lambda,\lambda,Z}).
\]
The last equality uses \eqref{eq:D-ground-blocks} and \eqref{eq:cutoff-state-factorization}.
Substituting the Hellmann--Feynman identity and \eqref{eq:G-Lambda-s} proves \eqref{eq:centered-derivative}.

For fixed \(\Lambda\), the Duhamel estimate in the proof of Lemma~\ref{lem:filtered-commutator} gives norm continuity of \(u\mapsto F_{\Lambda,u,Z}\).
Norm continuity of the cutoff flow in \eqref{eq:cutoff-flow-descent} and the definition \eqref{eq:G-Lambda-s} then give norm continuity of \(u\mapsto G_{\Lambda,u}(Z)\).
Moreover,
\[
 \|G_{\Lambda,u}(Z)\|\le2\|w\|_1 \|\cV_\Lambda(Z)\|.
\]
The norm continuity and this bound show that \eqref{eq:C-Lambda-s} defines a Bochner integral in \(\A^\st\) and that the finite sum on the right of \eqref{eq:centered-derivative} is norm-continuous.
Since \(K_{\Lambda,0}^\nf=K_0\), integration of \eqref{eq:centered-derivative} from \(0\) to \(\lambda\) proves \eqref{eq:bounded-normal-form}.
Evenness, self-adjointness, and vacuum annihilation follow from the definition and Lemma~\ref{lem:termwise-block}.
The correction in \eqref{eq:bounded-normal-form} is bounded and self-adjoint, so \(K_{\Lambda,\lambda}^\nf\) has form domain \(\operatorname{Dom}(K_0^{1/2})\).
The same holds for \(K_{\Lambda,\lambda}-E_{\Lambda,\lambda}\), which differs from \(K_0\) by a bounded operator.
Their unitary equivalence \eqref{eq:K-nf} therefore gives \(U_{\Lambda,\lambda}\operatorname{Dom}(K_0^{1/2})=\operatorname{Dom}(K_0^{1/2})\).
\end{proof}

\subsection{Localization by vacuum conditional expectations}

To localize the centered terms \(\cG_{\Lambda,\lambda}(Z)\), we require a map that preserves two-sided vacuum annihilation, a property the normalized-trace conditional expectation does not have.
We therefore use Araki's conditional expectation for an even product state \cite[Theorem~3.1]{Araki04}.

\begin{lemma}
\label{lem:vacuum-conditional-expectation}
Let \(\omega_\at\) be the graded product of even pure onsite states, and let \(X\Subset\Gamma\).
Under the graded identification \(\A^\st\cong\A_X^\st\widehat\otimes\A_{X^c}^\st\), the map \(A\widehat\otimes B\mapsto A\omega_\at(B)\) on homogeneous simple tensors has a unique norm-continuous extension
\begin{equation}
 \EE_X^\vac
 :=\id_{\A_X^\st}\widehat\otimes
 \omega_\at|_{\A_{X^c}^\st}
 \colon\A^\st\longrightarrow\A_X^\st.
 \label{eq:vacuum-expectation}
\end{equation}
Then \(\EE_X^\vac\) is a unital completely positive contraction and a conditional expectation onto \(\A_X^\st\).
More explicitly:
\begin{enumerate}[label=(\alph*),leftmargin=*]
\item it fixes \(\A_X^\st\) pointwise, has range \(\A_X^\st\), and \((\EE_X^\vac)^2=\EE_X^\vac\);
\item it is an \(\A_X^\st\)-bimodule map,
  \begin{equation}
   \EE_X^\vac(C_1AC_2)
   =C_1\EE_X^\vac(A)C_2,
   \qquad A\in\A^\st,\quad C_1,C_2\in\A_X^\st;
   \label{eq:vacuum-expectation-bimodule}
  \end{equation}
\item it preserves parity,
  \begin{equation}
   \EE_X^\vac\circ\Theta
   =\Theta\circ\EE_X^\vac,
   \qquad
   \EE_X^\vac\bigl((\A^\st)^+\bigr)
   \subset(\A_X^\st)^+;
   \label{eq:vacuum-expectation-parity}
  \end{equation}
\item it preserves the product state \(\omega_\at\),
  \begin{equation}
   \omega_\at\circ\EE_X^\vac=\omega_\at.
   \label{eq:vacuum-expectation-state}
  \end{equation}
\item if \(X\subset Y\Subset\Gamma\), then the expectations are consistent:
  \begin{equation}
   \EE_X^\vac\EE_Y^\vac
   =\EE_Y^\vac\EE_X^\vac
   =\EE_X^\vac;
   \label{eq:vacuum-expectation-tower}
  \end{equation}
\item for every \(A\in\A^\st\),
  \begin{equation}
   \|A-\EE_X^\vac(A)\|
   \le2\inf_{C\in\A_X^\st}\|A-C\|.
   \label{eq:vacuum-expectation-approximation}
  \end{equation}
\end{enumerate}
\end{lemma}

\begin{proof}
Araki's Theorem~3.1 \cite{Araki04}, applied to the local graded matrix algebras indexed by \(\Gamma\), gives the expectation and parts~(a)--(e).
As a conditional expectation, it is completely positive.
For \(C\in\A_X^\st\), contractivity and \(\EE_X^\vac(C)=C\) give
\[
 \|A-\EE_X^\vac(A)\|
 \le\|A-C\|+\|\EE_X^\vac(A-C)\|
 \le2\|A-C\|.
\]
Taking the infimum proves~(f).
\end{proof}

\begin{lemma}
\label{lem:vacuum-localization}
Let \(B\in\A^\st\) satisfy \(\pi_\at(B)\Omega_\at=0\).
Then, for every finite \(X\),
\begin{equation*}
 \EE_X^\vac(B)P_X^\vac=0.
\end{equation*}
If also \(\pi_\at(B^*)\Omega_\at=0\), then
\begin{equation}
 \EE_X^\vac(B)P_X^\vac
 =P_X^\vac\EE_X^\vac(B)=0.
 \label{eq:localized-vacuum-annihilation}
\end{equation}
\end{lemma}

\begin{proof}
Put \(B_X:=\EE_X^\vac(B)\in\A_X^\st\).
The bimodule and state-preservation properties in Lemma~\ref{lem:vacuum-conditional-expectation} give
\begin{equation*}
 \omega_\at(B_X^*B_X)
 =\omega_\at(\EE_X^\vac(B_X^*B))
 =\omega_\at(B_X^*B)=0,
\end{equation*}
where the last equality follows from \(\pi_\at(B)\Omega_\at=0\).
Let \(\pi_X\) be an irreducible representation of the local matrix algebra \(\A_X^\st\).
It is faithful, and the product of the even pure onsite states is represented by a unit vector \(\Omega_X\), with rank-one support projection \(P_X^\vac=\prod_{x\in X}P_x\), so
\begin{equation*}
 \|\pi_X(B_X)\Omega_X\|^2=\omega_\at(B_X^*B_X)=0,
 \qquad
 \pi_X(P_X^\vac)=|\Omega_X\rangle\langle\Omega_X|.
\end{equation*}
Thus \(\pi_X(B_XP_X^\vac)=0\), and faithfulness gives \(B_XP_X^\vac=0\).
If \(\pi_\at(B^*)\Omega_\at=0\), apply the same argument to \(B^*\) and take adjoints to obtain \(P_X^\vac B_X=0\).
\end{proof}

To decompose quasi-local observables into local terms, choose the dyadic radii
\begin{equation}
 R_0:=0,\qquad R_n:=2^{n-1}\quad(n\ge1).
 \label{eq:localization-radii}
\end{equation}
For a nonempty finite set \(X\), define the vacuum shells by
\begin{equation}
 \Delta_{X,0}^\vac:=\EE_X^\vac,
 \qquad
 \Delta_{X,n}^\vac
 :=\EE_{X^{(R_n)}}^\vac-
 \EE_{X^{(R_{n-1})}}^\vac,
 \quad n\ge1.
 \label{eq:vacuum-shells}
\end{equation}
For every \(B\in\A^\st\), \(B=\sum_{n\ge0}\Delta_{X,n}^\vac(B)\) in norm.
Indeed, the partial sums telescope to \(\EE_{X^{(R_n)}}^\vac(B)\).
Given a local approximant to \(B\), its finite support lies in \(X^{(R_n)}\) for all sufficiently large \(n\); hence \eqref{eq:vacuum-expectation-approximation} shows that \(\EE_{X^{(R_n)}}^\vac(B)\) converges in norm to \(B\).
This is the standard conditional-expectation shell decomposition; compare \cite[Lemma~5.1]{NSY19} and the fermionic local approximations in \cite[Section~4]{NSY18}.

Contractivity bounds the zeroth vacuum shell by \(\|B\|\), every subsequent shell by \(2\|B\|\), and \(\|B-\EE_X^\vac(B)\|\) by \(2\|B\|\).

Suppose now that \(\pi_\at(B)\Omega_\at=0\).
Lemma~\ref{lem:vacuum-localization} shows that each expectation annihilates the local vacuum on its own region.
For the shell difference, both expectations must annihilate the vacuum on the larger region \(X^{(R_n)}\).
For \(n\ge1\), product-vacuum factorization gives
\begin{equation*}
 P_{X^{(R_n)}}^\vac
 =P_{X^{(R_{n-1})}}^\vac
  P_{X^{(R_n)}\setminus X^{(R_{n-1})}}^\vac
 \le P_{X^{(R_{n-1})}}^\vac.
\end{equation*}
Thus
\begin{equation*}
 \EE_{X^{(R_{n-1})}}^\vac(B)P_{X^{(R_n)}}^\vac
 =\EE_{X^{(R_{n-1})}}^\vac(B)
  P_{X^{(R_{n-1})}}^\vac P_{X^{(R_n)}}^\vac=0.
\end{equation*}
The larger-region expectation already annihilates \(P_{X^{(R_n)}}^\vac\) on the right by the lemma.
Subtracting the two expectations gives \(\Delta_{X,n}^\vac(B)P_{X^{(R_n)}}^\vac=0\).
If also \(\pi_\at(B^*)\Omega_\at=0\), apply this argument to \(B^*\) and take adjoints to obtain
\begin{equation}
 \Delta_{X,n}^\vac(B)P_{X^{(R_n)}}^\vac
 =P_{X^{(R_n)}}^\vac\Delta_{X,n}^\vac(B)=0,
 \qquad n\ge0.
 \label{eq:vacuum-shell-annihilation}
\end{equation}
The \(n=0\) case is Lemma~\ref{lem:vacuum-localization} itself.

For every \(\lambda\in\R\), define the normalized local shells by
\begin{equation}
 R_{\Lambda,\lambda}(Z,n)
 :=\Delta_{Z,n}^\vac\!\left(
   \int_0^1G_{\Lambda,t\lambda}(Z)\,\dd t\right),
 \qquad \varnothing\ne Z\Subset\Gamma,\quad n\ge0.
 \label{eq:R-shells}
\end{equation}
Since \(\cG_{\Lambda,\lambda}(Z)=\lambda\int_0^1G_{\Lambda,t\lambda}(Z)\,\dd t\), the normal-form terms are \(\lambda R_{\Lambda,\lambda}(Z,n)\).
We set \(R_{\Lambda,\lambda}(\varnothing,n):=0\).
The maps \(\Delta_{Z,n}^\vac\) preserve parity and adjoints and take values in \(\A_{Z^{(R_n)}}^\st\), so \(R_{\Lambda,\lambda}(Z,n)\) is even, self-adjoint, and supported in \(Z^{(R_n)}\).
If in addition \(\lambda\in I_\Lambda\), then \(t\lambda\in I_\Lambda\) for \(t\in[0,1]\), and the annihilation identities \eqref{eq:G-vacuum-annihilation} and \eqref{eq:vacuum-shell-annihilation} give
\begin{equation*}
 R_{\Lambda,\lambda}(Z,n)P_{Z^{(R_n)}}^\vac
 =P_{Z^{(R_n)}}^\vac R_{\Lambda,\lambda}(Z,n)=0.
\end{equation*}
For \(\lambda\in I_\Lambda\) and nonempty \(Z\Subset\Gamma\), the definitions \eqref{eq:R-shells} and \eqref{eq:vacuum-shells} give
\begin{equation*}
 \lambda\sum_{n=0}^m R_{\Lambda,\lambda}(Z,n)
 =\EE_{Z^{(R_m)}}^\vac\bigl(\cG_{\Lambda,\lambda}(Z)\bigr)
 \stackrel{m\to\infty}{\longrightarrow} \cG_{\Lambda,\lambda}(Z)
\end{equation*}
in norm.
Proposition~\ref{prop:rooted-vacuum-normal-form} below proves that \(\sum_{Z,n}\lambda\pi_\at(R_{\Lambda,\lambda}(Z,n))\) converges absolutely in norm to \(K_{\Lambda,\lambda}^\nf-K_0\).

\subsection{Uniform locality and the rooted normal-form bound}

We next prove the locality estimates from Step~4 for the generator \(\eta_{\Lambda,\lambda}\), the automorphism \(\sigma_{\Lambda,\lambda}\), and the centered terms \(G_{\Lambda,\lambda}(Z)\).
To make their constants independent of the perturbation strength, restrict the coupling to
\begin{equation}
 J_{\rm loc}
 :=\{\lambda\in\R\,|\,|\lambda|\|\cV\|_{\mu_{\cV},1;2}\le g\}.
 \label{eq:locality-coupling-ball}
\end{equation}
Choose \(0<\mu<\mu_{\cV}\) with \(\|\Phi_0\|_{\mu,1;2}<\infty\).
For the quadratic application, fix \(\mu:=\frac12\min\{\mu_h,\mu_{\cV}\}\); for the interacting product-state application with exponential decay choose \(\mu<\mu_0\).
The locality estimates hold throughout \(J_{\rm loc}\).
Combined with the two-sided vacuum annihilation on \(I_\Lambda\), they give the rooted vacuum normal form on \(I_\Lambda\cap J_{\rm loc}\); the locality estimates do not use the gap.

To construct local approximants, let \(\F_X^{\rm tr}:\A^\st\to\A_X^\st\) denote the canonical normalized-trace CAR conditional expectation, with \(\F_\varnothing^{\rm tr}\) the expectation onto the scalars.
It is unital, contractive, parity preserving, fixes \(\A_X^\st\), and satisfies
\begin{equation}
 \F_X^{\rm tr}\F_Y^{\rm tr}
 =\F_{X\cap Y}^{\rm tr},
 \qquad
 \lim_{R\to\infty}\F_{X^{(R)}}^{\rm tr}(A)=A
 \quad(A\in\A^\st,\ X\ne\varnothing).
 \label{eq:tracial-CAR-expectation}
\end{equation}
See \cite[Section~4]{NSY18}, \cite[Proposition~2.1]{BTW26}, and \cite[Theorem~4.7, Corollary~4.8, and Lemma~4.9]{ArakiMoriya03}.
We use the same notation for the normalized-trace expectations \(\F_X^{\rm tr}:\A\to\A_X\) on the physical algebra.

To bound the locality errors of \(A^{\D_{\Lambda,\lambda}^\ph}\) and \(\cI_{\Lambda,\lambda}^\ph(A)\), we use the standard quasi-adiabatic filtering argument: combine the Lieb--Robinson bound for \(\tau_t^{\ph,\Lambda,\lambda}(A)\) with the time decay of \(w\) and \(W\) \cite[Section~6.5.1]{NSY19}.
\begin{lemma}
\label{lem:weighted-filter-locality}
There are constants \(a_{\rm LR}\), \(v_{\rm LR}\), \(C>0\), depending only on the geometric data, \(r,\mu,g\), and \(\|\Phi_0\|_{\mu,1;2}\), such that, for every finite \(\Lambda\), every \(\lambda\in J_{\rm loc}\), every nonempty finite \(X\), every \(A\in\A_X\), every \(t\in\R\), and every \(R\in\N_0\),
\begin{equation}
 \bigl\|(\id-\F_{X^{(R)}}^{\rm tr})
 \tau_t^{\ph,\Lambda,\lambda}(A)\bigr\|
 \le C|X|\|A\|\ell(R,t),
 \qquad
 \ell(R,t):=\min\{1,\E^{-a_{\rm LR}(R-v_{\rm LR}|t|)}\}.
 \label{eq:uniform-dynamics-locality}
\end{equation}
For \(k\in L^1(\R)\) satisfying \(|k(t)|\le C_k\E^{-c_k|t|^\theta}\) for some \(0<\theta<1\), there are \(C_k',\mu_k>0\) such that
\begin{equation}
 \bigl\|(\id-\F_{X^{(R)}}^{\rm tr})
 \int_\R k(t)\tau_t^{\ph,\Lambda,\lambda}(A)\,\dd t\bigr\|
 \le CC_k'|X|\|A\|\E^{-\mu_kR^\theta}.
 \label{eq:weighted-map-locality}
\end{equation}
The constants \(C_k',\mu_k\) additionally depend on \(C_k,c_k,\theta\).
All constants are independent of \(\Lambda,\lambda,X\).
The same estimates hold for the full physical dynamics \(\tau_t^{\ph,\lambda}\) of \(\Phi_0+\lambda\cV\).
\end{lemma}

\begin{proof}
On \(J_{\rm loc}\), the Hamiltonian interactions satisfy
\[
 \|\Phi_0+\lambda\cV_\Lambda\|_{\mu,1;2},\quad
 \|\Phi_0+\lambda\cV\|_{\mu,1;2}
 \le\|\Phi_0\|_{\mu,1;2}+g.
\]
This also bounds their first support moments.
The finite-region approximation estimate \cite[Theorem~2.9 and Remark~2.10]{FermionicLR}, with weight \(\E^{\mu R}\), and contractivity of the tracial expectation give \eqref{eq:uniform-dynamics-locality} with
\[
 C=2,\qquad a_{\rm LR}=\mu,\qquad
 v_{\rm LR}=\frac{2(\|\Phi_0\|_{\mu,1;2}+g)}{\mu}.
\]
Splitting the time integral at \(|t|=R/(2v_{\rm LR})\), as in \cite[Section~6.5.1]{NSY19}, gives
\begin{align}
 \int_\R|k(t)|\ell(R,t)\,\dd t
 &\le\|k\|_1\E^{-a_{\rm LR}R/2}
       +\int_{|t|>R/(2v_{\rm LR})}|k(t)|\,\dd t
 \label{eq:qk-time-split}\\
 &\le C_k'\E^{-\mu_kR^\theta},
 \label{eq:qk-stretched}
\end{align}
where the last inequality uses the assumed stretched-exponential decay of \(k\).
Integrating \eqref{eq:uniform-dynamics-locality} against \(|k(t)|\) proves \eqref{eq:weighted-map-locality}; this applies to \(w\) and \(W\) by \eqref{eq:w-stretched-time}--\eqref{eq:W-stretched-time}.
\end{proof}

\begin{lemma}
\label{lem:rooted-shell-summation}
Let \(J\subseteq\R\) be an interval, and let \(B_s(X,n)\in\A_{X^{(R_n)}}\) be a family of operators indexed by \(s\in J\), nonempty finite sets \(X\Subset\Gamma\), and \(n\in\N_0\), with radii \(R_n\) as in \eqref{eq:localization-radii}.
Let \(c\) be a nonnegative function on the nonempty finite subsets of \(\Gamma\), and let \((b_n)_{n\in\N_0}\) be a nonnegative sequence such that
\[
 \|B_s(X,n)\|\le c(X)b_n
 \qquad(s\in J,\ \varnothing\ne X\Subset\Gamma,\ n\in\N_0).
\]
Fix \(k\ge0\) and a nondecreasing function \(v:[0,\infty)\to[1,\infty)\) satisfying \(v(r+t)\le v(r)v(t)\) for \(r,t\ge0\), and assume
\[
 \sum_{n\ge0}(1+R_n)^{(k+1)\dvol}v(2R_n)b_n<\infty,
 \qquad
 \sup_{x\in\Gamma}\sum_{X\ni x}|X|^kv(\diam X)c(X)<\infty.
\]
Then, for each nonempty finite \(Y\Subset\Gamma\), the series
\[
 B_s(Y):=\sum_{\substack{\varnothing\ne X\Subset\Gamma,\ n\in\N_0\\X^{(R_n)}=Y}}B_s(X,n)
\]
converges absolutely in norm, uniformly in \(s\in J\).
Set \(B_s(\varnothing):=0\).
The resulting local operator family satisfies
\begin{align}
 \sup_{\substack{s\in J\\y\in\Gamma}}\sum_{Y\ni y}|Y|^k & v(\diam Y)\|B_s(Y)\|\nonumber\\*
 &\quad\le C_\Gamma^{k+1}
 \sum_{n\ge0}(1+R_n)^{(k+1)\dvol}v(2R_n)b_n
 \ \sup_{x\in\Gamma}\sum_{X\ni x}|X|^kv(\diam X)c(X).
 \label{eq:common-shell-root-bound}
\end{align}
The bound also holds when \(\|B_s(Y)\|\) on the left is replaced by \(\sum_{X,n:X^{(R_n)}=Y}\|B_s(X,n)\|\).
If \(s\mapsto B_s(X,n)\) is norm-continuous for every \(X,n\), then \(s\mapsto B_s(Y)\) is norm-continuous for every \(Y\).
All conclusions hold with the same bound after restricting \(X\) to any subfamily of nonempty finite sets, and the statement also holds with \(\A\) replaced throughout by \(\A^\st\).
If, in addition, \(\sum_Xc(X)<\infty\), the series \(\sum_{X,n}B_s(X,n)\) converges absolutely in norm, uniformly in \(s\in J\), and the sums over \(X\) and \(n\) may be interchanged.
\end{lemma}

\begin{proof}
For fixed \(n\in\N_0\) and \(y\in\Gamma\), apply the geometric bounds \eqref{eq:metric-thickening}--\eqref{eq:metric-root-count} with \(R=R_n\).
Together with the assumptions on \(v\), they give
\begin{align*}
 \sum_{X:y\in X^{(R_n)}}|X^{(R_n)}|^k & v(\diam X^{(R_n)})c(X)b_n\\
 &\le C_\Gamma^k(1+R_n)^{k\dvol}v(2R_n)b_n
 \sum_{x\in\ball_{R_n}(y)}\sum_{X\ni x}|X|^kv(\diam X)c(X).
\end{align*}
Bounding \(|\ball_{R_n}(y)|\) by \(C_\Gamma(1+R_n)^{\dvol}\) and summing over \(n\) gives the right-hand side of \eqref{eq:common-shell-root-bound}.
For a fixed nonempty \(Y\), choose \(y\in Y\).
Since \(|X^{(R_n)}|^kv(\diam X^{(R_n)})\ge1\), the preceding estimate implies \(\sum_{X,n:X^{(R_n)}=Y}c(X)b_n<\infty\).
Thus the series defining \(B_s(Y)\) converges absolutely and uniformly in \(s\), and the triangle inequality proves \eqref{eq:common-shell-root-bound}, including its sum-of-norms version.
Uniform convergence gives norm continuity whenever each summand is norm-continuous.
Restricting the index set only decreases the nonnegative sums.
Finally, the additional assumption gives
\[
 \sup_{s\in J}\sum_{X,n}\|B_s(X,n)\|
 \le\Bigl(\sum_Xc(X)\Bigr)\Bigl(\sum_nb_n\Bigr)<\infty,
\]
which proves absolute uniform convergence of the full series and justifies interchanging the sums.
\end{proof}

We estimate the image of a sum of shells under a quasi-local automorphism by approximating each transformed shell separately.

\begin{lemma}
\label{lem:shell-composition}
Let \(\alpha\) be an automorphism of \(\A^\st\) and \(Z\Subset\Gamma\) be nonempty.
Let \((F_n)_{n\in\N_0}\) be a sequence with \(F_n\in\A_{Z^{(R_n)}}^\st\), where \(R_n\) are the radii in \eqref{eq:localization-radii}.
Suppose there are constants \(C_F,C_\alpha\ge0\), a nonnegative sequence \((b_n)_{n\in\N_0}\), and a nonincreasing function \(q_\alpha:[0,\infty)\to[0,\infty)\) such that
\[
 \|F_n\|\le C_Fb_n\quad(n\in\N_0),
 \qquad
 \sum_{n\ge0}(1+R_n)^{\dvol}b_n<\infty,
\]
and, for every \(n\in\N_0\) and \(R\ge0\),
\[
 \inf_{B\in\A_{(Z^{(R_n)})^{(R)}}^\st}\|\alpha(F_n)-B\|
 \le C_\alpha|Z^{(R_n)}|\|F_n\|q_\alpha(R).
\]
Set \(F:=\sum_{n\ge0}F_n\), where the series converges absolutely in norm.
Then, for every \(R\ge0\),
\begin{equation}
 \inf_{B\in\A_{Z^{(R)}}^\st}\|\alpha(F)-B\|
 \le \max\{1,C_\Gamma C_\alpha\}\,C_F|Z|\,
 \Bigg(q_\alpha(R/2)\sum_{n\ge0}(1+R_n)^{\dvol}b_n
 +\sum_{R_n>R/2}b_n\Bigg).
 \label{eq:common-shell-composition}
\end{equation}

If \(b_n\le C_b\E^{-\mu R_n^\theta}\) and \(q_\alpha(R)\le C_b\E^{-\mu R^\theta}\), with \(C_b,\mu>0\) and \(0<\theta\le1\), the right-hand side of \eqref{eq:common-shell-composition} is bounded by \(C'C_F|Z|\E^{-\mu'R^\theta}\), where \(C',\mu'>0\) depend only on \(C_\Gamma,\dvol,C_\alpha,C_b,\mu,\theta\).
If \(b_n\le C_b(1+R_n)^{-q}\) and \(q_\alpha(R)\le C_b(1+R)^{-p}\), with \(C_b>0\), \(q>\dvol\), and \(p>0\), that right-hand side is bounded by \(C'C_F|Z|(1+R)^{-\min\{p,q\}}\), where \(C'\) depends only on \(C_\Gamma,\dvol,C_\alpha,C_b,p,q\).

For \(c\in\C\), set \(G:=\alpha(F)-c\1\).
Then, for every \(R\ge0\),
\[
 \|G-\EE_{Z^{(R)}}^\vac(G)\|
 \le2\inf_{B\in\A_{Z^{(R)}}^\st}\|\alpha(F)-B\|.
\]
\end{lemma}

\begin{proof}
For fixed \(R\ge0\), approximate \(\alpha(F_n)\), for each \(R_n\le R/2\), by \(B_{n,R}\in\A_{(Z^{(R_n)})^{(R/2)}}^\st\) with the error in the hypothesis.
The sum \(B_R:=\sum_{R_n\le R/2}B_{n,R}\) belongs to \(\A_{Z^{(R)}}^\st\), and
\[
 \|\alpha(F)-B_R\|
 \le C_\alpha C_Fq_\alpha(R/2)
       \sum_{R_n\le R/2}|Z^{(R_n)}|b_n
       +C_F\sum_{R_n>R/2}b_n.
\]
The thickening bound \(|Z^{(R_n)}|\le C_\Gamma|Z|(1+R_n)^{\dvol}\) gives \eqref{eq:common-shell-composition}.
Under either decay assumption, the weighted series in \eqref{eq:common-shell-composition} is finite, and the dyadic tails satisfy
\[
 \sum_{R_n>R/2}b_n\le
 \begin{cases}
  C\E^{-\mu R^\theta/2^{\theta+1}},&b_n\le C_b\E^{-\mu R_n^\theta},\\
  C(1+R)^{-q},&b_n\le C_b(1+R_n)^{-q}.
 \end{cases}
\]
Inserting these estimates and the bound on \(q_\alpha\) into \eqref{eq:common-shell-composition} gives the two decay estimates.
Since \(\EE_{Z^{(R)}}^\vac\) fixes the scalar term, \eqref{eq:vacuum-expectation-approximation} gives the final estimate for \(G\).
\end{proof}

Summability of the shell bounds gives both the rooted estimate and the local normal form.

\begin{proposition}
\label{prop:rooted-vacuum-normal-form}
Let \(J\) be an interval containing zero.
Suppose the centered terms \eqref{eq:G-Lambda-s} satisfy, uniformly for \(\lambda\in J\),
\[
 \|G_{\Lambda,\lambda}(Z)\|\le c_\Lambda(Z),\qquad
 \|G_{\Lambda,\lambda}(Z)-\EE_{Z^{(R)}}^\vac(G_{\Lambda,\lambda}(Z))\|
 \le c_\Lambda(Z)q_G(R),
\]
where \(c_\Lambda(Z)\ge0\) and \(q_G\) is nonnegative and nonincreasing.
Set \(b_0:=1\) and \(b_n:=2q_G(R_{n-1})\) for \(n\ge1\).
Assume \(\sum_Zc_\Lambda(Z)<\infty\) for each fixed \(\Lambda\) and \(\sum_n(1+R_n)^{\dvol}b_n<\infty\).
Then
\begin{equation}
 f_{\Lambda,\lambda}:=\sup_{y\in\Gamma}
 \sum_{\substack{Z\Subset\Gamma,\ n\ge0\\y\in Z^{(R_n)}}}
 \|R_{\Lambda,\lambda}(Z,n)\|
 \le C_\Gamma\sum_{n\ge0}(1+R_n)^{\dvol}b_n
                  \sup_{x\in\Gamma}\sum_{Z\ni x}c_\Lambda(Z).
 \label{eq:common-vacuum-root-bound}
\end{equation}
For fixed \(\Lambda\), \(\sum_{Z,n}\|R_{\Lambda,\lambda}(Z,n)\|<\infty\).
On \(I_\Lambda\cap J\), these terms annihilate their local vacua on both sides and
\begin{equation}
 K_{\Lambda,\lambda}^\nf
 =K_0+\lambda\sum_{Z\Subset\Gamma}\sum_{n\ge0}
          \pi_\at(R_{\Lambda,\lambda}(Z,n))
 \label{eq:normal-form-identity}
\end{equation}
holds on \(\operatorname{Dom}(K_0)\), and as a form identity on \(\operatorname{Dom}(K_0^{1/2})\).
\end{proposition}

\begin{proof}
For \(u\in J\) and \(n\ge1\), the assumed bounds on the two errors \(G_{\Lambda,u}(Z)-\EE_{Z^{(R_j)}}^\vac(G_{\Lambda,u}(Z))\), with \(j=n,n-1\), give
\[
 \|\Delta_{Z,n}^\vac(G_{\Lambda,u}(Z))\|
 \le c_\Lambda(Z)\bigl(q_G(R_n)+q_G(R_{n-1})\bigr)
 \le c_\Lambda(Z)b_n.
\]
For \(n=0\), contractivity gives \(\|\Delta_{Z,0}^\vac(G_{\Lambda,u}(Z))\|\le c_\Lambda(Z)b_0\).
Since \(t\lambda\in J\) for \(\lambda\in J\) and \(t\in[0,1]\), averaging these estimates in \eqref{eq:R-shells} yields \(\|R_{\Lambda,\lambda}(Z,n)\|\le c_\Lambda(Z)b_n\).
Apply Lemma~\ref{lem:rooted-shell-summation} with \(k=0\) and \(v=1\).
For fixed \(\Lambda\), \(\sum_{Z,n}\|R_{\Lambda,\lambda}(Z,n)\| \le(\sum_Zc_\Lambda(Z))(\sum_nb_n)<\infty\).
We may therefore sum the termwise telescoping identities following \eqref{eq:R-shells} over \(Z\).
Together with \eqref{eq:bounded-normal-form}, they give \eqref{eq:normal-form-identity}.
These terms annihilate their local vacua by the argument following \eqref{eq:R-shells}.
The operator series \(\lambda\sum_{Z,n}\pi_\at(R_{\Lambda,\lambda}(Z,n))\) converges absolutely in norm to the bounded perturbation in \eqref{eq:bounded-normal-form}.
By Proposition~\ref{prop:centered-derivative}, \(\operatorname{Dom}((K_{\Lambda,\lambda}^\nf)^{1/2})=\operatorname{Dom}(K_0^{1/2})\).
\end{proof}

Write \(\tau_t^{\ph,\lambda}\) for the dynamics of \(\Phi_0+\lambda\cV\), and define
\begin{equation}
 \cI_{\lambda}^\ph(A):=\int_\R W(t)\,\tau_t^{\ph,\lambda}(A)\,\dd t.
 \label{eq:full-physical-inverse-filter}
\end{equation}
For a nonempty finite set \(X\), use the radii \eqref{eq:localization-radii} to define the tracial shells
\begin{equation}
 \Delta_{X,0}^{\rm tr}:=\F_X^{\rm tr},\qquad
 \Delta_{X,n}^{\rm tr}:=
 \F_{X^{(R_n)}}^{\rm tr}
 -\F_{X^{(R_{n-1})}}^{\rm tr}\quad(n\ge1).
 \label{eq:tracial-shells}
\end{equation}
For \(B\) in either algebra, \(\sum_{j=0}^n\Delta_{X,j}^{\rm tr}(B)=\F_{X^{(R_n)}}^{\rm tr}(B)\stackrel{n\to\infty}{\longrightarrow} B\) in norm by \eqref{eq:tracial-CAR-expectation}.
For \(\lambda\in J_{\rm loc}\), define
\begin{align}
 \eta_{\lambda}(Y)
 &:=\sum_{X,n:X^{(R_n)}=Y}
       \Delta_{X,n}^{\rm tr}\cI_{\lambda}^\ph(\cV(X)),
 \label{eq:physical-generator-regrouping}\\
 \eta_{\Lambda,\lambda}(Y)
 &:=\sum_{X,n:X^{(R_n)}=Y}
       \Delta_{X,n}^{\rm tr}
           \cI_{\Lambda,\lambda}^\ph(\cV_\Lambda(X)).
 \label{eq:physical-cutoff-generator-regrouping}
\end{align}

We now bound the weighted interaction norms of \(\eta_\lambda\) and \(\eta_{\Lambda,\lambda}\), and show that the latter generates the cutoff flow \(\beta_{\Lambda,\lambda,\lambda'}\).

\begin{lemma}
\label{lem:inverse-filter-interaction}
For \(\lambda\in J_{\rm loc}\), the interactions \(\eta_{\lambda},\eta_{\Lambda,\lambda}\) in \eqref{eq:physical-generator-regrouping}--\eqref{eq:physical-cutoff-generator-regrouping} are well defined, even, and self-adjoint, and they depend componentwise norm-continuously on \(\lambda\).
For \(\lambda\in J_{\rm loc}\), define the cutoff shells by
\begin{equation}
 b_{\Lambda,\lambda}(Z,n):=\Delta_{Z,n}^{\rm tr}
       \cI_{\Lambda,\lambda}^\ph(\cV_\Lambda(Z)),
 \qquad \varnothing\ne Z\Subset\Gamma,\quad n\ge0,
 \label{eq:inverse-filter-shells}
\end{equation}
and set \(b_{\Lambda,\lambda}(\varnothing,n):=0\).
Thus
\begin{equation}
 \eta_{\Lambda,\lambda}(Y)
 =\sum_{\substack{\varnothing\ne Z\Subset\Gamma,\ n\ge0\\
                         Z^{(R_n)}=Y}}b_{\Lambda,\lambda}(Z,n).
 \label{eq:kappa-regrouping}
\end{equation}
There is a rate \(\mu_\eta>0\), depending only on \(C_\Gamma,\dvol,r,g,\mu,\mu_{\cV},\|\Phi_0\|_{\mu,1;2}\), and \(\theta\), and, for every \(k\ge0\), a constant \(C_{\eta,k}>0\) depending additionally on \(k\), such that
\begin{equation}
 \sup_{\lambda\in J_{\rm loc}}\|\eta_{\lambda}\|_{\mu_\eta,\theta;k},
 \quad
 \sup_{\substack{\Lambda\Subset\Gamma\\\lambda\in J_{\rm loc}}}
 \|\eta_{\Lambda,\lambda}\|_{\mu_\eta,\theta;k}
 \le C_{\eta,k}\|\cV\|_{\mu_{\cV},1;2}.
 \label{eq:kappa-stretched-bound}
\end{equation}
For every \(\lambda\in J_{\rm loc}\) and local \(A\), the derivation \(\delta_{\eta_{\lambda}}(A)\) is given by the series in \eqref{eq:spectral-flow-derivation}.
For fixed \(\Lambda\),
\begin{equation}
 \mathscr S_{\Lambda,\lambda}^\ph
 =\sum_Y\eta_{\Lambda,\lambda}(Y)
 \label{eq:S-interaction-decomposition}
\end{equation}
with absolute convergence in operator norm for every \(\lambda\in J_{\rm loc}\).
The interaction path \((\eta_{\Lambda,\lambda})_{\lambda\in J_{\rm loc}}\) generates the inner cocycle \(\beta_{\Lambda,\lambda,\lambda'}\) for \(\lambda,\lambda'\in J_{\rm loc}\).
\end{lemma}

\begin{proof}
The tracial-shell construction of the spectral-flow interaction is described in \cite[Section~6.5]{NSY19} and its fermionic version in \cite{NSY18}.
For \(n\ge2\), use the tracial tower identity \eqref{eq:tracial-CAR-expectation} and Lemma~\ref{lem:weighted-filter-locality} to estimate \(b_{\Lambda,\lambda}(Z,n)\).
Contractivity gives the same estimate for \(n=0,1\), after increasing \(C\).
Thus, for all \(n\ge0\),
\begin{equation}
 \|b_{\Lambda,\lambda}(Z,n)\|
 \le C|Z|\|\cV_\Lambda(Z)\|
                    \E^{-b_{\rm inv}R_n^\theta},
 \label{eq:inverse-shell-bound}
\end{equation}
with \(C,b_{\rm inv}>0\) independent of \(\Lambda\) and \(\lambda\in J_{\rm loc}\).
Choose
\begin{equation}
 0<\mu_\eta<2^{-\theta-1}b_{\rm inv}.
 \label{eq:mu-kappa-choice}
\end{equation}
Lemma~\ref{lem:rooted-shell-summation} and the support-moment estimate of Lemma~\ref{lem:support-moment-rate-trade} then give, for every \(k\ge0\),
\begin{align}
 \sup_{x\in\Gamma}\sum_{\substack{\varnothing\ne Z\Subset\Gamma,\ n\ge0\\x\in Z^{(R_n)}}}&
 |Z^{(R_n)}|^k\E^{\mu_\eta\diam(Z^{(R_n)})^\theta}
            \|b_{\Lambda,\lambda}(Z,n)\|\nonumber\\*
 & \le C_k\|\cV\|_{\mu_{\cV},1;2}
       \sum_{n\ge0}(1+R_n)^{(k+1)\dvol}
          \E^{-(b_{\rm inv}-2^\theta\mu_\eta)R_n^\theta}<\infty.
 \label{eq:kappa-root-count}
\end{align}
The same estimate holds for the full dynamics and sources \(\cV(Z)\).
By \eqref{eq:kappa-root-count}, the series in \eqref{eq:physical-generator-regrouping}--\eqref{eq:physical-cutoff-generator-regrouping} converge absolutely and satisfy \eqref{eq:kappa-stretched-bound}, with \(\mu_\eta\) independent of \(k\).
At \(k=0\), \eqref{eq:kappa-root-count} gives \(\sum_{Z,n}\|[\Delta_{Z,n}^{\rm tr}\cI_\lambda^\ph(\cV(Z)),A]\|<\infty\) for each local \(A\).
Grouping terms with support \(Y=Z^{(R_n)}\) and telescoping in \(n\) gives \eqref{eq:spectral-flow-derivation}.

Each filtered shell is norm-continuous in \(\lambda\) by Proposition~\ref{prop:nonautonomous-local-limit-package} and dominated convergence against \(|W|\).
For each fixed \(Y\), the bounds \eqref{eq:inverse-shell-bound} and \eqref{eq:kappa-root-count} give summable estimates independent of \(\lambda\).
The series defining \(\eta_\lambda(Y)\) and \(\eta_{\Lambda,\lambda}(Y)\) therefore converge uniformly in \(\lambda\), so both components are norm-continuous.
Since \(W\) is real and the dynamics and tracial expectations preserve parity and adjoints, the components \(\eta_\lambda(Y)\) and \(\eta_{\Lambda,\lambda}(Y)\) are even and self-adjoint.

At fixed \(\Lambda\), \eqref{eq:inverse-shell-bound} also gives
\[
 \sum_{Z,n}\|b_{\Lambda,\lambda}(Z,n)\|
 \le C\sum_{Z\subset\Lambda}|Z|\|\cV(Z)\|
          \sum_{n\ge0}\E^{-b_{\rm inv}R_n^\theta}<\infty.
\]
Regrouping and telescoping this series yields
\[
 \sum_Y\eta_{\Lambda,\lambda}(Y)
 =\sum_{Z\subset\Lambda}\cI_{\Lambda,\lambda}^\ph(\cV(Z))
 =\mathscr S_{\Lambda,\lambda}^\ph.
\]
By \eqref{eq:S-interaction-decomposition}, \(\delta_{\eta_{\Lambda,\lambda}}(A)=[\mathscr S_{\Lambda,\lambda}^\ph,A]\) for local \(A\).
Uniqueness in Proposition~\ref{prop:nonautonomous-local-limit-package} therefore shows that \(\eta_{\Lambda,\lambda}\) generates \(\beta_{\Lambda,\lambda,\lambda'}\).
\end{proof}

We also need locality after expressing the physical flow in product-vacuum coordinates.
The following lemma applies to both decay regimes.

\begin{lemma}
\label{lem:uniform-alpha-locality}
Let \(v(R):=\E^{cR^\theta}\), with \(c>0\), \(0<\theta<1\), or \(v(R):=(1+R)^p\), with \(p>0\).
Suppose the physical cutoff flow is generated by even, self-adjoint, componentwise norm-continuous interactions \(\eta_{\Lambda,u}\) on a compact interval \(J\) containing zero.
Assume
\begin{gather*}
 \sup_{\substack{\Lambda,\ u\in J\\x\in\Gamma}}
       \sum_{Y\ni x}v(\diam Y)\|\eta_{\Lambda,u}(Y)\|<\infty,\\
 \sup_{\Lambda}\int_J
 \sup_{x\in\Gamma}\sum_{Y\ni x}|Y|v(\diam Y)
                         \|\eta_{\Lambda,u}(Y)\|\,\dd u<\infty.
\end{gather*}
For the fixed coordinate path, use \(\Phi_{k_\st}\) from \eqref{eq:atomicizing-path-generator}; set \(\Phi_{k_\st}=0\) when \(\sigma_\st=\id\).
Assume
\[
 \sup_{x\in\Gamma}\sum_{Y\ni x}|Y|v(\diam Y)
                       \|\Phi_{k_\st}(Y)\|<\infty.
\]
Then, uniformly in finite \(\Lambda,X\), \(\lambda\in J\), \(A\in\A_X^\st\), and \(R\ge0\),
\begin{equation}
 \inf_{B\in\A_{X^{(R)}}^\st}
 \|\sigma_{\Lambda,\lambda}(A)-B\|
 \le C|X|\|A\|
 \begin{cases}
  \E^{-cR^\theta},&v(R)=\E^{cR^\theta},\\
  (1+R)^{-p},&v(R)=(1+R)^p.
 \end{cases}
 \label{eq:uniform-alpha-locality}
\end{equation}
The constant depends only on the integrated bound for \(\eta_{\Lambda,u}\) and the assumed bound on \(\Phi_{k_\st}\).
\end{lemma}

\begin{proof}
Choose a smooth nondecreasing function \(\chi:[0,1]\to[0,1]\) with \(\chi(0)=0\), \(\chi(1)=1\), and \(\chi'(0)=\chi'(1)=0\).
On the stacked algebra, define the interaction path
\[
 \Xi_t:=
 \begin{cases}
  \chi'(t)\Phi_{k_\st},&0\le t\le1,\\
  \lambda\chi'(t-1)(\eta_{\Lambda,\lambda\chi(t-1)}
                      \widehat\otimes\1),&1\le t\le2.
 \end{cases}
\]
It is componentwise continuous at the joining point.
In the convention \eqref{eq:path-cocycle-convention}, the cocycle generated by \(\Xi\) from \(0\) to \(2\) is \(\sigma_\st\circ(\beta_{\Lambda,0,\lambda}\widehat\otimes\id) =\sigma_{\Lambda,\lambda}\).
The factor \(\lambda\) also gives the correct orientation when \(\lambda<0\).
Changing variables in the two integrals gives
\begin{align*}
 &\int_0^2\sup_{x\in\Gamma}\sum_{Y\ni x}
           |Y|v(\diam Y)\|\Xi_t(Y)\|\,\dd t\\
 &\quad=\sup_{x\in\Gamma}\sum_{Y\ni x}
           |Y|v(\diam Y)\|\Phi_{k_\st}(Y)\|
   +\int_{\min(0,\lambda)}^{\max(0,\lambda)}\sup_{x\in\Gamma}\sum_{Y\ni x}
           |Y|v(\diam Y)\|\eta_{\Lambda,u}(Y)\|\,\dd u.
\end{align*}
The hypotheses bound this integral uniformly in \(\Lambda\) and \(\lambda\in J\).
Both weights are nondecreasing and satisfy \(v(0)=1\), \(v(R+S)\le v(R)v(S)\), and \(\lim_{R\to\infty}v(R)=\infty\).
The hypotheses and boundedness of \(\chi'\) also give a uniform zeroth-support-moment bound for \(\Xi_t\) on \([0,2]\).
The finite-region approximation estimate in \cite[Theorem~2.9 and Remark~2.10]{FermionicLR}, applied to this path, gives \eqref{eq:uniform-alpha-locality}.
\end{proof}

In the exponential case, the two bounds on \(\eta_{\Lambda,u}\) follow from Lemma~\ref{lem:inverse-filter-interaction} with \(c=\mu_\eta\).
In particular, with \(C_\eta:=C_{\eta,2}\),
\[
 \int_{J_{\rm loc}}\|\eta_{\Lambda,u}\|_{\mu_\eta,\theta;2}\,\dd u
 \le 2C_\eta g.
\]
The bound on the fixed coordinate path follows from Proposition~\ref{prop:exponential-stable-locality} in the quadratic case and is zero in the interacting product-state case.
Lemma~\ref{lem:uniform-alpha-locality} with \(c=\mu_\eta\) therefore gives stretched-exponential locality of \(\sigma_{\Lambda,\lambda}\).

\begin{proposition}
\label{prop:transported-D-locality}
There are \(C_{\rm ql}\ge2\|w\|_1\) and \(\mu_\zeta>0\) such that
\begin{align}
 \|G_{\Lambda,\lambda}(Z)\|
 &\le2\|w\|_1\|\cV_\Lambda(Z)\|,
 \label{eq:G-size}\\
 \|G_{\Lambda,\lambda}(Z)-\EE_{Z^{(R)}}^\vac(G_{\Lambda,\lambda}(Z))\|
 &\le C_{\rm ql}|Z|^2
       \|\cV_\Lambda(Z)\|\E^{-\mu_\zeta R^\theta}
 \label{eq:G-locality}
\end{align}
for all finite \(\Lambda\), \(\lambda\in J_{\rm loc}\), nonempty \(Z\Subset\Gamma\), and \(R\ge0\).
The constants depend only on the geometric and onsite data, \(g,\mu,\mu_{\cV},\theta\), \(\|\Phi_0\|_{\mu,1;2}\), and \(\|\Phi_{k_\st}\|_{\mu_\eta,\theta;1}\); they are independent of \(\cV,\Lambda,\lambda\).
\end{proposition}

\begin{proof}
The tracial expectations on the stacked and physical algebras satisfy
\[
 \F_Y^{\rm tr}(B\widehat\otimes\1)
 =\F_Y^{\rm tr}(B)\widehat\otimes\1,
 \qquad B\in\A,\quad Y\Subset\Gamma.
\]
The tracial shells of \(F_{\Lambda,\lambda,Z}\widehat\otimes\1\) obey \eqref{eq:inverse-shell-bound}, with \(w\) in place of \(W\), by Lemma~\ref{lem:weighted-filter-locality}.
Lemma~\ref{lem:uniform-alpha-locality} gives the stretched-exponential bound on \(\sigma_{\Lambda,\lambda}\).
Apply Lemma~\ref{lem:shell-composition} to \(F_{\Lambda,\lambda,Z}\widehat\otimes\1\), with \(C_F=C|Z|\|\cV_\Lambda(Z)\|\) and automorphism \(\sigma_{\Lambda,\lambda}\).
Its vacuum-expectation estimate with \(c=\omega_{\Lambda,\lambda}(F_{\Lambda,\lambda,Z})\) proves \eqref{eq:G-locality} for every \(R\ge0\).
Isometry and centering give \(\|G_{\Lambda,\lambda}(Z)\|\le2\|F_{\Lambda,\lambda,Z}\| \le2\|w\|_1\|\cV_\Lambda(Z)\|\).
\end{proof}

The factor \(|Z|^2\) in \eqref{eq:G-locality} is the support weight in the interaction norm of \(\cV\).
The following proposition is the main conclusion of this section.
It collects the vacuum normal-form identity and the uniform bound needed for the gap bootstrap in Section~\ref{sec:form-bound}.

\begin{proposition}
\label{prop:uniform-normal-form}
Let \(w,W\) be the filters constructed in Subsection~\ref{sec:compatible-filter-pair} for a fixed \(0<\theta<1\).
There is a constant \(C_{\rm nf}\ge1\), depending only on the fixed data of the unperturbed system, the geometry, the locality assumptions, and \(\theta\), such that for every finite \(\Lambda\) and \(\lambda\in J_{\rm loc}\),
\begin{equation}
 f_{\Lambda,\lambda}
 \le C_{\rm nf}\|\cV_\Lambda\|_{\mu_{\cV},1;2}
 \le C_{\rm nf}\|\cV\|_{\mu_{\cV},1;2}.
 \label{eq:normal-form-bound}
\end{equation}
On \(I_\Lambda\cap J_{\rm loc}\), the normalized shells annihilate their local vacua on both sides, and \eqref{eq:normal-form-identity} holds on \(\operatorname{Dom}(K_0^{1/2})\) as an absolutely convergent form series.
\end{proposition}

\begin{proof}
Apply Proposition~\ref{prop:rooted-vacuum-normal-form} with \(c_\Lambda(Z)=C_{\rm ql}|Z|^2 \|\cV_\Lambda(Z)\|\) and \(q_G(R)=\E^{-\mu_\zeta R^\theta}\).
The total sum of \(c_\Lambda\) is finite, and its rooted sum satisfies
\[
 \sup_{x\in\Gamma}\sum_{Z\ni x}c_\Lambda(Z)
 =C_{\rm ql}\sup_{x\in\Gamma}\sum_{Z\ni x}|Z|^2
              \|\cV_\Lambda(Z)\|
 \le C_{\rm ql}\|\cV_\Lambda\|_{\mu_{\cV},1;2}.
\]
Thus one may choose
\begin{equation}
 C_{\rm nf}:=\max\left\{1,\,
 C_\Gamma C_{\rm ql}
 \left[1+2\sum_{n\ge1}(1+R_n)^{\dvol}\E^{-\mu_\zeta R_{n-1}^\theta}\right]
 \right\}.
 \label{eq:explicit-C-nf}
\end{equation}
The series converges because its summands decay stretched-exponentially in \(R_{n-1}\).
The annihilation and normal-form identities follow from Proposition~\ref{prop:rooted-vacuum-normal-form}.
\end{proof}

\section{The relative form estimate and bootstrap closure}
\label{sec:form-bound}

\subsection{Vacuum annihilation and the excitation-number relative bound}

A self-adjoint term \(B_X\) supported in \(X\) that annihilates the local vacuum vanishes on the subspace with no excitations in \(X\), and therefore satisfies \(-\|B_X\|N_X\le B_X\le\|B_X\|N_X\), where \(N_X\) is the local excitation number.
This comparison with \(N\) follows Hastings \cite[Section~IV, Theorem~4]{Hastings19}; the use of local ground-state projections in relative form bounds is developed in \cite[Sections~4--7]{BH11}.

\begin{lemma}
\label{lem:one-local-term}
Let \(X\Subset\Gamma\) and let \(B_X=B_X^*\in\A_X^\st\) satisfy
\begin{equation*}
 B_XP_X^\vac
 =P_X^\vac B_X=0.
\end{equation*}
Then
\begin{equation}
 -\|B_X\|(\1-P_X^\vac)
 \le B_X\le
 \|B_X\|(\1-P_X^\vac),
 \qquad
 \1-P_X^\vac\le N_X.
 \label{eq:one-term-bound}
\end{equation}
\end{lemma}

\begin{proof}
The annihilation identities give \(B_X =(\1-P_X^\vac)B_X(\1-P_X^\vac)\), which proves the first two-sided operator inequality.
The even projections \(P_x\), \(x\in X\), commute, and hence
\[
 \1-P_X^\vac\le\sum_{x\in X}(\1-P_x)\le\sum_{x\in X}n_x=N_X.\qedhere
\]
\end{proof}

Summing this one-term estimate by sites converts a rooted bound on the local terms into a relative form bound by the total excitation number.
In the following form estimates, local observables act on \(\cH_\at\) through \(\pi_\at\); we suppress this representation in the notation.

\begin{lemma}
\label{lem:summed-form-bound}
Let \((B_j)_{j\in\cL}\) be a countable family of self-adjoint terms, with \(B_j\) supported in \(X_j\) and satisfying the hypotheses of Lemma~\ref{lem:one-local-term}.
If
\begin{equation}
 f:=\sup_{x\in\Gamma}\sum_{j:\,x\in X_j}
 \|B_j\|<\infty,
 \label{eq:f-def}
\end{equation}
then the series
\begin{equation*}
 \frakb(\varphi,\psi)
 :=\sum_{j\in\cL}
 \langle\varphi,B_j\psi\rangle
\end{equation*}
converges absolutely for \(\varphi,\psi\in\operatorname{Dom}(N^{1/2})\), defines a symmetric form continuous for the graph norm of \(N^{1/2}\), and satisfies
\begin{equation}
 |\frakb(\psi,\psi)|
 \le f\|N^{1/2}\psi\|^2.
 \label{eq:summed-N-bound}
\end{equation}
\end{lemma}

\begin{proof}
Put \(Q_j^{\rm exc}:=\1-P_{X_j}^\vac\).
Lemma~\ref{lem:one-local-term} gives \(B_j=Q_j^{\rm exc}B_jQ_j^{\rm exc}\) and \(Q_j^{\rm exc}\le N_{X_j}\).
For every finite \(\cL_0\subset\cL\),
\begin{equation*}
 \sum_{j\in\cL_0}\|B_j\| Q_j^{\rm exc} \le\sum_{j\in\cL_0}\|B_j\|N_{X_j}
 =\sum_{x\in\Gamma} \Bigg(\sum_{\substack{j\in\cL_0\\x\in X_j}}
 \|B_j\|\Bigg)n_x \le fN
\end{equation*}
as positive quadratic forms on \(\operatorname{Dom}(N^{1/2})\).
Letting \(\cL_0\) increase to \(\cL\) bounds the nonnegative series \(\sum_j\|B_j\|\langle\chi,Q_j^{\rm exc}\chi\rangle\) by \(f\|N^{1/2}\chi\|^2\), for \(\chi=\varphi,\psi\).
Cauchy--Schwarz for the scalar series therefore gives
\begin{align*}
 \sum_{j\in\cL}
 |\langle\varphi,B_j\psi\rangle|
 &\le\sum_j\|B_j\|\,
 \|Q_j^{\rm exc}\varphi\|\,
 \|Q_j^{\rm exc}\psi\|\\
 &\le
 \left(\sum_j\|B_j\|
 \langle\varphi,Q_j^{\rm exc}\varphi\rangle\right)^{1/2}
 \left(\sum_j\|B_j\|
 \langle\psi,Q_j^{\rm exc}\psi\rangle\right)^{1/2}\\
 &\le f\|N^{1/2}\varphi\|\,\|N^{1/2}\psi\|.
\end{align*}
This proves absolute convergence, continuity, and symmetry of the form.
Taking \(\varphi=\psi\) proves \eqref{eq:summed-N-bound}.
\end{proof}

\subsection{Gap improvement and closure of the bootstrap}

The following proposition closes the bootstrap by the continuation argument of \cite[Theorem~2.7]{NSY24}.

\begin{proposition}
\label{prop:abstract-bootstrap}
Use the finite-region Hamiltonians, bootstrap intervals, and spectral flow constructed in Section~\ref{sec:local-cutoffs}, and the moving-frame Hamiltonian \eqref{eq:K-nf}.
Fix \(M>0\), \(C_{\rm root}\ge1\), and \(J_M:=[-g/M,g/M]\).
Suppose there are bounded self-adjoint operators \(R_{\Lambda,\lambda,j}\in\pi_\at(\A_{X_j}^\st)\), with finite supports \(X_j\), such that, on \(I_\Lambda\cap J_M\), they annihilate their local vacua on both sides and satisfy
\[
 K_{\Lambda,\lambda}^{\nf}=K_0+\lambda\sum_jR_{\Lambda,\lambda,j}
\]
as a quadratic-form identity on \(\operatorname{Dom}(K_0^{1/2})\).
Assume that
\begin{equation}
 f_{\Lambda,\lambda}:=\sup_{x\in\Gamma}
       \sum_{j:\,x\in X_j}\|R_{\Lambda,\lambda,j}\|
 \le C_{\rm root}M
 \qquad(\Lambda\Subset\Gamma,\ \lambda\in J_M).
 \label{eq:abstract-rooted-normal-form}
\end{equation}
Then, on \(I_\Lambda\cap J_M\),
\[
 K_{\Lambda,\lambda}^{\nf}\ge(g-|\lambda|f_{\Lambda,\lambda})N.
\]
Moreover, with \(\lambda_{\rm boot}:=g/(4C_{\rm root}M)\), one has \([-\lambda_{\rm boot},\lambda_{\rm boot}]\subset I_\Lambda\) for every finite \(\Lambda\), and
\begin{equation}
 K_{\Lambda,\lambda}^{\nf}\ge(g-|\lambda|f_{\Lambda,\lambda})N
 \ge(g-C_{\rm root}|\lambda|M)N\ge\frac{3g}{4}N
 \qquad(|\lambda|\le\lambda_{\rm boot}).
 \label{eq:abstract-bootstrap-gap}
\end{equation}
These are quadratic-form inequalities on \(\operatorname{Dom}(K_0^{1/2})\).
In particular,
\[
 K_{\Lambda,\lambda}-E_{\Lambda,\lambda}
 \ge(g-|\lambda|f_{\Lambda,\lambda})(\1-P_{\Lambda,\lambda})
 \qquad(|\lambda|\le\lambda_{\rm boot}).
\]
\end{proposition}

\begin{proof}
Restrict the form estimate of Lemma~\ref{lem:summed-form-bound} to \(\operatorname{Dom}(K_0^{1/2})\subseteq\operatorname{Dom}(N^{1/2})\).
Together with \(K_0\ge gN\), it gives
\[
 K_{\Lambda,\lambda}^{\nf}
 \ge K_0-|\lambda|f_{\Lambda,\lambda}N
 \ge(g-|\lambda|f_{\Lambda,\lambda})N
\]
on \(I_\Lambda\cap J_M\).
Since \(C_{\rm root}\ge1\), the interval \([-\lambda_{\rm boot},\lambda_{\rm boot}]\) lies in \(J_M\), and on its intersection with \(I_\Lambda\) the coefficient is at least \(3g/4\).
By \eqref{eq:K-nf}, \eqref{eq:projection-transport}, and \(N\ge\1-P^\vac\), the Hamiltonian \(K_{\Lambda,\lambda}\) therefore has gap at least \(3g/4\) throughout this intersection.

If \(\lambda_\Lambda^+:=\sup I_\Lambda\le\lambda_{\rm boot}\), choose \(\lambda_0\in[0,\lambda_\Lambda^+)\) with \(\varepsilon:=(\lambda_\Lambda^+-\lambda_0)\|V_\Lambda\|<g/8\).
Bounded perturbation theory \cite[Chapter~IV, Section~3]{Kato} gives a simple ground-state eigenvalue of \(K_{\Lambda,\lambda_\Lambda^+}\) with gap at least \(3g/4-2\varepsilon>g/2\).
Hence \(\lambda_\Lambda^+\in\mathcal O_\Lambda\), contradicting maximality of \(I_\Lambda\).
The left endpoint is treated in the same way, proving the interval inclusion and \eqref{eq:abstract-bootstrap-gap}.

Conjugation by \(U_{\Lambda,\lambda}\), together with \(N\ge\1-P^\vac\) and \eqref{eq:projection-transport}, gives the last inequality.
\end{proof}

For the bootstrap in the exponential case, use the filter exponent \(1/2\) in Subsection~\ref{sec:compatible-filter-pair} and choose \(C_{\rm nf}\) as in \eqref{eq:explicit-C-nf} for this filter:
\begin{align}
 C_*&:=\left.C_{\rm nf}\right|_{\theta=1/2},
 \label{eq:physical-stability-constant}\\
 \lambda_*&:=\frac{g}{4C_*
                         \|\cV\|_{\mu_{\cV},1;2}}.
 \label{eq:lambda-star-bootstrap}
\end{align}
The rooted bound from Proposition~\ref{prop:uniform-normal-form} applies uniformly on the interval \(J_{\rm loc}\) in \eqref{eq:locality-coupling-ball}.
Since \(C_*\ge1\),
\begin{equation}
 I_*=[-\lambda_*,\lambda_*]\subset J_{\rm loc}.
 \label{eq:target-inside-locality-domain}
\end{equation}

\begin{corollary}
\label{cor:cutoff-gap-form}
For every finite \(\Lambda\), one has
\begin{equation}
 I_*\subset I_\Lambda.
 \label{eq:uniform-bootstrap-interval}
\end{equation}
For \(\lambda\in I_\Lambda\cap J_{\rm loc}\), the following form bound holds on \(\operatorname{Dom}(K_0^{1/2})\):
\begin{equation}
 K_{\Lambda,\lambda}^\nf
 \ge\left(g-C_*|\lambda|\|\cV\|_{\mu_{\cV},1;2}\right)N.
 \label{eq:gap-improvement-number}
\end{equation}
For \(\lambda\in I_*\), the cutoff Hamiltonian satisfies
\begin{equation}
 K_{\Lambda,\lambda}-E_{\Lambda,\lambda}
 \ge\left(g-C_*|\lambda|\|\cV\|_{\mu_{\cV},1;2}\right)
       (\1-P_{\Lambda,\lambda})
 \label{eq:gap-improvement}
\end{equation}
as a quadratic-form inequality on \(\operatorname{Dom}(K_0^{1/2})\).
In particular, \(K_{\Lambda,\lambda}\) has a simple ground state eigenvalue and gap at least \(3g/4\) on \(I_*\).
\end{corollary}

\begin{proof}
Proposition~\ref{prop:uniform-normal-form} gives \(f_{\Lambda,\lambda}\le C_*\|\cV\|_{\mu_{\cV},1;2}\).
Apply Proposition~\ref{prop:abstract-bootstrap} with \(M=\|\cV\|_{\mu_{\cV},1;2}\) and \(C_{\rm root}=C_*\), so that \(J_M=J_{\rm loc}\) and \(\lambda_{\rm boot}=\lambda_*\).
The proposition gives the interval inclusion and both form bounds.
\end{proof}

\section{The physical ground state and removal of the perturbation cutoff}
\label{sec:limit}

\subsection{The gap inequality for the physical state}

We now derive a gap inequality for the physical cutoff state \(\omega_{\Lambda,\lambda}\) from the spectral gap of \(K_{\Lambda,\lambda}\), using the covariance identity \eqref{eq:algebra-GNS-dynamics} and the state factorization \eqref{eq:cutoff-state-factorization}.
Figure~\ref{fig:physical-flow-descent} summarizes the maps and states in the two coordinate systems.

\begin{proposition}
\label{prop:decay-independent-descent}
Suppose \(\lambda\in I_\Lambda\) and \(K_{\Lambda,\lambda}-E_{\Lambda,\lambda}\ge g'(\1-P_{\Lambda,\lambda})\) for some \(g'>0\).
Then
\begin{equation}
 \operatorname{Re}\omega_{\Lambda,\lambda}
 (A^*\overline\delta_{\Phi_0+\lambda\cV_\Lambda}(A))
 \ge g'\bigl(\omega_{\Lambda,\lambda}(A^*A)
                     -|\omega_{\Lambda,\lambda}(A)|^2\bigr)
 \label{eq:physical-cutoff-gap}
\end{equation}
for every \(A\in\operatorname{Dom} (\overline\delta_{\Phi_0+\lambda\cV_\Lambda})\).
\end{proposition}

\begin{proof}
For \(A\in\operatorname{Dom}(\overline\delta_{\Phi_0+\lambda\cV_\Lambda})\), set \(B:=\pi_\at\sigma_\st(A\widehat\otimes\1)\).
By covariance \eqref{eq:atomic-cutoff-dynamics}--\eqref{eq:algebra-GNS-dynamics}, its orbit under \(K_{\Lambda,\lambda}\) is norm differentiable, with
\[
 [K_{\Lambda,\lambda},B]
 =\pi_\at\sigma_\st\bigl(
       \overline\delta_{\Phi_0+\lambda\cV_\Lambda}(A)
                              \widehat\otimes\1\bigr).
\]
Thus \(B\in C^1(K_{\Lambda,\lambda})\), so \(B\) maps \(\operatorname{Dom}(K_{\Lambda,\lambda})\) into itself.
Apply the spectral inequality to \((B-\omega_{\Lambda,\lambda}(A))U_{\Lambda,\lambda}\Omega_\at\), which belongs to that domain and is orthogonal to the ground vector by \eqref{eq:cutoff-state-factorization}.
Using \eqref{eq:cutoff-state-factorization} and the preceding commutator identity, we write this spectral inequality as \eqref{eq:physical-cutoff-gap}, as in \cite[(3.50)--(3.52)]{NSY24}.
\end{proof}

\subsection{Removing the cutoff from the spectral flow}

On \(I_*\), the physical cutoff state is the same for every filter constructed in Subsection~\ref{sec:compatible-filter-pair}.
The unitaries \(U_{\Lambda,\lambda}\) constructed with different filters all satisfy \eqref{eq:projection-transport}, by Proposition~\ref{prop:bootstrap-flow}.
Since \(P_{\Lambda,\lambda}\) has rank one, their vectors \(U_{\Lambda,\lambda}\Omega_\at\) differ only by a phase and define the same state in product-vacuum coordinates.
Equation~\eqref{eq:cutoff-state-factorization} shows that the corresponding physical states \(\omega_0\circ\beta_{\Lambda,0,\lambda}\) are also independent of the filter.
Indicating the filter exponent by a superscript, we obtain
\begin{equation}
 \omega_0\circ\beta_{\Lambda,0,\lambda}^{(\theta)}
 =\omega_0\circ\beta_{\Lambda,0,\lambda}^{(1/2)}
 =\omega_{\Lambda,\lambda},
 \qquad \lambda\in I_*,\quad 0<\theta<1.
 \label{eq:filter-independent-cutoff-state}
\end{equation}

Fix any \(0<\theta<1\) and use its filter \(W\) from Subsection~\ref{sec:compatible-filter-pair} in the following construction.
The interval \(I_*\) remains the one fixed by the bootstrap with exponent \(1/2\).
The interaction \(\eta_{\lambda}\) is defined by \eqref{eq:physical-generator-regrouping}.
Its cocycle is obtained as the limit of the physical cutoff flows.

\begin{lemma}
\label{lem:physical-generator}
Let \(\eta_{\lambda},\eta_{\Lambda,\lambda}\) be the interactions of Lemma~\ref{lem:inverse-filter-interaction}, with the rate \(\mu_\eta\) and constant \(C_\eta:=C_{\eta,2}\) from \eqref{eq:kappa-stretched-bound}.
Then
\begin{equation}
 \sup_{\lambda\in I_*}\|\eta_{\lambda}\|_{\mu_\eta,\theta;2},
 \quad
 \sup_{\substack{\Lambda\Subset\Gamma\\\lambda\in I_*}}
          \|\eta_{\Lambda,\lambda}\|_{\mu_\eta,\theta;2}
 \le C_\eta\|\cV\|_{\mu_{\cV},1;2}.
 \label{eq:physical-generator-stretched-bound}
\end{equation}
The path \((\eta_{\lambda})\) generates a unique strongly continuous, parity-preserving cocycle \((\beta_{\lambda,\lambda'})\), satisfying \eqref{eq:main-flow-locality}.
For every exhaustion \(\Lambda_n\nearrow\Gamma\) and observable \(A\),
\begin{equation}
 \sup_{\lambda,\lambda'\in I_*}
 \|\beta_{\Lambda_n,\lambda,\lambda'}(A)-\beta_{\lambda,\lambda'}(A)\|\stackrel{n\to\infty}{\longrightarrow}0.
 \label{eq:physical-flow-cutoff-limit}
\end{equation}
\end{lemma}

\begin{proof}
Lemma~\ref{lem:inverse-filter-interaction} gives \eqref{eq:physical-generator-stretched-bound} on the larger interval \(J_{\rm loc}\) and proves that \(\eta_{\Lambda,\lambda}\) generates \(\beta_{\Lambda,\lambda,\lambda'}\).
The common exponential bound and Proposition~\ref{prop:nonautonomous-local-limit-package} give convergence of \(\tau_t^{\ph,\Lambda_n,\lambda}(A)\) to \(\tau_t^{\ph,\lambda}(A)\), uniformly in \(\lambda\in I_*\) and in compact time intervals for every local \(A\).
Here we use its parameter-uniform conclusion for the constant paths in physical time generated by \(\Phi_0+\lambda\cV_{\Lambda_n}\), with the coupling \(\lambda\) as the additional parameter; each fixed interaction component eventually agrees with that of \(\Phi_0+\lambda\cV\), uniformly in \(\lambda\in I_*\).

For each nonempty finite output set \(Y\), the regrouping formulas \eqref{eq:physical-generator-regrouping}--\eqref{eq:physical-cutoff-generator-regrouping} contain only finitely many nonzero shells, whose input supports are subsets of \(Y\).
For each contributing \(X\subseteq Y\), \(\cV_{\Lambda_n}(X)=\cV(X)\) for all sufficiently large \(n\).
Compact-time convergence controls the filtered integrals on \(|t|\le T\), and the remaining contribution is bounded by \(2\|\cV(X)\|\int_{|t|>T}|W(t)|\,\dd t\).
Taking \(n\to\infty\) and then \(T\to\infty\) gives
\[
 \sup_{\lambda\in I_*}
 \|\eta_{\Lambda_n,\lambda}(Y)-\eta_\lambda(Y)\|
 \stackrel{n\to\infty}{\longrightarrow}0.
\]

Lemma~\ref{lem:componentwise-local-continuity}, with the uniform bound \eqref{eq:physical-generator-stretched-bound}, gives uniform convergence in every local weighted seminorm \eqref{eq:local-stretched-seminorm} with rate \(0<\nu<\mu_\eta\).
The same bound controls a polynomial interaction norm at any exponent \(p>\dvol\).
Since \(I_*\) is compact, uniform componentwise convergence implies \eqref{eq:nonautonomous-local-L1-convergence}, so Proposition~\ref{prop:nonautonomous-local-limit-package} proves \eqref{eq:physical-flow-cutoff-limit}.
The first support moment is bounded by the second, so \cite[Theorem~2.9 and Remark~2.10]{FermionicLR}, with spatial weight \(\E^{\mu_\eta R^\theta}\), gives \eqref{eq:main-flow-locality} for arbitrary observables, with uniform constants because
\[
 \int_{I_*}\|\eta_u\|_{\mu_\eta,\theta;2}\,\dd u
 \le2\lambda_*C_\eta\|\cV\|_{\mu_{\cV},1;2}
 =\frac{C_\eta g}{2C_*}. \qedhere
\]
\end{proof}

For a fixed filter \(W\), any local interaction representatives whose terms sum absolutely to \(\mathscr S_{\Lambda,\lambda}^\ph\) generate the same bounded cutoff cocycle \(\beta_{\Lambda,\lambda,\lambda'}\).
Under the uniform bounds and local convergence hypotheses of Proposition~\ref{prop:nonautonomous-local-limit-package}, their limiting cocycles agree as well.
The proof of Remark~\ref{rem:transport-regularity} uses this observation to compare the tracial-shell and lossless interaction representatives.

\subsection{The limiting physical ground state and its gap}

\begin{proposition}
\label{prop:gap-limit}
Let \(\varphi_n,\varphi\) be states and \(\Xi_n,\Xi\) even self-adjoint interactions generating strongly continuous dynamics.
Suppose \(\varphi_n\stackrel{n\to\infty}{\longrightarrow}\varphi\) pointwise on the algebra and \(\delta_{\Xi_n}(A)\stackrel{n\to\infty}{\longrightarrow}\delta_\Xi(A)\) in norm for every local \(A\).
If, for some \(g'>0\),
\[
 \operatorname{Re}\varphi_n(A^*\delta_{\Xi_n}(A))
 \ge g'\bigl(\varphi_n(A^*A)-|\varphi_n(A)|^2\bigr)
 \quad(A\in\A_{\loc}),
\]
then the same inequality holds for \(\varphi,\Xi\) on the local algebra.
If the local algebra is a graph core for the closed derivation generated by \(\Xi\) and its GNS vectors form an operator core as in Lemma~\ref{lem:local-form-core}, the inequality holds on the full closed-derivation domain and gives a ground state with GNS gap at least \(g'\) and a simple ground state eigenvalue.
\end{proposition}

\begin{proof}
The local gap inequality passes to the limit by state convergence and norm convergence of \(\delta_{\Xi_n}(A)\), as in \cite[proof of Theorem~2.8, (5.26)--(5.28)]{NSY24}.
Graph-norm approximation gives the inequality for every \(A\in\operatorname{Dom}(\overline\delta_\Xi)\).
Its nonnegative right-hand side shows that \(\varphi\) is a ground state.
Proposition~\ref{prop:form-spectrum}, with the assumed GNS core, gives a simple ground-state eigenvalue and gap at least \(g'\).
\end{proof}

For the physical cutoffs, rooted summability gives, for \(A\in\A_X\),
\begin{equation}
 \|(\delta_{\Phi_0+\lambda\cV_{\Lambda_n}}
           -\delta_{\Phi_0+\lambda\cV})(A)\|
 \le2|\lambda|\|A\|
       \sum_{\substack{Z\cap X\ne\varnothing\\Z\not\subset\Lambda_n}}
          \|\cV(Z)\|\stackrel{n\to\infty}{\longrightarrow}0.
 \label{eq:common-cutoff-derivation-limit}
\end{equation}
The omitted terms form the tail of an absolutely summable series.
This estimate applies in both decay regimes.

\begin{proof}[Proof of Theorem~\ref{thm:main}]
\emph{Exponential case.}
Suppose first that \(\theta_{\rm in}=1\).
Corollary~\ref{cor:cutoff-gap-form} shows that the spectral gap of \(K_{\Lambda,\lambda}\) is at least \(g_{\lambda}=g-C_*|\lambda|\|\cV\|_{\mu_{\cV},1;2}\ge3g/4\).
Proposition~\ref{prop:decay-independent-descent} transfers this inequality to \(\omega_{\Lambda,\lambda}\).
Using the filter exponent \(1/2\), Lemma~\ref{lem:physical-generator} gives pointwise convergence of these states to
\[
 \omega_{\lambda}:=\omega_0\circ\beta_{0,\lambda}^{(1/2)}.
\]
Proposition~\ref{prop:gap-limit}, \eqref{eq:common-cutoff-derivation-limit}, and the core properties in Lemma~\ref{lem:local-form-core} give \eqref{eq:main-gap} on the full closed-derivation domain and show that \(\omega_{\lambda}\) is a ground state with GNS gap at least \(g_{\lambda}\).
It is pure because \(\omega_0\) is pure and \(\beta_{0,\lambda}^{(1/2)}\) is an automorphism.

For any \(0<\theta<1\), apply Lemma~\ref{lem:physical-generator} with its corresponding filter on the same interval \(I_*\).
The constant \(C_*\) and the states \(\omega_{\lambda}\) were fixed with the filter of exponent \(1/2\).
Equation~\eqref{eq:filter-independent-cutoff-state} and the cocycle limit \eqref{eq:physical-flow-cutoff-limit} give, for every \(A\in\A\),
\[
 \omega_0(\beta_{0,\lambda}^{(\theta)}(A))
 =\lim_{n\to\infty}\omega_{\Lambda_n,\lambda}(A)=\omega_{\lambda}(A).
\]
Lemma~\ref{lem:inverse-filter-interaction} gives the generator bounds for every support moment at exponent \(\theta\), and Lemma~\ref{lem:physical-generator} gives the flow bound.
These constants depend on the data of Theorem~\ref{thm:main} and on \(\theta\); the interval \(I_*\) and the states \(\omega_{\lambda}\) do not depend on \(\theta\).

\emph{Stretched-exponential case.}
Let \(0<\theta_{\rm in}<1\) and equip \(\Gamma\) with the metric \(\rho(x,y):=d(x,y)^{\theta_{\rm in}}\).
The inequality \((u+v)^{\theta_{\rm in}} \le u^{\theta_{\rm in}}+v^{\theta_{\rm in}}\) proves the triangle inequality for \(\rho\).
Its balls satisfy
\[
 |\ball_R^\rho(x)|
 =|\ball_{R^{1/\theta_{\rm in}}}^d(x)|
 \le C_\Gamma(1+R)^{\dvol/\theta_{\rm in}}.
\]
Thus the preceding case applies with volume-growth exponent \(\dvol/\theta_{\rm in}\).
Write a superscript \(\rho\) on a norm when its distances are measured in this metric.
Then
\[
 \|\cV\|^\rho_{\mu_{\cV},1;2}
 =\|\cV\|_{\mu_{\cV},\theta_{\rm in};2},
 \qquad
 \|h\|_{\mu_h}^{\mathrm{Sch},\rho}
 =\|h\|_{\mu_h,\theta_{\rm in}}^{\rm Sch}=C_h.
\]
The algebra, the quadratic interaction, and its one-particle gap are unchanged.
Applying the exponential case in this metric gives an interval \(I_*\) and states \(\omega_{\lambda}\), both independent of \(\theta\), and the gap bound \eqref{eq:main-gap}, with \(C_*\) depending only on the data of Theorem~\ref{thm:main}.
Its baseline filter exponent \(1/2\) corresponds to spatial exponent \(\theta_{\rm in}/2\) in the original metric.
The closed derivation generated by \(\Phi_\lambda\) is the closure of the same operator on \(\A_{\loc}\) in both metrics, so its domain is unchanged.

For any \(0<\theta<\theta_{\rm in}\), apply the exponential case with locality exponent \(\theta/\theta_{\rm in}\) in the \(\rho\)-metric.
The generator bound has the desired spatial weight, since \((\diam_\rho X)^{\theta/\theta_{\rm in}}=(\diam_d X)^\theta\).
For the flow estimate at radius \(R\ge1\), use the integer radius \(\lfloor R^{\theta_{\rm in}}\rfloor\) in the \(\rho\)-metric.
Its neighborhood is contained in \(X^{(R)}\), and
\[
 \bigl\lfloor R^{\theta_{\rm in}}\bigr\rfloor^{\theta/\theta_{\rm in}}
 \ge2^{-\theta/\theta_{\rm in}}R^\theta.
\]
Multiplying the decay rate obtained in the \(\rho\)-metric by \(2^{-\theta/\theta_{\rm in}}\) gives \eqref{eq:main-flow-locality}; radius zero is covered by isometry.
Purity and state transport follow from the preceding case.
\end{proof}

\begin{proof}[Smoothness under the additional assumptions of Remark~\ref{rem:transport-regularity}]
Assume the cell condition, and choose \(m\in\N\) such that \(\cV(X)=0\) for \(|X|>m\).
Fix \(a>\dvol\), and choose the weight \(\kappa_0\) from \cite[Theorem~2.6]{FermionicLR}.
Then
\[
 \|\cV\|^{\#}_{a;\kappa_0}
 \le \E^{\kappa_0m}
 \left(\sup_{R\ge0}(1+R)^a
                 \E^{-\mu_{\cV}R^{\theta_{\rm in}}}\right)
       \|\cV\|_{\mu_{\cV},\theta_{\rm in};2}
 <\infty.
\]
The quadratic interaction \(\Phi_h\) has body order at most two, so its stretched-exponential bound likewise implies \(\|\Phi_h\|^{\#}_{a;\kappa_0}<\infty\).
To apply \cite[Corollary~6.1(iii)]{FermionicLR}, choose a strict common bound for the energy-interaction norms:
\[
 L>\|\Phi_h\|^{\#}_{a;\kappa_0}
       +\lambda_*\|\cV\|^{\#}_{a;\kappa_0}.
\]
The affine energy path \(\lambda\mapsto\Phi_h+\lambda\cV\) and the constant source \(\lambda\mapsto\cV\) are smooth in this norm, and the energy path has norm strictly less than \(L\) on \(I_*\).
Every filter \(W\) used for the transport in Theorem~\ref{thm:main} has all absolute time moments.
The filtered-flow result \cite[Corollary~6.1(iii)]{FermionicLR} therefore gives a coupling cocycle for which \(\lambda\mapsto\beta_{0,\lambda}(A)\) is norm-\(C^\infty\) for local \(A\).
With the same filter \(W\), the cutoff interaction representatives from \cite[Corollary~6.1(i)]{FermionicLR} sum in norm to \(\mathscr S_{\Lambda,\lambda}^\ph\), so their cocycles are \(\beta_{\Lambda,\lambda,\lambda'}\).
The cocycle obtained from \cite[Corollary~6.1(ii)]{FermionicLR} is therefore the norm limit on local observables of \(\beta_{\Lambda,\lambda,\lambda'}\).
Equation~\eqref{eq:physical-flow-cutoff-limit} gives the same limit for the cocycle in Theorem~\ref{thm:main}.
Thus the two cocycles coincide.
Applying \(\omega_0\) gives the same smoothness for \(\omega_{\lambda}(A)\).
\end{proof}

\section{The polynomial stability proof}
\label{sec:card-extension}

We prove Theorem~\ref{thm:card-stability} in the setting of Section~\ref{sec:local-cutoffs}.
The normal-form construction of Section~\ref{sec:normal-form} and the bootstrap, gap-transfer, and gap-limit arguments of Sections~\ref{sec:form-bound}--\ref{sec:limit} apply once we prove the polynomial filtering, locality, and convergence estimates below.
These estimates depend on the unperturbed interaction and the coordinate change through \(\|\Phi_0\|^{\#}_{p_\Phi;\kappa_0}\) and \(\|\Phi_{k_\st}\|^{\pol}_{p_\Phi;1}\), respectively.
We prove these bounds below for \(\Phi_0=\Phi_h\); Section~\ref{sec:product-reference-proof} applies the same estimates with the identity coordinate map.

We use the local representatives of evolved interactions constructed in \cite[Theorems~2.6--2.7]{FermionicLR}.
For each \(\cV_\Lambda(Z)\), these representatives decompose its filtered evolution into local terms supported in sets \(Y\supseteq Z\), satisfying the weighted bound \eqref{eq:card-filtered-source-total}.
The filtered generator interactions and their cutoff convergence follow from \cite[Corollary~6.1]{FermionicLR}, and the locality bound for their flows from \cite[Theorem~2.9]{FermionicLR}.
The physical perturbation terms are filtered first and then transformed by \(\sigma_{\Lambda,\lambda}\); the support-size norm controls the filtered interaction, and the locality bound for \(\sigma_{\Lambda,\lambda}\) controls its action on each local term.

Use the filters from Subsection~\ref{sec:compatible-filter-pair} with time-decay exponent \(1/2\), denoted by \(w_{\pol},W_{\pol}\).
The filtered maps and cutoff flows below are formed with these filters.
Their Fourier support properties and \(W_{\pol}'=w_{\pol}-\delta_0\) give the projection-transport identity \eqref{eq:projection-transport} and the commutator identity \eqref{eq:filtered-commutator}.
Moreover,
\[
 \int_\R(1+|t|)^m\bigl(|w_{\pol}(t)|+|W_{\pol}(t)|\bigr)\,\dd t
 <\infty\qquad(m\ge0),
\]
so both filters satisfy the moment assumptions of \cite[Theorem~2.7]{FermionicLR}.

\begin{proof}[Proof of Theorem~\ref{thm:card-stability}]
\phantomsection\label{proof:card-stability}
Set \(\Phi_0=\Phi_h\) and \(a:=p_\Phi\), and let \(\kappa_0=\kappa_0(\Dcell)\) be the weight from \cite[Theorem~2.6]{FermionicLR}.
Set
\[
 J_{\rm loc}
 :=\{\lambda\in\R\,|\,|\lambda|\|\cV\|^{\#}_{a;\kappa_0}\le g\}.
\]
Since \(\Phi_0(X)=0\) for \(|X|>2\), the quadratic-interaction bound \eqref{eq:quadratic-interaction-bounds} with \(k=0\) gives
\[
 \|\Phi_0\|^{\#}_{a;\kappa_0}
 \le 2r^2\E^{2\kappa_0}C_h.
\]
The triangle inequality and the definition of \(J_{\rm loc}\) give
\begin{equation}
 \|\Phi_0+\lambda\cV_\Lambda\|^{\#}_{a;\kappa_0},
 \quad
 \|\Phi_0+\lambda\cV\|^{\#}_{a;\kappa_0}
 \le \|\Phi_0\|^{\#}_{a;\kappa_0}+g
 \qquad(\lambda\in J_{\rm loc}).
 \label{eq:polynomial-energy-bound}
\end{equation}
This common energy bound is independent of \(\Lambda\) and \(\cV\).
Proposition~\ref{prop:exponential-stable-locality} bounds \(\|\Phi_{k_\st}\|^{\pol}_{a;1}\) in terms of the one-particle and geometric data.
We apply \cite[Theorem~2.7 and Corollary~6.1]{FermionicLR} with \(\lambda\in J_{\rm loc}\) as an external parameter and these energy interactions constant in physical time.

\paragraph{Physical generators and their cutoff limits.}
For a time-independent energy interaction \(\Phi\), the maps of \cite[Definition~2.5 and Theorem~2.6]{FermionicLR} assign to each nonempty \(X\Subset\Gamma\) and \(A\in\A_X^+\) a local operator family \(\cT^\Phi_{0,t;X}A\) supported on sets containing \(X\), with
\[
 \tau_t^\Phi(A)
 =\sum_{Y\supseteq X}(\cT^\Phi_{0,t;X}A)(Y).
\]
This series converges absolutely in norm, and the induced map on source interactions is
\[
 (\cT^\Phi_{0,t}\Psi)(Y)
 :=\sum_{\varnothing\ne X\subseteq Y}
       (\cT^\Phi_{0,t;X}\Psi(X))(Y).
\]
The absolute rooted estimates in that theorem justify grouping the terms in this way.
For the full physical path and its cutoffs, use the common Hamiltonian bound \eqref{eq:polynomial-energy-bound} to choose the same cell scales, partitions, and time cutoffs, as in \cite[Theorem~2.7 and Corollary~6.1]{FermionicLR}.
With the maps \(\cT^\Phi_{0,t}\) constructed from these choices, put
\begin{align*}
 \eta_{\Lambda,\lambda}(Y)
 &:=\int_\R W_{\pol}(t)
   \bigl(\cT^{\Phi_0+\lambda\cV_\Lambda}_{0,t}
                         \cV_\Lambda\bigr)(Y)\,\dd t,\\
 \eta_{\lambda}(Y)
 &:=\int_\R W_{\pol}(t)
   \bigl(\cT^{\Phi_0+\lambda\cV}_{0,t}
                         \cV\bigr)(Y)\,\dd t.
\end{align*}
Since \(W_{\pol}\) has all absolute time moments, \cite[Corollary~6.1]{FermionicLR} shows that \(\eta_{\Lambda,\lambda}\) and \(\eta_{\lambda}\) are even, self-adjoint, componentwise norm-continuous interactions satisfying
\begin{equation}
 \|\eta_{\Lambda,\lambda}\|^{\pol}_{a;k}
 +\|\eta_{\lambda}\|^{\pol}_{a;k}
 \le C_k\|\cV\|^{\#}_{a;\kappa_0}
 \qquad(\lambda\in J_{\rm loc},\ k\ge0).
 \label{eq:card-all-generators}
\end{equation}
These interactions generate unique coupling cocycles; denote the full cocycle by \(\beta_{\lambda,\lambda'}\).

At fixed \(\Lambda\), the exact representation identity in \cite[Corollary~6.1(i)]{FermionicLR} gives
\begin{equation}
 \sum_Y\eta_{\Lambda,\lambda}(Y)
 =\cI^\ph_{\Lambda,\lambda}(V_\Lambda^\ph)
 =\mathscr S^\ph_{\Lambda,\lambda}
 \label{eq:card-cutoff-generator-identities}
\end{equation}
with absolute convergence in norm.
By uniqueness in Proposition~\ref{prop:nonautonomous-local-limit-package}, these interactions generate the inner flow \(\beta_{\Lambda,\lambda,\lambda'}\) used in the normal-form construction.

For any exhaustion \(\Lambda_n\nearrow\Gamma\), the physical source terms \(\cV_{\Lambda_n}(X)\) equal \(\cV(X)\) eventually for each finite \(X\).
The energy and source paths therefore converge componentwise, uniformly on \(J_{\rm loc}\).
The convergence conclusion of \cite[Corollary~6.1]{FermionicLR} therefore gives, for every nonempty finite \(X\) and \(0\le q<a\),
\[
 \sup_{\lambda\in J_{\rm loc}}
 \|\eta_{\Lambda_n,\lambda}-\eta_{\lambda}\|^{\pol,\loc}_{q;X}
 \stackrel{n\to\infty}{\longrightarrow}0,
\]
with the seminorm \eqref{eq:polynomial-local-seminorm}, and
\begin{equation}
 \sup_{\lambda,\lambda'\in J_{\rm loc}}
 \|\beta_{\Lambda_n,\lambda,\lambda'}(A)-\beta_{\lambda,\lambda'}(A)\|
 \stackrel{n\to\infty}{\longrightarrow}0\qquad(A\in\A).
 \label{eq:polynomial-flow-cutoff-limit}
\end{equation}
For every local \(A\), the same corollary gives
\[
 \sum_Y[\eta_{\lambda}(Y),A]
 =\sum_X[\cI^\ph_{\lambda}(\cV(X)),A],
\]
with absolute convergence of both commutator series, proving \eqref{eq:spectral-flow-derivation}.

\paragraph{Localization of the centered terms.}
Equation~\eqref{eq:card-all-generators}, with \(k=1\), gives
\[
 \int_{J_{\rm loc}}\|\eta_{\Lambda,u}\|^{\pol}_{a;1}\,\dd u
 \le 2C_1g.
\]
Using this integrated generator bound and the bound on \(\|\Phi_{k_\st}\|^{\pol}_{a;1}\), Lemma~\ref{lem:uniform-alpha-locality} with \(v(R)=(1+R)^a\) gives locality of \(\sigma_{\Lambda,\lambda}\), uniformly in the cutoff and perturbation strength.

For fixed \(\Lambda,\lambda,Z\), the filtered representatives of the physical source are
\[
 F_Z(Y):=\int_\R w_{\pol}(t)
 \bigl(\cT^{\Phi_0+\lambda\cV_\Lambda}_{0,t;Z}
                         \cV_\Lambda(Z)\bigr)(Y)\,\dd t.
\]
They vanish unless \(Y\supseteq Z\), and \(F_{\Lambda,\lambda,Z}=\sum_{Y\supseteq Z}F_Z(Y)\) in norm.
Since \(w_{\pol}\) has all absolute time moments, the filtered single-source estimate in \cite[Theorem~2.7(iii)]{FermionicLR}, with \(p=a\) and \(k=1\), gives
\begin{equation}
 \sum_{Y\supseteq Z}|Y|(1+\diam Y)^a\|F_Z(Y)\|
 \le C\E^{\kappa_0|Z|}(1+\diam Z)^a
                         \|\cV_\Lambda(Z)\|.
 \label{eq:card-filtered-source-total}
\end{equation}

For \(R\ge0\), define a local approximant by
\[
 B_R:=\sum_{\substack{Y\supseteq Z\\\diam Y\le R/2}}
 \F_{Y^{(R/2)}}^{\rm tr}
             \bigl(\sigma_{\Lambda,\lambda}(F_Z(Y)\widehat\otimes\1)\bigr).
\]
Every set \(Y\) in this sum lies in the finite set \(Z^{(R/2)}\), so the sum is finite.
Moreover, \(Y^{(R/2)}\subset Z^{(R)}\), and hence \(B_R\in\A_{Z^{(R)}}^\st\).
For \(\diam Y\le R/2\), apply the locality bound \eqref{eq:uniform-alpha-locality} for \(\sigma_{\Lambda,\lambda}\) to \(F_Z(Y)\widehat\otimes\1\) with radius \(R/2\).
For \(\diam Y>R/2\), use \(\|\sigma_{\Lambda,\lambda}(F_Z(Y)\widehat\otimes\1)\|=\|F_Z(Y)\|\).
Summing gives
\begin{align*}
 \|\sigma_{\Lambda,\lambda}(F_{\Lambda,\lambda,Z}\widehat\otimes\1)-B_R\|
 &\le C(1+R)^{-a}\sum_Y|Y|\|F_Z(Y)\|
        +\sum_{\diam Y>R/2}\|F_Z(Y)\|\\*
 &\le C(1+R)^{-a}\sum_Y|Y|(1+\diam Y)^a\|F_Z(Y)\|.
\end{align*}
In the second line we used \(1+\diam Y>(1+R)/2\) in the omitted sum.
By \eqref{eq:G-Lambda-s}, \(G_{\Lambda,\lambda}(Z)=\sigma_{\Lambda,\lambda}(F_{\Lambda,\lambda,Z}\widehat\otimes\1) -\omega_{\Lambda,\lambda}(F_{\Lambda,\lambda,Z})\1\).
Thus \(B_R-\omega_{\Lambda,\lambda}(F_{\Lambda,\lambda,Z})\1\) approximates \(G_{\Lambda,\lambda}(Z)\) with the same error.
Applying \eqref{eq:card-filtered-source-total} and the vacuum-expectation bound \eqref{eq:vacuum-expectation-approximation} gives
\begin{equation}
 \|G_{\Lambda,\lambda}(Z)
       -\EE^\vac_{Z^{(R)}}(G_{\Lambda,\lambda}(Z))\|
 \le C\E^{\kappa_0|Z|}(1+\diam Z)^a
       \|\cV_\Lambda(Z)\|(1+R)^{-a}.
 \label{eq:card-centered-locality}
\end{equation}
In Proposition~\ref{prop:rooted-vacuum-normal-form}, take
\[
 c_\Lambda(Z)=C\E^{\kappa_0|Z|}(1+\diam Z)^a
                   \|\cV_\Lambda(Z)\|,
 \qquad q_G(R)=(1+R)^{-a}.
\]
The bound \(\|G_{\Lambda,\lambda}(Z)\|\le2\|w_{\pol}\|_1\|\cV_\Lambda(Z)\|\) gives \(\|G_{\Lambda,\lambda}(Z)\|\le c_\Lambda(Z)\) after choosing \(C\ge2\|w_{\pol}\|_1\).
The total coefficient sum \(\sum_Zc_\Lambda(Z)\) is finite.
The definition \eqref{eq:card-norm} gives
\[
 \sup_{x\in\Gamma}\sum_{Z\ni x}c_\Lambda(Z)
 =C\|\cV_\Lambda\|^{\#}_{a;\kappa_0}
 \le C\|\cV\|^{\#}_{a;\kappa_0}.
\]
The coefficients \(b_n\) in Proposition~\ref{prop:rooted-vacuum-normal-form} satisfy \(b_n\le C(1+R_n)^{-a}\).
Since \(a>\dvol\) and the radii \(R_n\) grow dyadically, this gives \(\sum_n(1+R_n)^{\dvol}b_n<\infty\).
We may therefore choose \(C_*\ge1\), depending only on \(a,g\), the geometry, the onsite dimension, \(\|\Phi_0\|^{\#}_{a;\kappa_0}\), and \(\|\Phi_{k_\st}\|^{\pol}_{a;1}\), such that
\begin{equation}
 f_{\Lambda,\lambda}\le C_*\|\cV\|^{\#}_{a;\kappa_0},
 \qquad \lambda\in J_{\rm loc}.
 \label{eq:card-normal-form}
\end{equation}
Proposition~\ref{prop:rooted-vacuum-normal-form} also gives \(R_{\Lambda,\lambda}(Z,n)P_{Z^{(R_n)}}^\vac =P_{Z^{(R_n)}}^\vac R_{\Lambda,\lambda}(Z,n)=0\) on \(I_\Lambda\cap J_{\rm loc}\).
Apply Proposition~\ref{prop:abstract-bootstrap} with \(M=\|\cV\|^{\#}_{a;\kappa_0}\) and \(C_{\rm root}=C_*\).
It gives \(I_*\subset I_\Lambda\) and, for \(\lambda\in I_*\), a spectral gap of \(K_{\Lambda,\lambda}\) at least \(g-C_*|\lambda|\|\cV\|^{\#}_{a;\kappa_0}\).

\paragraph{The ground-state gap and locality at exponent \(a\).}
Proposition~\ref{prop:decay-independent-descent} transfers the cutoff gap inequality to \(\omega_{\Lambda,\lambda}=\omega_0\circ\beta_{\Lambda,0,\lambda}\).
For every \(A\in\A\), the cocycle convergence \eqref{eq:polynomial-flow-cutoff-limit} gives \(\omega_{\Lambda_n,\lambda}(A)\stackrel{n\to\infty}{\longrightarrow}\omega_{\lambda}(A)\), uniformly for \(\lambda\in I_*\), where \(\omega_{\lambda}:=\omega_0\circ\beta_{0,\lambda}\).
The state \(\omega_{\lambda}\) is pure, as in the proof of Theorem~\ref{thm:main}.
For local \(A\), \eqref{eq:common-cutoff-derivation-limit} gives \(\delta_{\Phi_0+\lambda\cV_{\Lambda_n}}(A)\stackrel{n\to\infty}{\longrightarrow}\delta_{\Phi_{\lambda}}(A)\) in norm.
Proposition~\ref{prop:gap-limit} therefore passes the physical cutoff gap inequality \eqref{eq:physical-cutoff-gap} to \(\omega_{\lambda}\) on \(\A_{\loc}\).
The support-size bound on the energy interaction gives \(\|\Phi_{\lambda}\|^{\pol}_{a;1}<\infty\), with \(a>\dvol\).
The graph-core property in Lemma~\ref{lem:local-form-core} extends this inequality to \eqref{eq:polynomial-main-gap} for every \(\lambda\in I_*\) and \(A\in\operatorname{Dom}(\overline\delta_{\Phi_{\lambda}})\).

Equation~\eqref{eq:card-all-generators} proves \eqref{eq:card-generator}.
Moreover, with \(C_\eta:=C_{\eta,1}\),
\[
 |I_*|\sup_{\lambda\in I_*}\|\eta_{\lambda}\|^{\pol}_{a;1}
 \le \frac{gC_\eta}{2C_*}.
\]
Let \(\beta^{X,R}_{\lambda,\lambda'}\) be the cocycle generated by the terms \(\eta_u(Y)\) with \(Y\subset X^{(R)}\).
For \(A\in\A_X\), its value \(\beta^{X,R}_{\lambda,\lambda'}(A)\) belongs to \(\A_{X^{(R)}}\), and \cite[Theorem~2.9]{FermionicLR} gives
\[
 \|\beta_{\lambda,\lambda'}(A)-\beta^{X,R}_{\lambda,\lambda'}(A)\|
 \le |X|\|A\|\bigl(\E^{gC_\eta/C_*}-1\bigr)(1+R)^{-a}.
\]
Thus \eqref{eq:card-flow} holds with \(C_\beta:=\E^{gC_\eta/C_*}-1\), for arbitrary local observables and both parameter orderings.
The propagation constants depend only on \(a\), the geometric data, and the common energy bound \eqref{eq:polynomial-energy-bound} for \(\Phi_0+\lambda\cV\).
Together with the fixed quadratic path in Proposition~\ref{prop:exponential-stable-locality}, these bounds show that \(C_*\), \(C_\eta\), and \(C_\beta\) depend only on the data listed in Theorem~\ref{thm:card-stability}.

The affine path \(\lambda\mapsto\Phi_0+\lambda\cV\) and the constant source path \(\lambda\mapsto\cV\) are \(C^\infty\) in \(\|\cdot\|^{\#}_{a;\kappa_0}\) on \(\R\).
The definition of \(I_*\) and \(C_*\ge1\) give
\[
 \sup_{\lambda\in I_*}\|\Phi_0+\lambda\cV\|^{\#}_{a;\kappa_0}
 \le\|\Phi_0\|^{\#}_{a;\kappa_0}+g/4
 <\|\Phi_0\|^{\#}_{a;\kappa_0}+g.
\]
Thus the path lies strictly within the common energy bound \eqref{eq:polynomial-energy-bound} used to choose the maps \(\cT^\Phi_{0,t}\).
Since \(W_{\pol}\) has all absolute time moments, \cite[Corollary~6.1(iii)]{FermionicLR} gives \(C^\infty\) dependence of \(\eta_{\lambda}\) in every \(\|\cdot\|^{\pol}_{a;k}\)-norm and of \(\beta_{\lambda_0,\lambda}(A)\) in operator norm for local \(A\) and fixed \(\lambda_0\in I_*\).
Applying the bounded functional \(\omega_0\) shows that \(\lambda\mapsto\omega_{\lambda}(A)\) is \(C^\infty\).

For the stronger support-size estimate \eqref{eq:card-generator-support-weight}, apply \cite[Theorem~2.7(vi)]{FermionicLR} to the same physical interaction \(\Phi_0+\lambda\cV\), source \(\cV\), and filter \(W_{\pol}\), with \(p=a\) and time-decay exponent \(\alpha=1/2\).
Equation~\eqref{eq:W-stretched-time} gives, for some \(b>0\),
\[
 \int_\R |W_{\pol}(t)|\E^{b(1+|t|)^{1/2}}\,\dd t<\infty.
\]
The support-size exponent given by that theorem is \((1+2\dvol(a+1)/(a-\dvol))^{-1}=\theta_{\#}\).
The common energy bound on \(J_{\rm loc}\) therefore gives \eqref{eq:card-generator-support-weight} for the already chosen representatives \(\eta_{\lambda}\), with \(c_{\#},C_{\#}\) depending only on the data listed in Theorem~\ref{thm:card-stability}.
\end{proof}

\begin{proof}[Proof of Remark~\ref{rem:transport-symmetry}]
\phantomsection\label{proof:transport-symmetry}
Suppose first that \(\zeta\) is linear.
Invariance of the interaction terms implies that \(\zeta\) commutes with the energy dynamics.
It also preserves the unique tracial state and every local algebra, so uniqueness of the trace-preserving conditional expectation gives \(\zeta\circ\F_Y^{\rm tr}=\F_Y^{\rm tr}\circ\zeta\).
For Theorem~\ref{thm:main}, each term in \eqref{eq:physical-generator-regrouping} is therefore invariant.
For Theorem~\ref{thm:card-stability}, covariance of the local representatives in \cite[Remark~4.5]{FermionicLR} gives termwise invariance after filtering.
Uniqueness of the coupling cocycle proves \(\zeta\circ\beta_{\lambda,\lambda'}=\beta_{\lambda,\lambda'}\circ\zeta\).

For a conjugate-linear \(\zeta\), invariance of the cutoff interaction gives
\[
 \zeta\circ\tau_t^{\ph,\Lambda,\lambda}
 =\tau_{-t}^{\ph,\Lambda,\lambda}\circ\zeta,
 \qquad
 \zeta(\mathscr S_{\Lambda,\lambda}^\ph)
 =-\mathscr S_{\Lambda,\lambda}^\ph.
\]
The second identity follows from \eqref{eq:physical-inverse-filter}--\eqref{eq:physical-cutoff-generator} by changing \(t\) to \(-t\), since \(W\) is real and odd.
As \(\zeta(\I)=-\I\), the defining equation \(\partial_\lambda u_{\Lambda,\lambda}=-\I\mathscr S_{\Lambda,\lambda}^\ph u_{\Lambda,\lambda}\), with \(u_{\Lambda,0}=\1\), implies \(\zeta(u_{\Lambda,\lambda})=u_{\Lambda,\lambda}\) by uniqueness.
Thus \(\zeta\) commutes with \(\beta_{\Lambda,\lambda,\lambda'}\).
Passing to the cutoff limit gives the same commutation relation for \(\beta_{\lambda,\lambda'}\), by \eqref{eq:physical-flow-cutoff-limit} in the stretched-exponential case and \eqref{eq:polynomial-flow-cutoff-limit} in the polynomial case.
\end{proof}

\section{Stability of interacting systems with product ground states}
\label{sec:product-reference-proof}

We first prove the GNS gap from the excitation-number bound, then apply the normal-form argument with the identity coordinate map in each decay regime.

\begin{proof}[Proof of Theorem~\ref{thm:product-reference-stability}]
\phantomsection\label{proof:product-reference-stability}
The projections \(\pi_0(P_x)\) commute.
If \(A\in\A_X\), \(\pi_0(\1-P_x)\pi_0(A)\Omega_0=0\) for \(x\notin X\), since \(P_x\) is even, \([P_x,A]=0\), and \(\pi_0(P_x)\Omega_0=\Omega_0\).
Thus the excitation form \eqref{eq:product-reference-number} has a dense domain containing all local GNS vectors.
Its increasing finite sums define a closed form and hence the positive self-adjoint operator \(N\).

For finite \(\Lambda\), put \(P_\Lambda:=\prod_{x\in\Lambda}\pi_0(P_x)\).
If \(\Lambda\supset X\), the rank-one support projections and the product state give
\[
 P_\Lambda\pi_0(A)\Omega_0=\omega_0(A)\Omega_0.
\]
Since local GNS vectors are dense, \(P_\Lambda\) converges strongly along any exhaustion to \(|\Omega_0\rangle\langle\Omega_0|\).
The commuting-projection inequality
\[
 \sum_{x\in\Lambda}\pi_0(\1-P_x)\ge\1-P_\Lambda
\]
therefore proves
\[
 N\ge\1-|\Omega_0\rangle\langle\Omega_0|,
 \qquad \ker N=\C\Omega_0.
\]
The product state is pure: an operator in the commutant of \(\pi_0(\A)\) commutes with the strong limit \(|\Omega_0\rangle\langle\Omega_0|\), so it is scalar on \(\Omega_0\), and cyclicity makes it scalar everywhere.
Since \(K_0\Omega_0=0\), assumption \eqref{eq:product-reference-coercivity} gives a simple ground state eigenvalue zero and a gap at least \(g\).

With the identity coordinate map specified in Section~\ref{sec:local-cutoffs}, the cutoff Hamiltonian is
\[
 K_{\Lambda,\lambda}=K_0+\lambda\pi_0(V_\Lambda^\ph).
\]
Lemma~\ref{lem:operator-facts} and Proposition~\ref{prop:bootstrap-flow} construct the candidate spectral flow and prove the ground-state transport identity \eqref{eq:projection-transport} on \(I_\Lambda\).
The even rank-one product projections \(P_x\) allow us to apply Lemmas~\ref{lem:vacuum-localization} and~\ref{lem:one-local-term}.
The bounds on \(N\), together with \eqref{eq:product-reference-coercivity}, give the excitation-number and form-domain hypotheses required by Lemma~\ref{lem:summed-form-bound} and Proposition~\ref{prop:abstract-bootstrap}.
The dynamics of the unperturbed interaction \(\Phi_0\) is implemented by \(K_0\), as required in Lemma~\ref{lem:operator-facts}(b).
It remains to prove the centered-term locality bounds required by Proposition~\ref{prop:rooted-vacuum-normal-form} in the two decay regimes.

\emph{Part (a).}
First let \(\theta_{\rm in}=1\), choose \(\mu:=\frac12\min\{\mu_0,\mu_{\cV}\}\), and use the baseline filter exponent \(1/2\).
On the interval \(J_{\rm loc}\) in \eqref{eq:locality-coupling-ball},
\[
 \|\Phi_0+\lambda\cV_\Lambda\|_{\mu,1;2},
 \quad \|\Phi_0+\lambda\cV\|_{\mu,1;2}
 \le C_0+g.
\]
Apply Lemma~\ref{lem:weighted-filter-locality} with the common energy-interaction bound \(C_0+g\), choosing \(v_{\rm LR}=2(C_0+g)/\mu\) in its proof.
Lemma~\ref{lem:inverse-filter-interaction} bounds the generator and its integrated norm, and Lemma~\ref{lem:uniform-alpha-locality} applies with \(\Phi_{k_\st}=0\).
Proposition~\ref{prop:transported-D-locality} gives the centered-term bound.
Use the resulting uniform locality constants to evaluate \(C_*\) in \eqref{eq:physical-stability-constant}.
Propositions~\ref{prop:uniform-normal-form} and~\ref{prop:abstract-bootstrap} then give \(I_*\subset I_\Lambda\) and the linear gap bound \eqref{eq:gap-improvement} at every cutoff.
All these constants depend only on the data in part~(a); this uses that the expectation maps have norm one independently of the onsite state.

The ground-vector state of \(K_{\Lambda,\lambda}\) is already the physical state \(\omega_{\Lambda,\lambda}=\omega_0\circ\beta_{\Lambda,0,\lambda}\).
Proposition~\ref{prop:decay-independent-descent} gives the physical gap inequality \eqref{eq:physical-cutoff-gap}.
Lemma~\ref{lem:physical-generator} gives pointwise convergence of \(\omega_{\Lambda_n,\lambda}\) to \(\omega_\lambda:=\omega_0\circ\beta_{0,\lambda}\).
Together with \eqref{eq:common-cutoff-derivation-limit}, Proposition~\ref{prop:gap-limit} gives the gap inequality for \(\omega_\lambda\) on the local algebra.
Lemma~\ref{lem:local-form-core} extends it to \(\operatorname{Dom}(\overline\delta_{\Phi_0+\lambda\cV})\).
The state \(\omega_\lambda\) is pure because \(\omega_0\) is pure and \(\beta_{0,\lambda}\) is an automorphism.
For every other filter exponent \(0<\theta<1\), use Lemma~\ref{lem:physical-generator} on the same interval.
Equation~\eqref{eq:filter-independent-cutoff-state}, together with convergence of \(\beta_{\Lambda_n,0,\lambda}(A)\), shows that \(\omega_0\circ\beta_{0,\lambda}(A)\) is independent of the filter for every local \(A\).

For \(0<\theta_{\rm in}<1\), use the metric \(\rho=d^{\theta_{\rm in}}\) as in the proof of Theorem~\ref{thm:main}.
The volume-growth exponent becomes \(\dvol/\theta_{\rm in}\), and the norms of part~(a) satisfy
\[
 \|\Phi_0\|^\rho_{\mu_0,1;2}
 =\|\Phi_0\|_{\mu_0,\theta_{\rm in};2},
 \qquad
 \|\cV\|^\rho_{\mu_{\cV},1;2}
 =\|\cV\|_{\mu_{\cV},\theta_{\rm in};2}.
\]
The state, \(N\), \(K_0\), and the coercivity assumption are unchanged.
Applying the case \(\theta_{\rm in}=1\) in the metric \(\rho\), and using radius \(\lfloor R^{\theta_{\rm in}}\rfloor\) as in the proof of Theorem~\ref{thm:main}, gives the stated locality estimates for every \(0<\theta<\theta_{\rm in}\).

\emph{Part (b).}
Use the interval \(J_{\rm loc}\) from Section~\ref{sec:card-extension}.
The hypothesis on \(\Phi_0\) gives
\[
 \|\Phi_0+\lambda\cV_\Lambda\|^{\#}_{p_\Phi;\kappa_0},
 \quad
 \|\Phi_0+\lambda\cV\|^{\#}_{p_\Phi;\kappa_0}
 \le C_0+g
 \qquad(\lambda\in J_{\rm loc}).
\]
Here \(\sigma_\st=\id\), so we take \(\Phi_{k_\st}=0\).
The estimates of Section~\ref{sec:card-extension} therefore apply with Hamiltonian norm bound \(C_0+g\), with constants depending only on the data in part~(b).
The generator bounds and exact cutoff identity \eqref{eq:card-all-generators}--\eqref{eq:card-cutoff-generator-identities}, centered-term estimate \eqref{eq:card-centered-locality}, and rooted bound \eqref{eq:card-normal-form} hold with \(\sigma_\st=\id\).
Proposition~\ref{prop:abstract-bootstrap} gives the stated interval and cutoff gap.
The cutoff states \(\omega_{\Lambda,\lambda}\) are physical, and the cocycle convergence \eqref{eq:polynomial-flow-cutoff-limit} gives their pointwise convergence to \(\omega_\lambda:=\omega_0\circ\beta_{0,\lambda}\).
Together with \eqref{eq:common-cutoff-derivation-limit}, Proposition~\ref{prop:gap-limit} gives the gap inequality on the local algebra, and Lemma~\ref{lem:local-form-core} extends it to \(\operatorname{Dom}(\overline\delta_{\Phi_0+\lambda\cV})\).
The state \(\omega_\lambda\) is pure because \(\omega_0\) is pure and \(\beta_{0,\lambda}\) is an automorphism.

The argument in Section~\ref{sec:card-extension} gives the approximation estimate \eqref{eq:card-flow} for \(\beta_{\lambda,\lambda'}(A)\), with error decaying as \((1+R)^{-p_\Phi}\).
It also gives estimate \eqref{eq:card-generator-support-weight} for the generator \(\eta_\lambda\), including the stretched-exponential weight in the support size.
For smoothness, the strict bound needed in \cite[Corollary~6.1(iii)]{FermionicLR} follows from
\[
 \sup_{\lambda\in I_*}\|\Phi_0+\lambda\cV\|^{\#}_{p_\Phi;\kappa_0}
 \le C_0+g/4<C_0+g.
\]
The maps \(\lambda\mapsto\Phi_0+\lambda\cV\) and \(\lambda\mapsto\cV\) are \(C^\infty\) in \(\|\cdot\|^{\#}_{p_\Phi;\kappa_0}\), and \(W_{\pol}\) has every absolute polynomial moment.
Thus \cite[Corollary~6.1(iii)]{FermionicLR} gives the stated smoothness of \(\eta_\lambda\), \(\beta_{\lambda_0,\lambda}(A)\), and \(\omega_\lambda(A)\), for local \(A\) and fixed \(\lambda_0\in I_*\).

For the spin-system formulation, use the ordinary product-state conditional expectations \(\EE_X^\vac=\id_{\A_X}\otimes\omega_0|_{\A_{X^c}}\) and the tracial expectations on tensor-product algebras.
The product-state expectations are contractive, \(\A_X\)-bimodular, and preserve \(\omega_0\); together with the rank-one product projections \(P_x\), these properties give the vacuum-localization and excitation-number estimates of Lemmas~\ref{lem:vacuum-localization} and~\ref{lem:summed-form-bound}.
The normal-form and bootstrap arguments therefore apply.
The locality estimates used in both parts hold for spin systems by \cite[Remark~6.6]{FermionicLR}, and the graph-core proof in Lemma~\ref{lem:local-form-core} applies with the same commutator estimates.
The cutoff limits and the preceding argument therefore give all stated conclusions on the spin algebra.
\end{proof}

\appendix

\section{Auxiliary geometric and analytic results}
\label{app:basic-background}

This appendix collects the geometric estimates, non-autonomous CAR dynamics and local limits, ground-state form and core results, and weighted one-particle functional calculus used in the proof.

\subsection{CAR conventions and metric-space geometry}
\label{app:CAR-geometry}

Inner products are linear in the second variable and the annihilation map \(f\mapsto a(f)\) is antilinear.
When labels \(u=(x,i)\) and \(v=(y,j)\) include indices for the fermionic modes, Majorana generators, or the two copies, we set \(d(u,v):=d(x,y)\).

For a parity-homogeneous observable \(A\), let \(|A|\in\{0,1\}\) be defined by \(\Theta(A)=(-1)^{|A|}A\).
Whenever two CAR systems \(\A_j=\CAR(\mathfrak h_j)\) are stacked, let \(\jmath_j\) denote their canonical embeddings into \(\CAR(\mathfrak h_1\oplus\mathfrak h_2)\).
We set
\begin{equation*}
 \A_1\widehat\otimes\A_2
 :=\CAR(\mathfrak h_1\oplus\mathfrak h_2),
 \qquad
 A\widehat\otimes B:=\jmath_1(A)\jmath_2(B).
\end{equation*}
The linear span of these simple tensors is dense.
For homogeneous factors, multiplication and involution are given by
\begin{align*}
 (A_1\widehat\otimes B_1)(A_2\widehat\otimes B_2)
 &=(-1)^{|B_1||A_2|}A_1A_2\widehat\otimes B_1B_2,\\
 (A\widehat\otimes B)^*
 &=(-1)^{|A||B|}A^*\widehat\otimes B^*
\end{align*}
and these formulas extend linearly and continuously.
For parity-preserving automorphisms \(\alpha_j\) of \(\A_j\), their graded tensor product is the automorphism determined by
\[
 (\alpha_1\widehat\otimes\alpha_2)(A\widehat\otimes B)
 :=\alpha_1(A)\widehat\otimes\alpha_2(B).
\]
Since \(\alpha_1,\alpha_2\) preserve parity, the assignment
\[
 a(f\oplus g)\longmapsto
 \jmath_1(\alpha_1(a(f)))+\jmath_2(\alpha_2(a(g)))
\]
preserves the CAR for \(f\in\mathfrak h_1\) and \(g\in\mathfrak h_2\).
By \cite[Theorem~5.2.5]{BR97}, this assignment extends to a \(*\)-endomorphism of \(\A_1\widehat\otimes\A_2\).
The same construction with \(\alpha_1^{-1},\alpha_2^{-1}\) gives its inverse.
For interactions \(\Phi\) and \(\Psi\) on the two layers, their graded tensor sum is defined termwise by
\[
 (\Phi\widehat\otimes\1+\1\widehat\otimes\Psi)(X)
 :=\Phi(X)\widehat\otimes\1+\1\widehat\otimes\Psi(X).
\]
If at least one of two states \(\rho_1,\rho_2\) is even, their graded product state is determined by
\begin{equation*}
 (\rho_1\widehat\otimes\rho_2)(A\widehat\otimes B)
 :=\rho_1(A)\rho_2(B)
\end{equation*}
on homogeneous simple tensors.
Its positivity and existence under this evenness hypothesis follow from the CAR joint-extension theorem \cite{ArakiMoriyaJoint03}.

We use the standard CAR facts about quasifree states, their GNS representations, Bogoliubov dynamics, and second quantization from \cite[Sections~5.2.3--5.2.4]{BR97}; see also \cite{Araki70}.

\begin{lemma}
\label{lem:metric-space-geometry}
For every nonempty finite \(X\subset\Gamma\), \(R\ge0\), and \(y\in\Gamma\),
\begin{align}
 |X^{(R)}|&\le C_\Gamma|X|(1+R)^{\dvol},
 &\diam X^{(R)}&\le\diam X+2R,
 \label{eq:metric-thickening}\\
 \1_{\{y\in X^{(R)}\}}
 &\le\sum_{x\in\ball_R(y)}\1_{\{x\in X\}}.
 \label{eq:metric-root-count}
\end{align}
In particular, \(X^{(R)}\) is finite.
For \(p>\dvol\),
\begin{equation}
 \sup_{x\in\Gamma}\sum_{y:\,d(x,y)>R}(1+d(x,y))^{-p}
 \le C_p(1+R)^{\dvol-p}.
 \label{eq:Fp-tail}
\end{equation}
\end{lemma}

\begin{proof}
The three geometric bounds follow by covering \(X^{(R)}\) with the balls \(\ball_R(x)\), \(x\in X\), and by the triangle inequality.
For \(R\ge1\), decomposition into dyadic annuli gives
\begin{equation*}
 \sum_{y:\,d(x,y)>R}(1+d(x,y))^{-p}
 \le C\sum_{j\ge0}
 (1+2^{j+1}R)^{\dvol}(1+2^jR)^{-p},
\end{equation*}
which proves \eqref{eq:Fp-tail}.
For \(0\le R<1\), the volume bound controls the sites in \(\ball_1(x)\), and the estimate at \(R=1\) controls the remaining sum.
\end{proof}

Polynomial volume growth allows additional powers of the support size to be absorbed by a small decrease in the exponential decay rate.

\begin{lemma}
\label{lem:support-moment-rate-trade}
Let \(\Phi\) be an interaction, and fix decay rates \(0<\nu<\mu\) and support-moment orders \(k,\ell\ge0\).
Then
\begin{equation}
 \|\Phi\|_{\nu,1;k+\ell}
 \le C_\Gamma^\ell\Bigl[\sup_{R\ge0}
 (1+R)^{\dvol\ell}\E^{-(\mu-\nu)R}\Bigr]
      \|\Phi\|_{\mu,1;k}.
 \label{eq:support-moment-rate-trade}
\end{equation}
For a stretched-exponential exponent \(0<\theta<1\) and a decay rate \(a>0\), one also has
\begin{equation}
 \|\Phi\|_{a,\theta;k+\ell}
 \le C_\Gamma^\ell\Bigl[\sup_{R\ge0}
 (1+R)^{\dvol\ell}
 \E^{aR^\theta-\mu R}\Bigr]
      \|\Phi\|_{\mu,1;k}.
 \label{eq:support-moment-stretched-trade}
\end{equation}
\end{lemma}

\begin{proof}
For nonempty \(X\) and \(x\in X\), \(X\subset\ball_{\diam X}(x)\), and hence
\begin{equation*}
 |X|\le C_\Gamma(1+\diam X)^{\dvol}.
\end{equation*}
Insert this estimate for the additional factor \(|X|^\ell\) in the interaction sum.
The functions \((1+R)^{\dvol\ell}\E^{-(\mu-\nu)R}\) and \((1+R)^{\dvol\ell}\E^{aR^\theta-\mu R}\) tend to zero as \(R\to\infty\), so both suprema are finite.
\end{proof}

\subsection{Weighted interaction topologies and non-autonomous dynamics}
\label{app:nonautonomous-background}

For a nonempty set \(X\Subset\Gamma\), define the rooted seminorm used for local convergence of interactions by
\begin{equation}
 \|\Phi\|^{\rm loc}_{\mu,\theta;2;X}
 :=\max_{x\in X}\sum_{Y\ni x}
 |Y|^2\E^{\mu\diam(Y)^\theta}\|\Phi(Y)\|.
 \label{eq:local-stretched-seminorm}
\end{equation}
For polynomial weights, set
\begin{equation}
 \|\Xi\|^{\pol,\loc}_{p;X}
 :=\max_{x\in X}\sum_{Y\ni x}
 (1+\diam Y)^p\|\Xi(Y)\|.
 \label{eq:polynomial-local-seminorm}
\end{equation}
The interaction norms in the main theorems control \(\|\Phi\|^{\pol}_{p;1}\): for \(p\ge0\), \(\mu,\kappa>0\), and \(0<\theta\le1\),
\begin{align*}
 \|\Phi\|^{\pol}_{p;1}
 \le\left(\sup_{r\ge0}(1+r)^p\E^{-\mu r^\theta}\right)
       \|\Phi\|_{\mu,\theta;2},\qquad
 \|\Phi\|^{\pol}_{p;1}
 \le\left(\sup_{m\ge1}m\E^{-\kappa m}\right)
       \|\Phi\|^{\#}_{p;\kappa}.
\end{align*}
Finite \(\|\Phi\|^{\pol}_{0;0}\)-norm makes the sum in \eqref{eq:local-derivation} absolutely convergent on local observables.
Proposition~\ref{prop:nonautonomous-local-limit-package} constructs the dynamics of componentwise norm-continuous paths with a uniform \(\|\cdot\|^{\pol}_{p;0}\)-bound for some \(p>\dvol\).
For a constant interaction with \(\|\Phi\|^{\pol}_{p;1}<\infty\) for some \(p>\dvol\), Lemma~\ref{lem:local-form-core} shows that its dynamics has generator \(\I\overline\delta_\Phi\), with the closure defined in Section~\ref{sec:setting}.
We use the convention
\begin{equation*}
 \frac{\dd}{\dd t}\tau_t^\Phi(A)
 =\I\tau_t^\Phi(\overline\delta_\Phi(A)),
 \qquad A\in\operatorname{Dom}(\overline\delta_\Phi).
\end{equation*}
The free interaction \(\Phi_h\) defined in the main setup generates the quasifree dynamics of \(h\).
The quadratic-interaction estimates \eqref{eq:quadratic-interaction-bounds} give the required first support moment in both locality regimes.

\begin{lemma}
\label{lem:componentwise-local-continuity}
Let \((\Xi_s)_{s\in J}\) be a componentwise norm-continuous interaction path with a uniform bound in \(\|\cdot\|_{\mu,\theta;2}\), where \(\mu>0\) and \(0<\theta\le1\).
Then, for every nonempty \(X\Subset\Gamma\), \(0<\nu<\mu\), and \(s_0\in J\),
\begin{equation}
 \lim_{s\to s_0}
 \|\Xi_s-\Xi_{s_0}\|^{\rm loc}_{\nu,\theta;2;X}=0.
 \label{eq:local-weighted-continuity}
\end{equation}
If instead the path has a uniform \(\|\cdot\|^{\pol}_{p;0}\)-bound, with \(p>0\), it is continuous in every local seminorm \(\|\cdot\|^{\pol,\loc}_{q;X}\), \(0\le q<p\).

Suppose that further paths \(\Xi_s^{(n)}\) have the same uniform bound and converge componentwise to \(\Xi_s\), uniformly in \(s\in J\).
Then they converge uniformly in \(s\) in each of the respective weaker local seminorms.
\end{lemma}

\begin{proof}
Let \(B\) be a common bound in the stronger interaction norm.
For a difference of two such interactions, the contribution of sets with diameter greater than \(R\) to the weaker local seminorm is at most
\[
 2B\E^{-(\mu-\nu)R^\theta}
 \quad\text{or}\quad
 2B(1+R)^{q-p},
\]
respectively.
These bounds tend to zero uniformly in the path parameters and approximation index.
Only finitely many sets of diameter at most \(R\) meet the finite root region \(X\), by Lemma~\ref{lem:metric-space-geometry}.
Componentwise continuity or uniform componentwise convergence therefore controls the remaining finite sum.
Taking first the parameter or approximation limit and then \(R\to\infty\) proves continuity and convergence in the weaker local seminorms.
\end{proof}

We use the following parameter-uniform form of \cite[Proposition~2.4]{FermionicLR}.

\begin{proposition}
\label{prop:nonautonomous-local-limit-package}
Let \(J\subset\R\) be a compact interval and \(\cP\) a nonempty parameter set.
For \(n\in\N\cup\{\infty\}\) and \(z\in\cP\), let \((\Xi_u^{(n,z)})_{u\in J}\) be componentwise norm-continuous paths of even self-adjoint interactions.
Suppose that, for some \(p>\dvol\) and \(0\le M<\infty\),
\[
 \sup_{\substack{n\in\N\cup\{\infty\}\\z\in\cP,\ u\in J}}
       \|\Xi_u^{(n,z)}\|^{\pol}_{p;0}\le M.
\]
Each path \(\Xi=\Xi^{(n,z)}\) generates a unique strongly continuous cocycle \((\alpha_{s,t}^\Xi)_{s,t\in J}\) of parity-preserving automorphisms satisfying, for \(A\in\A_{\loc}\),
\begin{equation}
 \alpha_{s,t}^\Xi(A)
 =A+\I\int_s^t\alpha_{s,u}^\Xi(\delta_{\Xi_u}(A))\,\dd u.
 \label{eq:nonautonomous-integral-equation}
\end{equation}
If, for every nonempty \(Z\Subset\Gamma\),
\begin{equation}
 \sup_{z\in\cP}\int_J
 \|\Xi_u^{(n,z)}(Z)-\Xi_u^{(\infty,z)}(Z)\|\,\dd u
 \stackrel{n\to\infty}{\longrightarrow}0,
 \label{eq:nonautonomous-local-L1-convergence}
\end{equation}
then, for every \(A\in\A\),
\begin{equation}
 \sup_{z\in\cP}\sup_{s,t\in J}
 \|\alpha_{s,t}^{\Xi^{(n,z)}}(A)
   -\alpha_{s,t}^{\Xi^{(\infty,z)}}(A)\|\stackrel{n\to\infty}{\longrightarrow}0.
 \label{eq:nonautonomous-cocycle-local-limit}
\end{equation}
\end{proposition}

\begin{proof}
Existence, uniqueness, and \eqref{eq:nonautonomous-integral-equation} follow from \cite[Proposition~2.4]{FermionicLR}.
The finite-region comparison in its proof has constants depending only on the common norm bound \(M\), \(p\), \(|J|\), and the geometry, so the same estimate holds after taking the supremum over \(z\in\cP\).
For \(A\in\A_X\), use the comparison with the finite region \(X^{(R)}\).
At fixed \(R\), \eqref{eq:nonautonomous-local-L1-convergence} makes the finite sum of interaction differences tend to zero uniformly in \(z\) as \(n\to\infty\), and then \(R\to\infty\) gives \eqref{eq:nonautonomous-cocycle-local-limit}.
Approximating \(A\in\A\) by local observables extends the convergence, since all cocycles are isometric.
\end{proof}

\subsection{Gap inequalities and GNS operator cores}
\label{app:GNS-background}

Under the locality hypothesis of Lemma~\ref{lem:local-form-core}, \eqref{eq:ground-form} agrees with the standard algebraic ground-state condition \cite[Definition~5.3.18 and Proposition~5.3.19]{BR97}.
Indeed, testing on \(\1+zA\) for \(z\in\C\) gives \(\omega(\delta_\Phi(A))=0\), and applying this identity to \(A^*A\) shows that \(\omega(A^*\delta_\Phi(A))\) is real.
The graph-core property extends the inequality to \(\operatorname{Dom}(\overline\delta_\Phi)\).
The cited characterization then shows that \(\omega\) is invariant under \(\tau^\Phi\) and has a positive GNS Hamiltonian.

\begin{proposition}
\label{prop:form-spectrum}
Let \(\Phi\) generate strongly continuous dynamics \(\tau^\Phi\) with generator \(\I\delta\).
Assume that \(\A_{\loc}\subseteq\operatorname{Dom}(\delta)\) and \(\delta=\delta_\Phi\) on \(\A_{\loc}\).
Let \(\omega\) be a ground state of \(\tau^\Phi\) in the sense of \cite[Definition~5.3.18]{BR97}.
Write \((\pi_\omega,\cH_\omega,\Omega_\omega)\) for its GNS triple and \(L_\omega\ge0\) for its GNS Hamiltonian, and suppose that \(\pi_\omega(\A_{\loc})\Omega_\omega\) is a core for \(L_\omega^{1/2}\).
Then, for every \(g'>0\), \eqref{eq:GNS-form-gap} is equivalent to the quadratic-form inequality
\begin{equation}
 L_\omega\ge
 g'\left(\1-|\Omega_\omega\rangle\langle\Omega_\omega|\right).
 \label{eq:GNS-operator-gap}
\end{equation}
When these equivalent conditions hold, \(\ker L_\omega=\C\Omega_\omega\) and \(\spec(L_\omega)\subset\{0\}\cup[g',\infty)\).
\end{proposition}

\begin{proof}
Differentiating the canonical GNS implementation gives
\begin{equation}
 L_\omega\pi_\omega(A)\Omega_\omega
 =\pi_\omega(\delta_\Phi(A))\Omega_\omega,
 \qquad A\in\A_{\loc}.
 \label{eq:GNS-generator-on-local-vectors}
\end{equation}
Thus \eqref{eq:GNS-form-gap} is \eqref{eq:GNS-operator-gap} restricted to the centered vectors \((\pi_\omega(A)-\omega(A))\Omega_\omega\).
By the form-core assumption and \(L_\omega^{1/2}\Omega_\omega=0\), the centered local vectors form a core for the form restricted to \(\Omega_\omega^\perp\).
Extending the inequality by continuity in the form norm gives \eqref{eq:GNS-operator-gap}, and restriction proves the converse.
The spectral theorem then gives a simple eigenvalue at zero and gap at least \(g'\) \cite[(2.20)--(2.22)]{NSY24}.
\end{proof}

If \(\|\Phi\|^{\pol}_{p;1}<\infty\) for some \(p>\dvol\), Lemma~\ref{lem:local-form-core} gives the generator and core hypotheses of Proposition~\ref{prop:form-spectrum} for every ground state of \(\Phi\).

\begin{lemma}
\label{lem:local-form-core}
Let \(\Phi\) be an even self-adjoint interaction with \(\|\Phi\|^{\pol}_{p;1}<\infty\) for some \(p>\dvol\), and let \(\I\delta\) be the generator of its dynamics.
Then \(\A_{\loc}\) is a graph core for \(\delta\), so \(\delta=\overline{\delta_\Phi|_{\A_{\loc}}}\).
If \(\omega\) is a ground state of \(\Phi\), then \(\pi_\omega(\A_{\loc})\Omega_\omega\) is an operator core, and hence a form core, for the GNS Hamiltonian \(L_\omega\).
\end{lemma}

\begin{proof}
For \(A\in\A_X\) and finite \(\Omega\supset X\), the local observables \(A_{\Omega,t}:=\tau_t^{\Phi_\Omega}(A)\) converge in norm to \(\tau_t^\Phi(A)\), uniformly on compact time intervals, by Proposition~\ref{prop:nonautonomous-local-limit-package}.
Applying \cite[Lemma~3.1(i)--(ii)]{FermionicLR}, with weight \(w(r)=(1+r)^p\), to \(\Phi_\Omega\), viewed as an interaction on the finite region \(\Omega\cup Y\) with all other terms zero, gives
\[
 \|[\Phi(Y),A_{\Omega,t}]\|
 \le2\E^{2T\|\Phi\|^{\pol}_{p;1}}\|A\|\|\Phi(Y)\|
       \sum_{x\in X}\sum_{y\in Y}(1+d(x,y))^{-p},
 \qquad |t|\le T,
\]
independently of \(\Omega\), and this bound is summable in \(Y\) because
\[
 \sum_Y\|\Phi(Y)\|\sum_{x\in X}\sum_{y\in Y}(1+d(x,y))^{-p}
 \le |X|\sup_x\sum_y(1+d(x,y))^{-p}\|\Phi\|^{\pol}_{0;0}<\infty.
\]
By \eqref{eq:nonautonomous-integral-equation}, every local observable \(B\) belongs to \(\operatorname{Dom}(\delta)\) and satisfies \(\delta(B)=\delta_\Phi(B)\).
Dominated convergence and closedness of \(\delta\) therefore give \(A_{\Omega,t}\to\tau_t^\Phi(A)\) in graph norm.
For \(f\in C_c^\infty(\R)\), the local observables \(\int f(t)A_{\Omega,t}\,\dd t\) converge in graph norm to \(\int f(t)\tau_t^\Phi(A)\,\dd t\).
The map \(B\mapsto\int f(t)\tau_t^\Phi(B)\,\mathrm dt\) is bounded from the algebra norm to the graph norm, with operator norm at most \(\|f\|_1+\|f'\|_1\).
Since local observables are norm-dense, these convolutions of any \(B\in\A\) can also be approximated by local observables in graph norm.
Applying a time approximate identity to \(B\in\operatorname{Dom}(\delta)\) then gives approximation of \(B\) by local observables in graph norm.

Equation~\eqref{eq:GNS-generator-on-local-vectors} extends to \(A\in\operatorname{Dom}(\delta)\) by graph approximation.
Since \(\I\delta\) generates an isometric group, \(\delta+\I\) has a bounded inverse.
The graph-core property therefore makes \((\delta+\I)\A_{\loc}\) norm-dense in \(\A\).
The GNS identity \eqref{eq:GNS-generator-on-local-vectors} and cyclicity show that \((L_\omega+\I)\pi_\omega(\A_{\loc})\Omega_\omega\) is dense in \(\cH_\omega\).
Applying \((L_\omega+\I)^{-1}\) gives graph-norm density in \(\operatorname{Dom}(L_\omega)\).
\end{proof}

\subsection{Weighted functional calculus for the one-particle operator}
\label{app:free-functional-calculus}

The proof of the following lemma is the weighted-conjugation argument of Combes and Thomas \cite{CombesThomas73}, carried out here to obtain a two-sided weighted Schur bound with constants controlled on a general polynomial-growth metric space.

\begin{lemma}
\label{lem:weighted-free-functional-calculus}
Under \eqref{eq:h-locality} and \eqref{eq:one-particle-gap}, there is \(\mu_1\in(0,\mu_h)\), depending only on the one-particle locality and gap data, such that
\begin{equation}
 \|\sgnh\|_{\mu_1,\theta_{\rm in}}^{\rm Sch}
 +g^{-1}\bigl\||h|\bigr\|_{\mu_1,\theta_{\rm in}}^{\rm Sch}
 \le C(C_\Gamma,\dvol,r,C_h/g,\mu_h,\theta_{\rm in}).
 \label{eq:free-functional-locality}
\end{equation}
\end{lemma}

\begin{proof}
First let \(\theta_{\rm in}=1\), put \(P:=\1_{(-\infty,0)}(h)\), and choose the positively oriented contour
\[
 \cC:=\partial\bigl([-C_h-g,-g/2]+\I[-g,g]\bigr),
 \qquad
 \dist_{\cC}:=\inf_{z\in\cC}\dist(z,\spec(h))\ge g/2.
\]
For \(0<\mu<\mu_h\) and \(L\in\N\), the bounded multiplication weights \(D_{x,\mu,L}(y):=\E^{\mu\min\{d(x,y),L\}}\) satisfy the Combes--Thomas conjugation bound \cite{CombesThomas73}
\begin{equation}
 \sup_{x\in\Gamma,\,L\in\N}\|D_{x,\mu,L}^{-1}hD_{x,\mu,L}-h\|
 \le\varepsilon_h(\mu),
 \qquad
 \varepsilon_h(\mu):=C_h\sup_{\rho\ge0}
       (\E^{\mu\rho}-1)\E^{-\mu_h\rho}.
 \label{eq:capped-weight-error}
\end{equation}
Choose \(0<\mu_1<\mu_h/2\) so that \(\varepsilon_h(2\mu_1)<\dist_{\cC}/2\), which is possible since \(\varepsilon_h(\mu)\stackrel{\mu\searrow0}{\longrightarrow}0\).
A Neumann series for the conjugated resolvent gives
\[
 \|D_{x,2\mu_1,L}^{-1}(z-h)^{-1}D_{x,2\mu_1,L}\|
 \le 2/\dist_{\cC},
 \qquad z\in\cC.
\]
Letting \(L\to\infty\) gives \(\|(z-h)^{-1}_{xy}\|\le2\dist_{\cC}^{-1}\E^{-2\mu_1d(x,y)}\).
Multiplying by \(\E^{\mu_1d(x,y)}\) and summing over \(x\) or \(y\) gives
\[
 \sup_{z\in\cC}\|(z-h)^{-1}\|_{\mu_1}^{\rm Sch}
 \le\frac{2}{\dist_{\cC}}\sup_x\sum_y\E^{-\mu_1d(x,y)}<\infty.
\]
The contour formula for \(P\), together with \(\sgnh=\1-2P\), \(|h|=h\sgnh\), and submultiplicativity of the weighted Schur norm, proves \eqref{eq:free-functional-locality}.
The choice of \(\mu_1\) and the ratio \(\operatorname{length}(\cC)/\dist_{\cC}\) depend only on \(C_h/g\) and \(\mu_h\).
Volume growth bounds \(\sup_x\sum_y\E^{-\mu_1d(x,y)}\) in terms of \(C_\Gamma,\dvol,\mu_1\).
For \(0<\theta_{\rm in}<1\), apply the argument in the metric \(\rho=d^{\theta_{\rm in}}\), which has polynomial volume growth with exponent \(\dvol/\theta_{\rm in}\).
The exponential Schur norms for \(\rho\) equal the \(\theta_{\rm in}\)-weighted Schur norms for \(d\).
\end{proof}

Polynomial weights also allow functional calculus without loss of the spatial exponent.
We use the modified differential-norm method of Sun \cite{Sun07}, as discussed in \cite[Section~5]{GrochenigKlotz13}.
We prove \eqref{eq:polynomial-Schur-square} for operators on \(\Hsp_{\mathrm M}\), then apply it to the sign-function series.

\begin{lemma}
\label{lem:polynomial-free-functional-calculus}
Assume \eqref{eq:volume-growth} with \(\dvol>0\), and let \(p>0\).
Let \(h=h^*\) on \(\Hsp_{\mathrm M}\) satisfy \eqref{eq:one-particle-gap} and \(C_h:=\|h\|^{\mathrm{Sch,pol}}_p<\infty\).
Then
\begin{equation}
 \|\operatorname{sgn}(h)\|^{\mathrm{Sch,pol}}_p
 +g^{-1}\bigl\||h|\bigr\|^{\mathrm{Sch,pol}}_p
 \le C(p,\dvol,C_\Gamma,r,C_h/g).
 \label{eq:polynomial-free-functional-locality}
\end{equation}
\end{lemma}

\begin{proof}
For \(q\in\{0,p\}\), let \(\cS_q\) be the space of bounded operators \(k\) on \(\Hsp_{\mathrm M}\) with finite weighted Schur norm \(\|k\|^{\mathrm{Sch,pol}}_q\), and write \(\|k\|_{\cS_q}:=\|k\|^{\mathrm{Sch,pol}}_q\).
The space \(\cS_q\) is complete.
The triangle inequality for the metric gives
\begin{align}
 \|k\ell\|_{\cS_p}
 &\le\|k\|_{\cS_p}\|\ell\|_{\cS_p},
 \label{eq:polynomial-Schur-product}\\
 \|k\ell\|_{\cS_p}
 &\le 2^{\max\{p-1,0\}}\bigl(\|k\|_{\cS_p}\|\ell\|_{\cS_0}
             +\|k\|_{\cS_0}\|\ell\|_{\cS_p}\bigr).
 \label{eq:polynomial-Schur-differential-product}
\end{align}
Indeed, \(1+d(x,y)\le(1+d(x,z))(1+d(z,y))\) proves the first bound.
The second follows from \((1+d(x,y))^p\le 2^{\max\{p-1,0\}}((1+d(x,z))^p+(1+d(z,y))^p)\) by summing the matrix product in both Schur directions.

The finite block dimension gives
\[
 \sum_y\|k_{xy}\|^2
 \le\operatorname{Tr}_{\C^{2r}}((kk^*)_{xx})\le2r\|k\|^2\,,\qquad
 \sum_x\|k_{xy}\|^2
 \le\operatorname{Tr}_{\C^{2r}}((k^*k)_{yy})\le2r\|k\|^2.
\]
Split each row and column at distance \(R\ge1\).
Cauchy--Schwarz bounds the row sum over \(d(x,y)\le R\) by \(|\ball_R(x)|^{1/2}(\sum_y\|k_{xy}\|^2)^{1/2}\), and the remaining sum is at most \((1+R)^{-p}\|k\|_{\cS_p}\).
Applying the column estimate in the same way and using volume growth gives
\begin{equation}
 \|k\|_{\cS_0}
 \le C R^{\dvol/2}\|k\|
       +R^{-p}\|k\|_{\cS_p}.
 \label{eq:polynomial-Schur-interpolation-split}
\end{equation}
For \(k\ne0\), choose \(R:=(\|k\|_{\cS_p}/\|k\|)^{1/(p+\dvol/2)}\).
By the Schur test, \(R\ge1\), so
\[
 \|k\|_{\cS_0}
 \le C\|k\|_{\cS_p}^{\vartheta}\|k\|^{1-\vartheta},
 \qquad \vartheta:=\frac{\dvol}{2p+\dvol}\in(0,1).
\]
Combining this with \eqref{eq:polynomial-Schur-differential-product} gives, for a constant \(D_p\ge1\) depending only on \(p,\dvol,C_\Gamma,r\),
\begin{equation}
 \|k^2\|_{\cS_p}
 \le D_p\|k\|_{\cS_p}^{1+\vartheta}\|k\|^{1-\vartheta}.
 \label{eq:polynomial-Schur-square}
\end{equation}

Set \(\mathsf b:=\1-h^2/(4C_h^2)\) and \(\mathsf q:=1-g^2/(4C_h^2)\).
Since \(g\le\|h\|\le C_h\),
\[
 \tfrac34\1\le \mathsf b\le \mathsf q\1,\qquad
 \tfrac34\le \mathsf q<1,\qquad \|\mathsf b\|_{\cS_p}\le\tfrac54.
\]
Using \(\|\mathsf b^{2^j}\|\le \mathsf q^{2^j}\) in \eqref{eq:polynomial-Schur-square} gives the recurrence
\[
 \mathsf q^{-2^{j+1}}\|\mathsf b^{2^{j+1}}\|_{\cS_p}
 \le D_p\left(\mathsf q^{-2^j}\|\mathsf b^{2^j}\|_{\cS_p}\right)^{1+\vartheta}.
\]
Induction therefore yields
\begin{equation}
 \mathsf q^{-2^j}\|\mathsf b^{2^j}\|_{\cS_p}
 \le D_p^{((1+\vartheta)^j-1)/\vartheta}
       \left(\frac{\|\mathsf b\|_{\cS_p}}{\mathsf q}\right)^{(1+\vartheta)^j}
 \le\exp\bigl(C(1+\vartheta)^j\bigr).
 \label{eq:polynomial-Schur-dyadic-powers}
\end{equation}
The constant in the last bound depends only on \(p,\dvol,C_\Gamma,r\).
For \(n\ge1\), write \(n=\sum_{j=0}^J\varepsilon_j2^j\), with \(\varepsilon_j\in\{0,1\}\) and \(2^J\le n\).
Submultiplicativity and \(\sum_{j=0}^J(1+\vartheta)^j\le (1+\vartheta)n^{\log_2(1+\vartheta)}/\vartheta\) give
\begin{equation}
 \|\mathsf b^n\|_{\cS_p}
 \le \mathsf q^n\exp(Cn^{\log_2(1+\vartheta)}).
 \label{eq:polynomial-Schur-powers}
\end{equation}
The series
\begin{equation}
 \frac h{2C_h}\sum_{n=0}^\infty
       \frac{\binom{2n}{n}}{4^n}\mathsf b^n
 \label{eq:polynomial-sign-series}
\end{equation}
converges absolutely in \(\cS_p\): its scalar coefficients are at most one, and \(\log_2(1+\vartheta)<1\) gives \(\sum_n \mathsf q^n\exp(Cn^{\log_2(1+\vartheta)})<\infty\).
In operator norm, the binomial series equals \((h/(2C_h))(\1-\mathsf b)^{-1/2}=\operatorname{sgn}(h)\).
Its weighted norm is bounded in terms of the data in \eqref{eq:polynomial-free-functional-locality}, since \(\mathsf q=1-(g/(2C_h))^2\).
Finally, \(|h|=h\operatorname{sgn}(h)\) and \eqref{eq:polynomial-Schur-product} give \eqref{eq:polynomial-free-functional-locality}.
\end{proof}

\section*{Acknowledgments}

The author thanks Lennart Becker, Marius Lemm, Marius Wesle, and Tom Wessel for many helpful remarks and discussions.
He thanks Domenico Monaco and Gianluca Panati at Sapienza University of Rome and Massimo Moscolari at Politecnico di Milano for their hospitality during his two-month sabbatical visits to their respective institutions.
This work was supported by the Deutsche Forschungsgemeinschaft (DFG, German Research Foundation) -- Project numbers 470903074 and 465199066.

\section*{Declaration on the use of generative AI}

During the preparation of this manuscript, the author used OpenAI's ChatGPT Pro (GPT-5.6 Sol) to assist in exploring and critically examining mathematical arguments and literature.
The final version was checked for consistency and polished using GPT-6 Astra.
The author reviewed and revised all AI-assisted material.
The author takes full responsibility for the content of the manuscript.

\end{document}